\documentclass[11pt]{article}

\usepackage[english]{babel}

\usepackage[margin=1in]{geometry}
\usepackage{tikz}

\usepackage{amsmath,amssymb,amsthm}
\usepackage{mathtools}
\usepackage{graphicx}
\usepackage{paralist}
\usepackage[colorlinks=true, allcolors=black!25!blue]{hyperref}
\usepackage{enumitem}

\usepackage{thmtools}
\usepackage{thm-restate}

\newcommand{\vprofile}{\mathbf{m}}
\newcommand{\profile}{m}
\newcommand{\vjobsleft}{{\boldsymbol \nu}}
\newcommand{\jobsleft}{\nu}

\newcommand{\timedtypes}[1]{\mathcal{J}(#1)}
\newcommand{\ve}{\mathbf{e}}
\newcommand{\vzero}{{\boldsymbol 0}}

\newcommand{\schedule}{\mbox{${\mathcal{S}}$}}

\newcommand{\Prob}[1]{q_{#1}} 
\newcommand{\ecost}[1]{\textup{cost}(#1)}

\newcommand{\eps}{\varepsilon}

\newcommand{\N}{\mathbb{N}_0}
\newcommand{\kk}{h}

\newcommand{\cost}{\mathrm{Cost}}
\newcommand{\numberjobs}{N}
\newcommand{\numtypes}{n}
\newcommand{\types}{\{1,\dots,\numtypes\}}
\newcommand{\numbergroups}{\gamma}
\newcommand{\vnumberofjobspertype}{\mathbf{\numberjobs}}
\newcommand{\maxgroupsize}{z}
\newcommand{\policy}{\Pi}

\newcommand{\Cplus}[2]{C^+_{#1}(#2)}
\newcommand{\Splus}[2]{S^+_{#1}(#2)}
\newcommand{\Cstar}[2]{C^\star_{#1}(#2)}
\newcommand{\pmax}[1]{p^\text{max}_{#1}}

\def\Exp{{\mathbb{E}}}
\newcommand{\E}[1]{\mbox{$\Exp[\,#1\,]$}}

\def\Var{{\mathbb{V}\rm ar}}

\newcommand{\bigO}[1]{\mbox{$\textup{O}(\,#1\,)$}}
\newcommand{\bigOtilde}[1]{\mbox{$\textup{\~O}(\,#1\,)$}}
\newcommand{\eopt}{\textup{OPT}}
\newcommand{\BigO}[1]{\mbox{$\textup{O}\left(#1\right)$}}

\DeclareMathOperator*{\argmin}{argmin}

\newtheorem{theorem}{Theorem}[section]
\newtheorem{definition}[theorem]{Definition}
\newtheorem{lemma}[theorem]{Lemma}
\newtheorem{fact}[theorem]{Fact}

\newtheorem{corollary}[theorem]{Corollary}

\DeclarePairedDelimiterX\set[1]\lbrace\rbrace{\def\given{\;\delimsize\vert\;}#1}

\newcommand{\zerolength}[3]{
	\foreach \y in {#2}{
		\filldraw[thick,fill=#3] (#1,0.75+\y*0.1) circle (0.05cm);
	}
}
\newcommand{\shortjobs}[1]{
	\foreach \x in {#1}{
		\draw[thick,fill=red1] (\x,0) rectangle (\x+0.1,0.7);
	}
}
\newcommand{\longjob}[1]{
	\draw[thick,fill=blue1] (#1,0) rectangle (#1+8,0.7);
}
\colorlet{blue1}{blue!80!black!40}
\colorlet{red1}{red!80!black!40}

\title{Stochastic scheduling with Bernoulli-type jobs\\ through policy stratification}
\author{Antonios Antoniadis \and Ruben Hoeksma \and Kevin Schewior \and Marc Uetz}
\date{}

\begin{document}

\maketitle

\begin{abstract}
    This paper addresses the problem of computing a scheduling policy that minimizes the total expected completion time of a set of $N$ jobs with stochastic processing times on $m$ parallel identical machines. When all processing times follow Bernoulli-type distributions, Gupta et al.\ (SODA '23) exhibited approximation algorithms with an approximation guarantee  $\tilde{\textup{O}}(\sqrt{m})$, where $m$ is the number of machines and $\tilde{\textup{O}}(\cdot)$ suppresses polylogarithmic factors in~$N$, improving upon an earlier ${\textup{O}}(m)$ approximation by Eberle et al.\ (OR Letters '19) for a special case.  The present paper shows that, quite unexpectedly, the problem with Bernoulli-type jobs admits a PTAS whenever the number of different job-size parameters is bounded by a constant. The result is based on a series of transformations of an optimal scheduling policy to a ``stratified'' policy that makes scheduling decisions at specific points in time only, while losing only a negligible factor in expected cost. An optimal stratified policy is computed using dynamic programming. Two technical issues are solved, namely (i) to ensure that, with at most a slight delay, the stratified policy has an information advantage over the optimal policy, allowing it to simulate its decisions, and (ii) to ensure that the delays do not accumulate, thus solving the trade-off between the complexity of the scheduling policy and its expected cost. Our results also imply a quasi-polynomial $\bigO{\log N}$-approximation for the case with an arbitrary number of job sizes.
\end{abstract}

\thispagestyle{empty}

\vspace{\fill}
\pagenumbering{roman}

\newpage

\tableofcontents

\newpage

\pagenumbering{arabic}

\section{Introduction, Related Work, and Results}
This paper treats scheduling $N$ independent non-preemptive jobs on $m$ parallel identical machines to minimize the total completion time of the jobs, which is a basic and classic scheduling problem. It asks to assign jobs to machines and process the jobs on each machine sequentially and without job interruption, so as to minimize $\sum_j c_j$, where $c_j$ is the completion time of job~$j$. It is well known to be solvable in polynomial time by scheduling the jobs greedily in the shortest processing time first order \cite{BCS74}. The situation gets more realistic, but also more intricate, when the processing times are not known with certainty before starting a job. Clearly, without any knowledge of the processing times, no algorithm can achieve optimal, or even near-optimal schedules.
A scheduling model that interpolates between the two extremes no knowledge and full knowledge, is stochastic scheduling~\cite{MRW84,MRW85,Pinedo}, which is the model considered in this paper. Before explaining the actual contribution of this paper, we next sketch the context and precise setting that is being considered.

\paragraph{Stochastic scheduling.} Instead of knowing the processing time of a job~$j$ exactly, one assumes it is governed by a random variable $X_j$ with known probability distribution. The actual realization of a processing time of a job becomes known only when a job is finished. As usual, we assume independence of $X_j$ across the set of jobs.  Due to the stochastic input of the scheduling problem, its solution is no longer a schedule, but an object known as non-anticipatory \emph{scheduling policy}~\cite{MRW84}. It can be defined as a mapping of states to scheduling decisions. A state, at given time $t$, is a ``snapshot'' of the schedule at time $t$, that is the jobs still in process at time $t$, their remaining processing time distributions, as well as the processing time distributions of the jobs yet unprocessed. A scheduling decision consists of starting an unprocessed job on one of the idle machines, and updating $t$ (if necessary). Note that by independence across jobs, the realizations of processing times of jobs that are finished by time $t$ do not have to be included in the state at time $t$. Being non-anticipatory, the scheduling decisions of a policy may only use information that is available at time~$t$, and in particular it
must not depend on the realizations of jobs that have not yet been started by time~$t$.
As the resulting schedule is random, it is natural to minimize its \emph{expected total completion time} $\E{\sum_j C_j}= \sum_j\E{C_j}$, where $C_j$ is the random variable for $j$'s completion time.

\paragraph{Optimal policies and lower bounds.}
While the problem with known processing times can be solved in \bigO{n\log n} time, the complexity status of the stochastic counterpart is not yet well understood. For the model considered in this paper, at least the existence of an optimal scheduling policy follows easily because the problem can be modeled as a finite Markov decision process. A more general result in \cite[Theorem~4.2.6]{MRW84} yields the same.
However generally speaking, and neglecting details about the precise input encoding of the random variables $X_j$, to date there is little evidence showing that stochastic scheduling problems are significantly harder than the special case where the processing times are deterministic. One notable exception is a PSPACE-hardness result for ``stochastic scheduling''~\cite[Thm.\ 19.7]{PapadComplexityBook1994}, however the underlying reduction from stochastic SAT works only for scheduling problems with precedence-constrained jobs, and does not yield lower bounds on the computational complexity of scheduling stochastic jobs that are not precedence related.

There exist some positive results for specific assumptions on the processing time distributions. For example, scheduling jobs with shortest expected processing time first (dubbed SEPT) is optimal for exponential distributions~\cite{BDF81,WP80} and more generally when all distributions $X_j$ can be stochastically ordered~\cite{WVW86}. However, in general the expected performance of SEPT can be a $\Omega{(\,\Delta\,)}$ factor 
larger than that of an optimal policy~\cite{IMP2015}, where $\Delta$ is the largest squared coefficient of variation of the underlying processing time distributions, i.e., $\Delta := \max_j \Var[\,X_j\,]/\E{X_j}^2$.
In fact, beyond the above mentioned cases, little is known about how to compute optimal scheduling policies in polynomial time. Note that the above mentioned positive results show optimality of a very simple static list scheduling policy that determines a fixed ordering of the jobs a priori, and this list of jobs is purely based on the given input data $X_j$, $j=1,\dots,N$. Policies that determine such a list by computing some priority index per job as a function $f(X_j)$ of that job alone, and independent of the (number of) other jobs,
are also called \emph{index policies}. It is interesting to note that in general, no index policy can perform well for the problem considered in this paper. Indeed, it has been shown in \cite{EFMM2019} that there exist instances with only two different types of jobs with processing times that follow Bernoulli-type distributions so that the expected cost of any index policy exceeds the expected cost of an optimal scheduling policy by a factor at least $\Omega(\,\Delta^{1/4}\,)$. 
A comparable lower bound also exists for any ``fixed assignment'' scheduling policy that distributes the set of jobs over the set of machines at time 0. Indeed, there exist simple instances with only one type of stochastic jobs such that the expected cost of any such policy exceeds that of an optimal scheduling policy by a factor $\Omega(\,\Delta\,)$ \cite{SSU2016}. This shows that an optimal or close-to-optimal scheduling policy has to be adaptive to the realizations of jobs, and testifies to a large adaptivity gap.

\paragraph{Approximation and earlier work.}
The $\Omega(\,\Delta\,)$ lower bounds for the expected performance of certain list or index scheduling policies match with  $\bigO{\Delta}$ upper bounds on the expected performance of approximation algorithms that have been obtained during the past two decades. Following~\cite{MSU99}, for stochastic scheduling one defines an $\alpha$-approximation algorithm as a policy that achieves an expected cost at most $\alpha$ times the expected cost of an optimal scheduling policy, so if $\Pi$ is a scheduling policy and $\Pi^*$ is an (unknown) optimal scheduling policy, we require
\[
    \sum\nolimits_j \E{C_j^{\Pi}}\le \alpha\, \sum\nolimits_j  \E{C_j^{\Pi^{*}}}\,.
\]
Note that this is a comparison to an adversary that is also subject to the same uncertainty about actual processing times, and not the optimal clairvoyant offline solution.
There exist positive results also for more general models than the one considered in this paper, such as jobs with weights and/or release dates~\cite{MSU99}, precedence-constrained jobs~\cite{SU2005}, unrelated machines~\cite{SSU2016,BBC2016}, and also jobs that may appear online over time~\cite{MUV2006,S2008,GMUX2020,GMUX2021,DissJaeger2021,J2023}. All these works give $\bigO{\Delta}$-approximation algorithms, which means that the results are good, e.g., for exponentially distributed processing times where $\Delta=1$, but weak for stochastic scheduling problems with jobs that may have heavy-tailed distributions, hence they are weak also for Bernoulli-type distributions.

There are three notable exceptions that give approximation algorithms with performance guarantees that are independent of the actual processing time distributions. In~\cite{IMP2015} an algorithm is proposed with performance guarantee $\bigO{\log^2 N + m \log N}$ by defining a clever balance between scheduling jobs with small expected processing time and jobs with high probability of having a very large processing time. Further, there exists a simple list scheduling algorithm with performance guarantee $\bigO{m}$ for the special case when there are two types of jobs only~\cite{EFMM2019}, namely deterministic jobs with identical processing times and i.i.d.\ stochastic jobs that follow a Bernoulli-type distribution. 
The algorithm either schedules all deterministic jobs first, or vice versa, depending on the number of jobs of each type. Note that, because of this dependence on the number of jobs, it is not an index policy and hence does not contradict the $\Omega(\,\Delta^{1/4}\,)$ lower bound for index policies from the same paper. Recently, in~\cite{GMZ2023}, an algorithm was proposed that achieves a performance guarantee that is even sublinear in $m$, namely $\bigOtilde{\sqrt{m}}$, where $\bigOtilde{\cdot}$ suppresses  $\textup{poly}(\log N)$ factors. That result holds for scheduling jobs that all follow Bernoulli-type distributions. It is based on the idea to first express the objective in terms of $\log N$ ``free times'', which are the points in time when a policy starts to schedule the next of $\log N$ batches of jobs. The approximation algorithm decides how these batches should be defined so as to not differ too much from what the optimal policy does, or rather has to do to be optimal, and schedules jobs within each batch by an index policy.

\paragraph{Overview of model and main results.}

This paper addresses stochastic scheduling of jobs with Bernoulli-type processing times as in \cite{EFMM2019,GMZ2023}. That is, a set of $N$ jobs is to be scheduled on $m$ parallel and identical machines from time $0$ on.  Bernoulli-type processing times  means that the processing time is not known before starting a job, and a job $j$ has a random processing time $X_j$ that either realizes to  $x_j=p_j>0$ or to
$x_j=0$, with corresponding discrete probabilities  $q_j$ and $1-q_j$, respectively. For brevity, we refer to such jobs as Bernoulli jobs. We refer to parameter $p_j$ as a job's (non-zero) \emph{size}. The probabilities may differ per job, and in particular, the model includes a mix of deterministic and stochastic jobs.

Admittedly, this is a restricted type of distributions, yet is not uncommon also in other contexts of analyzing random service times, e.g.~\cite{RBB2019}. 
More importantly,  Bernoulli-type distributions encompass the core difficulty in scheduling systems with stochastic processing times, namely the uncertainty for how long the next job to be scheduled will block the machine. Once started, however, the realization of the random processing time $X_j$ is known with certainty. In that sense, Bernoulli jobs can be seen as a ``Vanilla version'' of stochastic scheduling. The stochastic nature of processing times generally requires optimal scheduling policies to be adaptive, also for Bernoulli jobs~\cite{Uetz03,SSU2016,EFMM2019}, which is one of the core problems that rule out many of the known techniques to obtain (close to) optimal scheduling policies.

Despite being a restricted class of distributions, even the problem with Bernoulli jobs has resisted all previous attempts in getting constant-factor approximation algorithms. As discussed earlier, this is true for all previously proposed approximation algorithms for stochastic scheduling, both for general distributions and for Bernoulli jobs. They either fail on inputs where the processing times exhibit a large variance, because the lower bounds used in the analysis are an $\Omega(\Delta)$ factor away from the optimum \cite{MSU99,SU2005,MUV2006,SSU2016,GMUX2021,J2023}, with $\Delta$ being an upper bound on the squared coefficient of variation of the processing times.  
Or they fail because the analysis uses volume-related arguments in comparison to the optimal policy~\cite{IMP2015,EFMM2019}, or a more sophisticated adaptation thereof \cite{GMZ2023}. This yields performance guarantees that depend on the number of machines. 

So despite many efforts, there is no algorithm known for stochastic scheduling that achieves a constant-factor approximation, not even for problems with only \emph{two} different types of jobs with Bernoulli-type distributions. (Unless one makes additional assumptions on the number of these jobs~\cite{EFMM2019}.)
Our main results are therefore a major leap forward in our understanding of the computational complexity of stochastic scheduling problems.
\begin{restatable}{theorem}{mainthmboundedn}
    There exists a PTAS for stochastic parallel machine scheduling of Bernoulli jobs to minimize the total expected completion times, given that there is a constant upper bound on the number of different size parameters $p_j$.   
\end{restatable}
That means that we show how to compute a (non-anticipatory) scheduling policy in polynomial time that achieves an expected cost at most a factor $(1+\eps)$ more than that of an optimal non-anticipatory scheduling policy, for any $\eps>0$.
Moreover, adapting some of the techniques of \cite{GMZ2023}, our main result can also be shown to yield the following.
\begin{restatable}{theorem}{mainthmunboundedn}
    There exists a quasi-polynomial time algorithm to compute a scheduling policy with performance guarantee \bigO{\log N} for stochastic parallel machine scheduling of Bernoulli jobs to minimize the total expected completion times.   
\end{restatable}
With the caveat that our result is only quasi-polynomial time, this latter performance bound substantially improves the earlier $\bigOtilde{\sqrt{m}}$ bound of \cite{GMZ2023}, as the latter also includes $\log N$ terms, suppressed by $\bigOtilde{\cdot}$.

\paragraph{Overview of methodological contribution.}
Our methodological contribution is conceptually simple, but turns out to be technically demanding. First, it is more or less folklore, due to the discrete character of the scheduling problem, that one can formulate the problem to compute an optimal non-anticipatory scheduling policy $\Pi$ by dynamic programming. This is clearly not polynomial time in general. The basic idea is to restrict the relevant variables in the dynamic program in such a way that this approach becomes polynomial time. In a nutshell, the requirements to get there, is (i) to restrict the number of different size parameters $p_j$ to be at most constant, and (ii) to make sure that the number of different machine-load profiles one needs to consider is polynomially bounded. Note that the first assumption is not yet sufficient to yield the second.
Therefore, we show that an optimal non-anticipatory scheduling policy $\Pi$ can be transformed to a ``stratified'' policy $\Pi'$ with expected cost at most $(1+\eps)$ more. That stratified policy $\Pi'$ starts jobs only at a set of predefined time points, and possibly with a small delay when compared to $\Pi$, so that we can bound the number of different machine-load profiles. This still sounds rather simple, yet the technical challenge lies in the fact that ``small delays'' need to be defined relative to the sizes of the different jobs, hence can differ by a large (i.e.\ exponential) amount. That has the effect that we cannot simply ``simulate $\Pi$ with some small delays here and there'', but in defining $\Pi'$, we fundamentally depart from the order in which the jobs are scheduled. That said, we still have to make sure that the stratified policy $\Pi'$, when scheduling any one job, has an information advantage over $\Pi$, meaning that all information about job realizations that $\Pi$ potentially used to make its scheduling decision, is also available for $\Pi'$ at the corresponding time. Since too much delay can accumulate and thus exceed the desired bound on the expected cost of $\Pi'$,  we end up with a rather complex mechanism to ensure the right balance between the information advantage of $\Pi'$ over $\Pi$ on the one hand, and the effect of the total delays on the other hand.
The result is what we refer to as a structure theorem, namely a transformation of an optimal non-anticipatory policy $\Pi$ to another, stratified non-anticipatory policy $\Pi'$ with expected cost not much more. Finally, an optimal stratified non-anticipatory policy, which performs by definition at least as good as $\Pi'$, can be computed in polynomial time by using a carefully crafted dynamic programming algorithm.

\paragraph{Organization.} The paper starts with fixing notation and collecting some helpful insights on optimal scheduling policies in Section~\ref{sec:notation}. 
To convey the basic idea of how to transform an optimal policy to a stratified policy, we present the corresponding structure theorem
for the  simpler special case of two Bernoulli job types in Section~\ref{sec:warmup_two_types}. Section~\ref{sec:multi-types} does the same for a constant number of types, which is technically more involved because of several, nested levels of time points to define the stratified policy $\Pi'$. Section~\ref{sec:groups} discusses how to undo a simplifying assumption that we make in both Sections~\ref{sec:warmup_two_types} and~\ref{sec:multi-types}, namely that the size parameters $p_j$ are sufficiently separated. 
The actual algorithmic solution for both cases is the dynamic program presented in Section~\ref{sec:dp}. Finally, we explain how our results yield a quasi-polynomial-time \bigO{\log N} approximation in Section~\ref{sec:unboundedtypes}. To improve readability of the paper, some of the technical proofs are omitted in the main text. These can be found in Appendix~\ref{app:omittedproofs}.

\section{Notation and Preliminaries}\label{sec:notation}
From this point on, in contrast to the usual indexing of jobs by index~$j$, for readability 
reasons we use $j$ to index job ``types'' instead, and therefore use~$J$ to index the jobs. 
Hence each job $J\in\{1,\dots,N\}$ has a random \emph{processing time} $X_J$ which follows 
a Bernoulli-type distribution with \emph{size} parameter $p_J>0$ and probability $0< q_J\le 1$, which means
\[
    X_J=\begin{cases}
        p_J & \text{with probability}\ q_J\,,\\
        0 & \text{with probability}\ 1-q_J\,.
    \end{cases}
\]
We refer to these realizations as \emph{long} and \emph{short}, respectively.

We mainly consider the restricted variant in which the $N$ jobs are not all different, but of a smaller number  of different job \emph{types} $j\in\types$, such that $p_J=p_{J'}$ for any two jobs $J,J'$ of the same type~$j$. We also denote the size parameter of all type-$j$ jobs by $p_j$ (and it will be clear from the context whether the subscript is an index of a job or a type). W.l.o.g.\ we may then assume that different job types have different sizes. Let $N_j$ denote the number of jobs of type $j$, and $\vnumberofjobspertype=(N_1,\dots,N_n)$ be the corresponding vector, so that $N=\sum_j N_j$ is the number of jobs. We do \emph{not} require that $q_J=q_{J'}$ for two jobs $J,J'$ of the same type~$j$. Specifically, we consider the case with only two different job types as a warm-up case. We always assume that jobs or job types are ordered such that $p_1\ge \dots\ge p_N$ and $p_1 > \dots > p_n$, respectively. 
We denote by $m$ the number of machines. Our goal is to compute, in polynomial time, a scheduling policy $\policy$ to schedule the jobs non-preemptively on $m$ parallel and identical machines that has expected cost 
\[
    \ecost{\policy}:=\sum\nolimits_J \E{C^{\policy}_J}
\]
at most $(1+\eps)$ times the expected cost of an optimal policy $\Pi^*$, for any $\eps>0$. Here, $C^{\policy}_J$ is the (random) completion time of job $J$ under policy $\policy$.

A \emph{scheduling policy} $\policy$ maps partial schedules together with a decision time~$t$ to a scheduling decision at time~$t$. It is \emph{non-anticipatory}, which means that the scheduling decision must not be based on realizations of processing times of jobs that have not yet been started at time $t$. The policy may base its decisions on the distribution of the remaining processing times of all jobs that have been started at of before $t$, however. For jobs with Bernoulli distributions, that means that once a job $J$ is started at time~$t$, its processing time gets known at time~$t$.  

Given that the processing times are random, a policy $\policy$ yields a \emph{stochastic schedule} $\schedule(\policy)$, which is a distribution over the set of $m$-machine schedules. A stochastic schedule is given by (typically correlated) stochastic \emph{start times} $S_1,\dots,S_N$ next to machine assignment for each job. Note that $C_J=S_J+X_J$. As usual, a statement like the preceding one about random variables should be understood as being true for all realizations, and
the \emph{feasibility} of a stochastic schedule refers to the fact that each machine processes at most one job at any time. On several occasions, we first define a stochastic schedule by defining a stochastic start time and a machine assignment for each job $J$, and subsequently argue that there exists a non-anticipatory policy that attains it.

\paragraph{Properties of optimal policies.} It is generally not known under which conditions an optimal scheduling policy may be assumed to be \emph{non-idling}, i.e., no machine is left idle unless all jobs have already been started. For the setting with job weights $w_j$ and total expected weighted completion time objective, $\sum_J w_J \E{C_J}$, simple examples show that deliberate idling of machines may be necessary to gather information and make optimal scheduling decisions later \cite{Uetz03}.  However, this information gain does not exist in the setting with Bernoulli jobs, because once a job $j$ is started, its realized processing time gets known to be $0$ or $p_J$, instantaneously. This implies the following lemma, which is crucial for our approach.

\begin{restatable}{lemma}{nonidling}\label{lem:non-idling}
    If processing times follow Bernoulli distributions, an optimal scheduling policy 
    must be non-idling, and in particular it starts jobs only at time $0$ or at completion times of other jobs.
\end{restatable}

The second property, namely starting jobs only at time $0$ or at other jobs' completion times is also known as \emph{elementary} scheduling policies~\cite{MRW84}. 
Note that the same proof no longer works once we leave the regime of Bernoulli distributions, because of the potential information gain while processing jobs.
In addition, the following lemma from \cite{GMZ2023} exhibits in which order the jobs with identical size parameter $p_J$ have to be scheduled. The proof is an exchange argument and can be found (for a slightly weaker statement) in~\cite{GMZ2023}.

\begin{lemma}[Lemma 2.2 in \cite{GMZ2023}]\label{lem:identical_s_J}
    Consider an instance in which processing times follow Bernoulli distributions and some scheduling policy $\Pi$. Let $\Pi'$ arise from $\Pi$ by, for each type $j$, reordering jobs $J$ of type $j$ to be started in non-decreasing order of probability $q_J$. Then $\cost(\Pi')\leq \cost(\Pi)$.
\end{lemma}

\paragraph{Filling.} One approach, which we frequently use, is the following greedy scheduling of a given type of jobs in certain idle periods, which we define here for later reference. First, for different job types $j$, we create ``type-$j$ {spaces}'', i.e., idle periods on one or several  machines. Subsequently any such type $j$ space is \emph{filled} by greedily scheduling jobs of type $j$ in it.
\begin{definition}[Filling]\label{def:filling}
    A type-$j$ space is an idle period of length $p_j$ located on some machine.
    To fill a (set of) type-$j$ space(s) across the machines, type-$j$ spaces are considered in time, that is, in order of increasing left endpoints. For any given type-$j$ space, a type-$j$ job $J$ with minimum $q_J$ among the unscheduled jobs is started at the left endpoint. 
    If the job is short, we continue considering the same space and iterate. Once a job turns out long, the space is filled. We then
    move on to the next type-$j$ space. This continues either until all type-$j$ spaces are filled, or until there are no more type-$j$ jobs.
\end{definition}

We note that in all stochastic schedules described in this paper, we always ensure that strictly 
more type-$j$ spaces are created than there are long type-$j$ jobs. In turn, filling always runs 
out of jobs before all the respective spaces are filled.

\paragraph{Scheduling Bernoulli jobs optimally by dynamic programming.} Given Lemma~\ref{lem:non-idling}, it is helpful to recall  that an optimal scheduling policy has a representation as a possibly gigantic yet finite decision tree, with as nodes the exponentially many possible states together with the policies' scheduling decision to schedule the next job $J$ on some idle machine, and as outgoing edges of such a node the two possible realizations of job $J$'s processing time, leading to two different states with corresponding probabilities $q_J$ and $(1-q_J)$, respectively. 
Any root-leave path of this tree has length $N$, and corresponds to an order in which the $N$ jobs have been started. This is more or less folklore; see also \cite{GMZ2023} for the same discussion\footnote{In the context scheduling normally distributed jobs, also \cite{TCZ2021} exhibits certain dynamic programming algorithms, yet for the objective to maximize the probability of having a bounded makespan, and different from what we do here.}.
That means that one can compute an optimal (non-idling) scheduling policy by dynamic programming, albeit clearly not efficiently. As our approach to obtain a PTAS is a refinement of this idea, it helps to understand how this can be done.

To that end, denote by $\vprofile$ an $m$-dimensional vector of machine loads with the interpretation that $\profile_i$ equals the sum of the non-zero-job sizes $p_J$ that have been scheduled on machine $i$.  Then, given $\vprofile$, time $t^*:=t^*(\vprofile):=\min_\ell m_\ell$ is the earliest point in time when a machine is idle for processing the next job. Let $\vjobsleft\subseteq \{1,\dots,N\}$ be a variable to denote the set of yet unscheduled jobs\footnote{One may wonder about the choice of notation here; this is purely to be in line with the later parts of the paper, where $\vjobsleft$ denotes the vector of the numbers of leftover jobs per type.}. 
Denote by $\ecost{\vprofile,\vjobsleft}$
the minimal expected sum of completion times for scheduling unscheduled job set $\vjobsleft$, given the current machine loads~$\vprofile$. Note that by independence across jobs, the vector of machine loads $\vprofile$ suffices to describe the snapshot of the schedule that is relevant for future decisions, and together with the leftover jobs $\vjobsleft$ we have all relevant information to make a scheduling decision at time $t^*$. Also note that, due to the policy being non-idling, it is implicit that the first job, say job $J\in\vjobsleft$, is to be scheduled at time $t^*$ on a machine $i^*\in\argmin_\ell m_\ell$. At time $t^*$ the policy could start any job $J\in \vjobsleft$. This results in two possible states: With probability $q_J$ the processing time of job $J$ equals $X_J=p_J$, resulting in state $(\vprofile^J,\vjobsleft\setminus J)$, with $\vprofile^J$ as new vector of machine loads, so $m^J_{i^*}= m_{i^*}+p_J$ and $m^J_i=m_i$ for all machines $i\neq i^*$. With probability $(1-q_J)$  the processing time of job $J$ equals $X_J=0$, resulting in state 
$(\vprofile,\vjobsleft\setminus J)$.
Therefore the minimal expected sum of completion times for scheduling unscheduled job set $\vjobsleft$, given the current machine loads~$\vprofile$, are 
\begin{equation}\label{eq:dp_recusrion_simple}
    \ecost{\vprofile,\vjobsleft}= \min_{J\in\vjobsleft}\left[q_J\biggl(\ecost{\vprofile^J,\vjobsleft\setminus J}+(t^*+p_J)\biggr) + (1-q_J)\biggl(\ecost{\vprofile,\vjobsleft\setminus J}+t^*\biggr)\right]\,. 
\end{equation}
Clearly, the set of all possible machine profiles $\vprofile$ may be exponential, yet it is finite. To compute the (expected cost of an) optimal scheduling policy, we need to compute $\ecost{\mathbf{0},\{1,\dots,N\}}$, which we can do via \eqref{eq:dp_recusrion_simple} as follows. For a given set of remaining jobs $\vjobsleft$, let us denote by $\vprofile(\vjobsleft)$ the set of the finitely many machine load profiles for the jobs $\{1,\dots,N\}\setminus\vjobsleft$ that are already scheduled. 
Because jobs follow Bernoulli distributions, $\vprofile(\vjobsleft)$ can be any $m$-partition of any subset of the sizes $\{p_J\ |\ J\in  \{1,\dots,N\}\setminus\vjobsleft\}$. As a sanity check, note $\vprofile(\{1,\dots,N\})=\{\mathbf{0}\}$.
The base case is given as $\ecost{\vprofile,\emptyset}=0$ for all $\vprofile\in\vprofile(\emptyset)$. 
We compute $\ecost{\mathbf{0},\{1,\dots,N\}}$ by computing $\ecost{\vprofile,\vjobsleft}$ for all $\vprofile\in\vprofile(\vjobsleft)$ via \eqref{eq:dp_recusrion_simple}, for all $\vjobsleft$ in order $|\vjobsleft|=0$,  $|\vjobsleft|=1$, \dots, $|\vjobsleft|=N$. 
Note that this can be exponential in $N$, yet
it is finite. A rough estimate gives a computation time of \bigO{N2^{N^2+N}}, since there are \bigO{2^N} many sets of unscheduled jobs $\vjobsleft$, \bigO{(2^N)^m} $\subseteq$ \bigO{2^{N^2}} many load profiles $\vprofile\in\vprofile(\vjobsleft)$, and the recursion itself has depth \bigO{N}. The policy itself is defined by the mapping of the states $(\vprofile,\vjobsleft)$ to the decision to start job $J$ as a minimizer in~\eqref{eq:dp_recusrion_simple}.
We summarize.
\begin{theorem}
    The above dynamic program yields an optimal non-anticipatory scheduling policy for scheduling Bernoulli jobs.
\end{theorem}
Note that other, maybe simpler dynamic programming formulations are also possible, yet the above presentation lends itself for the necessary adjustments later in the paper when we assume a bounded number of job types; see Section~\ref{sec:dp}. 
Also note that the computation time is polynomial for a constant number of jobs. Using Lemma~\ref{lem:identical_s_J}, one can also argue that the computation time of the above dynamic program is polynomial if there is only a constant number of job types and a constant number of machines. However the sole assumption of a constant number of job types does not suffice to get an efficient algorithm.

\section{Basic Ideas and Structural Theorem for Two Job Types}\label{sec:warmup_two_types}
To introduce the main ideas behind the structural theorem and not cover them up with excessive notation, we here consider the special case where the input consists of only two types of jobs with size parameters $p_1>p_2$.

\paragraph{Overview.}
The basic idea is to ensure that jobs of type $p_i$ are only started at multiples of $\eps p_i$ so that the number of different load profiles per job type is bounded by \bigO{m^{1/\eps}}, which leads to a polynomial bound on the state space in the dynamic program. However, there are a few observations and caveats in dealing with this, discussed next.

A first assumption that we make is that the two size parameters $p_1,p_2$ are 
sufficiently separated, parametrized by $\eps$, that is,
\begin{align*}
    p_2\leq \eps^2 p_1\,.
\end{align*} 
If this assumption is not true, it means that the two size parameters have a constant ratio.
This case is much simpler, essentially because, when restricting to start times that are multiples of $\eps p_2$ for \emph{both} types of jobs, we only lose a $1+\bigO{\eps}$ factor in the objective function and obtain a polynomial number of possible machine load profiles to keep track of in the dynamic program. 
An additional assumption that is w.l.o.g.\ and will be helpful is that $1/\eps$ is an integer. 

Now, the main issue that has to be resolved is the following. Assume that we would simply alter the given optimal scheduling policy $\Pi$ by ``aligning'' all jobs to be started only at their eligible time points $\eps p_1$ and $\eps p_2$, respectively. Then we either violate the relative order of start times and hence the resulting policy would no longer be guaranteed to be non-anticipatory, because the decision to start jobs could depend on realizations of jobs that have been started earlier. Or, in the attempt to fix this issue by just enforcing this relative order, we would have to potentially introduce too large delays for the jobs with smaller size $p_2$, because $\eps p_1 \gg p_2$. The solution is to depart from the way $\Pi$ schedules the small sized jobs by introducing idle periods on the machines in which we greedily schedule ``enough of them'', so that -in a nutshell- at any later point in time our new policy has no information deficit in comparison to the original, optimal policy $\Pi$, and therefore we can still simulate its scheduling decisions of long jobs. Based on an optimal scheduling policy $\Pi$, this process is done in 
four phases, which we sketch next.

In Phase 1, we generate a certain idle interval on each machine, while maintaining a bounded total expected cost. In Phase 2, these idle intervals are partially filled with greedy processing of small jobs (type-2) in order of their index, and all other processing of type-2 jobs is done on the basis of greedy filling of the type-2 spaces with type-2 jobs, according to Definition~\ref{def:filling}. As a result, the (stochastic) schedule has the property that it will always be ``ahead'' in the processing of type-2 volume, when compared to the original (stochastic) schedule.
Then we slightly stretch the $\eps p_1$ time grid in Phase 3, to avoid long type-2 jobs ``crossing'' the $\eps p_1$ time points, and to generate some more idle intervals to accommodate the final phase. In the final Phase 4, we reorder the scheduling of type-1 and type-2 jobs within all intervals of the (stretched) $\eps p_1$ time grid, to make sure that the type-1 jobs, if any, start at the beginning of that interval, and the type-2 jobs are scheduled in an interval after that. The resulting delay of type-2 jobs is affordable because of being ahead with processing of type-2 volume. Arguing that the resulting policy is still non-anticipatory is key, and crucially uses the property of always being ahead in terms of processing type-2 volume. That type-2 jobs are scheduled at eligible time points will follow from scheduling them greedily.

\paragraph{Details.} Let $\Pi$ be an optimal non-anticipatory policy. By Lemma~\ref{lem:non-idling}, it is non-idling, and by Lemma~\ref{lem:identical_s_J}, if
it starts a job $J$ of type $j$, then $J$ has the smallest index of all remaining jobs of type $j$. As a matter of fact, the last property will remain to hold throughout.
We construct a ``stratified'' policy $\Pi'$ (defined next) from $\Pi$ in 
four
phases, so that $\Pi'$ does not have a much higher total expected cost than $\Pi$. 

\begin{definition}[Stratified policy for two types]\label{def:well-formedness-2}
    Given time points $Q_2\supseteq Q_1\ni 0$  
    and \emph{threshold time points} $p_1^\circ\in Q_1, p_2^\circ\in Q_2$,
    we say policy $\Pi$ is \emph{stratified with respect to $(Q_1,Q_2)$ and $(p^\circ_1,p_2^\circ)$}
    if it satisfies the following two properties for all times $t$.
    \begin{enumerate}[label=(\roman*), noitemsep]
        \item\label{item:2types-stratified:a} If $\Pi$ starts a job $J$ of type $j$ at time $t$, then $t\in Q_j$.
        \item\label{item:2types-stratified:b} If $\Pi$ starts a job $J$ of type $j$ and $J$ is long, 
            then $\Pi$ starts no job on the same machine between $c_J$, the completion time of $J$, 
            and $t' := \min \{t\in Q_j\mid t\ge \max\{p_j^\circ,c_J \}\}$, the next time point from $Q_j$ not before~$p_j^\circ$.
        \item\label{item:2types-stratified:c} $\Pi$ starts jobs of the same type in order of index.
    \end{enumerate}
\end{definition}

For two job types we will use $p_j^\circ=p_j$, but in the next section with more than two types we will need to adapt the choice of the $p^\circ_j$'s. For simplicity we therefore omit the superscript ${}^\circ$ in this section. The time points $Q_1,Q_2$ will be defined at the end of the following sequence of 
four transformations. In each of the four
phases that follow, we define a certain stochastic schedule. For the first and last, we will also argue that there exists a policy that computes this stochastic schedule. The first stochastic schedule is $\schedule=\schedule(\Pi)$.

\subsection{Phase~\texorpdfstring{$1$}{1}: Generating idle time}
In this phase, we insert idle intervals that will be filled with a number of type-2 jobs in the subsequent phase. 
To formally describe this phase, consider any job $J$ of type~$j$.

\noindent\fbox{%
    \begin{minipage}[t][][t]{0.98\textwidth}\vspace{0pt}
        \textbf{Stochastic schedule $\schedule^1$:} For $J$, let $S_J^1=S_J$ if $S_J<p_1$, 
        and $S_J^1=S_J +3(1+\eps)\eps p_1$ otherwise. Every job is started on the same 
        machine as in $\schedule$.
    \end{minipage}
}
\smallskip

\begin{figure}
    \begin{tikzpicture}

	\begin{scope}
		\clip(-1,-0.7) rectangle (13.6,2.5);
		\node at (0,2) {before:};
		
		\shortjobs{0,0.1,...,1.7}
		\foreach \x/\num/\col in {0.0/0/red1,0.1/0/red1,0.2/0/red1,0.2/1/red1,0.2/2/red1,0.3/0/red1,0.3/1/red1,0.5/0/red1,0.5/1/red1,0.8/0/blue1,0.8/1/red1,0.9/0/blue1,0.9/1/blue1,1.0/0/red1,1.0/1/red1,1.0/2/red1,1.1/0/blue1,1.1/1/red1,1.2/0/red1,1.2/1/red1,1.3/0/red1,1.4/0/red1,1.5/0/red1,1.6/0/red1,1.7/0/blue1,1.7/1/red1,1.7/2/red1}{
			\zerolength{\x}{\num}{\col};
		}
		\longjob{1.7}
		\shortjobs{9.7,9.8,...,10.45}
		\foreach \x/\num/\col in {9.7/0/red1,9.8/0/red1,9.8/1/red1,9.9/0/blue1,9.9/1/blue1,9.9/2/red1,9.9/3/red1,9.9/4/red1,9.9/5/red1,10.2/0/blue1,10.2/1/blue1,10.2/2/red1,10.2/3/red1,10.2/4/red1,10.2/5/red1,10.2/6/red1,10.3/0/red1,10.3/1/red1,10.3/2/red1,10.3/3/red1,10.4/0/blue1,10.4/1/red1}{
			\zerolength{\x}{\num}{\col};
		}
		\longjob{10.5}
		
		\draw[->,thick] (0,0) -- (14,0);
		\foreach \x in {0,8,9.7}{
			\draw[thick] (\x,-0.1) -- (\x,0.1);
		}
		
		\node at (8,-0.4) {\small $p_1$};
        \node at (0,-0.4) {\small $0$};
        \node at (9.7,-0.4) {\small $C_i^+$};
	\end{scope}
	\begin{scope}[yshift = -3.5cm]
		\clip(-1,-0.7) rectangle (13.6,2.5);
		\node at (0,2) {after:};
		
		\shortjobs{0,0.1,...,1.7}
		\foreach \x/\num/\col in {0.0/0/red1,0.1/0/red1,0.2/0/red1,0.2/1/red1,0.2/2/red1,0.3/0/red1,0.3/1/red1,0.5/0/red1,0.5/1/red1,0.8/0/blue1,0.8/1/red1,0.9/0/blue1,0.9/1/blue1,1.0/0/red1,1.0/1/red1,1.0/2/red1,1.1/0/blue1,1.1/1/red1,1.2/0/red1,1.2/1/red1,1.3/0/red1,1.4/0/red1,1.5/0/red1,1.6/0/red1,1.7/0/blue1,1.7/1/red1,1.7/2/red1}{
			\zerolength{\x}{\num}{\col};
		}
		\longjob{1.7}
		\shortjobs{13.075,13.175,...,13.775}
		\foreach \x/\num/\col in {13.075/0/red1,13.175/0/red1,13.175/1/red1,13.275/0/blue1,13.275/1/blue1,13.275/2/red1,13.275/3/red1,13.275/4/red1,13.275/5/red1,13.575/0/blue1,13.575/1/blue1,13.575/2/red1,13.575/3/red1,13.575/4/red1,13.575/5/red1,13.575/6/red1,13.675/0/red1,13.675/1/red1,13.675/2/red1,13.675/3/red1,13.775/0/blue1,13.775/1/red1}{
			\zerolength{\x}{\num}{\col};
		}
		
		\draw[->,thick] (0,0) -- (14,0);
		\foreach \x in {0,8,13.075,9.7}{
			\draw[thick] (\x,-0.1) -- (\x,0.1);
		}
		
		\draw[<-,thick] (9.7,0.35) -- (10.4,0.35);
		\draw[->,thick] (12.375,0.35) -- (13.075,0.35);
		\node at (11.3875,0.35) {\small $3(1+\varepsilon)\varepsilon p_1$};
		
		\node at (8,-0.4) {\small $p_1$};
        \node at (0,-0.4) {\small $0$};
        \node at (9.7,-0.4) {\small $C_i^+$};
        \node at (13.075,-0.4) {\small $S_i^+$};
	\end{scope}
    \end{tikzpicture}
    \caption{A specific machine under schedules $\schedule$ and $\schedule^1$ respectively. 
    Blue jobs have length $p_1$, red jobs have length $p_2$, and the little circles denote 
    zero length jobs. Note that the order of circles does not matter because we are depicting 
    a schedule, not a policy. In our figures, $\eps=1/8$, and $p_2=p_1/80$.}
    \label{fig:ph1}
\end{figure}

For a particular machine, $\schedule^1$ is shown in Figure~\ref{fig:ph1}.
Note that $\schedule^1$ has an idle-time interval on every machine $i$, from the first completion time at or after time $p_1$ on that machine, called $C_i^+$ in the following, until time $S_i^+:= C_i^+ + 3(1+\eps)\eps p_1$. If on some machine $i$ no job ends after time $p_1$, there is no reserved idle interval necessary, and one may think of $C_i^+=\infty$. The following lemma is straightforward.

\begin{restatable}{lemma}{lemtwophaseone}\label{lem:2-phase1}
    Stochastic schedule $S^1$ is feasible, there exists a policy $\Pi^1$ such that $S(\Pi^1)=S^1$, 
    and for every job $J$, $S_J^1\le (1+\eps_1)S_J$ where $\eps_1\in\bigO{\eps}$.
\end{restatable}

\subsection{Phase~\texorpdfstring{$2$}{2}: (Partially) filling the idle time}
In this phase, we transform $\schedule^1$ into $\schedule^2$ by partially filling the newly created 
idle intervals $[C_i^+,S_i^+)$ with jobs of type $2$ of a specific total volume. This results in a 
policy that, after some specific time point, has an information surplus regarding the realizations 
of such jobs, which in turn will allow us to ``rearrange'' some jobs later.

\noindent\fbox{%
    \begin{minipage}[t][][t]{0.98\textwidth}\vspace{0pt}
        \textbf{Stochastic schedule $\schedule^2$:} For each machine $i$, copy $\schedule^1$ up to time $C_i^+$. If $C_i^+<\infty$, reserve $(\lfloor \frac{\eps p_1}{p_2}\rfloor+1)$ consecutive \emph{type~$2$-spaces} in the interval $[C_i^+ +(1+\eps)\eps p_1, C_i^+ +(1+\eps)\eps p_1 + (\lfloor \frac{\eps p_1}{p_2} \rfloor +1) p_2)$ on machine~$i$.
        
        Now suppose $\policy^1$ starts job~$J$ on machine $i$ at time $t \ge S_i^+$. There are three cases:
        \begin{compactitem}
            \item If $J$ is of type $1$, then in $\schedule^2$, job $J$ is also started on machine $i$ at time $t$.
            \item If $J$ is of type $2$ and long, then in $\schedule^2$, reserve a type-$2$ space $[t,t+p_2)$ on machine $i$.
            \item If $J$ is of type $2$ and short, then we ignore it (for now).
        \end{compactitem}
        The stochastic schedule $\schedule^2$ is then obtained by filling the type-$2$ spaces with the unscheduled type-$2$ jobs according to Definition~\ref{def:filling}, plus one additional dummy job of type $2$ that is always long. This dummy job is the last job to be filled and its completion time does not contribute to the objective function.
    \end{minipage}
}
\smallskip

\begin{figure}
    \begin{tikzpicture}
	\begin{scope}
		\clip(-1,-0.7) rectangle (13.6,2.5);
		\node at (0,2) {before:};
		
		\shortjobs{0,0.1,...,1.7}
		\foreach \x/\num/\col in {0.0/0/red1,0.1/0/red1,0.2/0/red1,0.2/1/red1,0.2/2/red1,0.3/0/red1,0.3/1/red1,0.5/0/red1,0.5/1/red1,0.8/0/blue1,0.8/1/red1,0.9/0/blue1,0.9/1/blue1,1.0/0/red1,1.0/1/red1,1.0/2/red1,1.1/0/blue1,1.1/1/red1,1.2/0/red1,1.2/1/red1,1.3/0/red1,1.4/0/red1,1.5/0/red1,1.6/0/red1,1.7/0/blue1,1.7/1/red1,1.7/2/red1}{
			\zerolength{\x}{\num}{\col};
		}
		\longjob{1.7}
		\shortjobs{13.075,13.175,...,13.775}
		\foreach \x/\num/\col in {13.075/0/red1,13.175/0/red1,13.175/1/red1,13.275/0/blue1,13.275/1/blue1,13.275/2/red1,13.275/3/red1,13.275/4/red1,13.275/5/red1,13.575/0/blue1,13.575/1/blue1,13.575/2/red1,13.575/3/red1,13.575/4/red1,13.575/5/red1,13.575/6/red1,13.675/0/red1,13.675/1/red1,13.675/2/red1,13.675/3/red1,13.775/0/blue1,13.775/1/red1}{
			\zerolength{\x}{\num}{\col};
		}
		
		\draw[->,thick] (0,0) -- (14,0);
		\foreach \x in {0,8,13.075,9.7}{
			\draw[thick] (\x,-0.1) -- (\x,0.1);
		}
		
		\node at (8,-0.4) {\small $p_1$};
        \node at (0,-0.4) {\small $0$};
        \node at (9.7,-0.4) {\small $C_i^+$};
        \node at (13.075,-0.4) {\small $S_i^+$};
	\end{scope}
	\begin{scope}[yshift = -3.5cm]
		\clip(-1,-0.7) rectangle (13.6,2.5);
		\node at (0,2) {after:};
		
		\shortjobs{0,0.1,...,1.7}
		\foreach \x/\num/\col in {0.0/0/red1,0.1/0/red1,0.2/0/red1,0.2/1/red1,0.2/2/red1,0.3/0/red1,0.3/1/red1,0.5/0/red1,0.5/1/red1,0.8/0/blue1,0.8/1/red1,0.9/0/blue1,0.9/1/blue1,1.0/0/red1,1.0/1/red1,1.0/2/red1,1.1/0/blue1,1.1/1/red1,1.2/0/red1,1.2/1/red1,1.3/0/red1,1.4/0/red1,1.5/0/red1,1.6/0/red1,1.7/0/blue1,1.7/1/red1,1.7/2/red1}{
			\zerolength{\x}{\num}{\col};
		}
		\longjob{1.7}
		\shortjobs{13.075,13.175,...,13.775}
		\foreach \x/\num/\col in {13.075/0/red1,13.175/0/red1,13.275/2/red1,13.275/3/red1,13.375/0/red1,13.575/2/red1,13.775/0/red1,13.775/1/red1,13.775/2/red1,13.275/0/blue1,13.275/1/blue1,13.575/0/blue1,13.575/1/blue1}{
			\zerolength{\x}{\num}{\col};
		}
		
		\shortjobs{10.825,10.925,...,11.9}
		\foreach \x/\num in {10.825/0,10.925/0,10.925/1,11.025/0,11.025/1,11.025/2,11.025/3,11.325/0,11.325/1,11.325/2,11.325/3,11.325/4,11.425/0,11.425/1,11.425/2,11.425/3,11.525/0,11.725/0,11.825/0}{
			\zerolength{\x}{\num}{red1};
		}
		
		\draw[->,thick] (0,0) -- (14,0);
		\foreach \x in {0,8,13.075,9.7}{
			\draw[thick] (\x,-0.1) -- (\x,0.1);
		}
		
		\node at (8,-0.4) {\small $p_1$};
        \node at (0,-0.4) {\small $0$};
        \node at (9.7,-0.4) {\small $C_i^+$};
        \node at (13.075,-0.4) {\small $S_i^+$};
	\end{scope}
    \end{tikzpicture}

    \caption{A specific machine with respect to $\schedule^1$ and $\schedule^2$. 
    Note the filling of type 2-spaces.}
    \label{fig:ph2}
\end{figure}

See Figure~\ref{fig:ph2} for a depiction of $\schedule^2$ on a specific machine.

We introduce the dummy job to ensure that one additional space is reserved for 
zero-length type-$2$ jobs that are not followed by a long type-$2$ job, to make sure 
that they can be started in this and subsequent phases. Note that, due to the filling 
of type-2 jobs in the idle interval that has been created in Phase~1, all type-2 jobs 
have been relocated to be processed earlier, so that there is at least one empty space. 
Thus, this additional dummy job does not delay any other jobs.
 
Next, we note some straightforward properties of $\schedule^2$. 

\begin{restatable}{lemma}{lemmatwophasefeasible}\label{lem:2-phase2-feasible}
    Stochastic schedule $\schedule^2$ is feasible.
\end{restatable}

In this phase and the next one, we do not specifically argue that there is a (non-anticipatory) policy that computes $\schedule^2$, or that the cost of $\schedule^2$ is reasonably bounded, even though these statements hold. We later argue that analogous statements hold after Phase $4$. To this end, we prove another lemma, for which we let $V_2(t,i,\schedule)$ be the (stochastic) total type-$2$ volume processed up to time $t$  on machine $i$ by schedule $\schedule$.

\begin{restatable}{lemma}{lemtwotypesphasetwo}\label{lem:2-types-phase2}
    Consider any machine $i$. The following hold:
    \begin{enumerate}[label=(\roman*)]
        \item \label{item:2-types-phase2:equal_volume_1} For any time $t\geq 0$, 
        the number of type-$1$ jobs that $\schedule^2$ has started on machine $i$ 
        by time $t$ is equal to the number of type-$1$ jobs that $\schedule^1$ 
        has started on machine $i$ by time $t$.
        \item \label{item:2-types-phase2:ahead} For any $t\ge S_i^+$:    
        if $\schedule^2$ has not started all type-$2$ jobs by $t$, 
        then $V_2(t,i,\schedule^2) > V_2(t,i,\schedule^1) + \eps p_1$.
        \item \label{item:2-types-phase2:idle} Machine $i$ is idle throughout 
        $[C_i^+, C_i^+ +(1+\eps)\eps p_1)$ and $[C^+_i+2(1+\eps)\eps p_1, S_i^+)$.
    \end{enumerate}
\end{restatable}

\begin{restatable}{lemma}{lemmatwotwovolumephtwo}\label{lem:2-2-volume-ph2}
    The total volume of type-$2$ jobs started on a machine $i$ under $\schedule^2$ in an interval $I_k$ 
    for $k\geq 1$ is strictly less than $(1+\eps)\eps p_1$ and, if $\schedule^2$ starts a long type-1 
    job on $i$ within $I_k$, then it is strictly less than $\eps p_1$.
\end{restatable}

\subsection{Phase~\texorpdfstring{$3$}{3}: Extending intervals}
To define the next phase, we first partition $[0,\infty)$ into intervals. More specifically, 
let $l_0:=0$, and $l_{k+1}:=p_1+k\eps p_1$, for all $k\in\mathbb{N}_{\ge 0}$. 
Further, we let $I_{k}=[l_k,l_{k+1})$, $k\ge 0$, be the corresponding time intervals in between.

Observe that in $\schedule^2$ it is possible that a single type-$2$ job is processed in two 
consecutive intervals. As preparation for the next phase, we would like to avoid this. In 
addition, and in anticipation of satisfying $(ii)$ in Definition~\ref{def:well-formedness-2}, 
we would like to ensure that, after the completion time of any long type-$1$ job on machine $i$, 
there is sufficient time where no job starts on machine $i$. For these reasons, we generously 
extend intervals by a $(1+5\eps)$ factor.
Formally, for $I_k=[l_k,l_{k+1})$ with $k\ge 0$, let $I_k':=[l_k',l_{k+1}')$ with $l_k':=(1+5\eps)l_k$.

To keep track of how time points in different schedules correspond, for 
any two phases\footnote{Our functions do not refer to the original schedule 
because the transformation is machine-dependent between $\schedule$ and $\schedule^1$.} 
$i,j\in\{1,2,3\}$ we define the function $f^{i\to j}:\mathbb{R}_{\geq 0}\to\mathbb{R}_{\geq 0}$ 
to map any time in $\schedule^i$ to a corresponding time in schedule $\schedule^j$. So for 
$i,j\in\{1,2,3\}$ and some time $t$, the function $f^{i\to j}$ keeps the distance from the 
left endpoint of the ambient interval fixed, where we use $I_0,I_1,\dots$ for 
$\schedule^1,\schedule^2$ and $I_0',I_1',\dots$ for $\schedule^3$. It is sufficient to give 
a formal definition for consecutive phases; the other functions can be obtained by composition:
\begin{align*}
    f^{1\to 2}&=f^{2\to 1} \text{ is the identity,}\\
    f^{2\to 3}(t) &= l'_k+(t-l_k)\text{ where $k=\max\{k' \mid l_{k'}\leq t\}$,}\\
    f^{3\to 2}(t) &= \inf\{t' \mid f^{2\to 3}(t')\geq t\} = \min\{l_k+(t-l'_k),l_{k+1}\}\text{ where $k=\max\{k' \mid l'_{k'}\leq t\}$.}
\end{align*}
Clearly, all of these functions are non-decreasing.

\noindent\fbox{%
    \begin{minipage}[t][][t]{0.98\textwidth}\vspace{0pt}
        \textbf{Stochastic schedule $\schedule^3$:} Every job started in $\schedule^2$ at time $t$ is started in $\schedule^3$ at time $f^{2\to3}(t)$ on the same machine.
    \end{minipage}
}
\smallskip

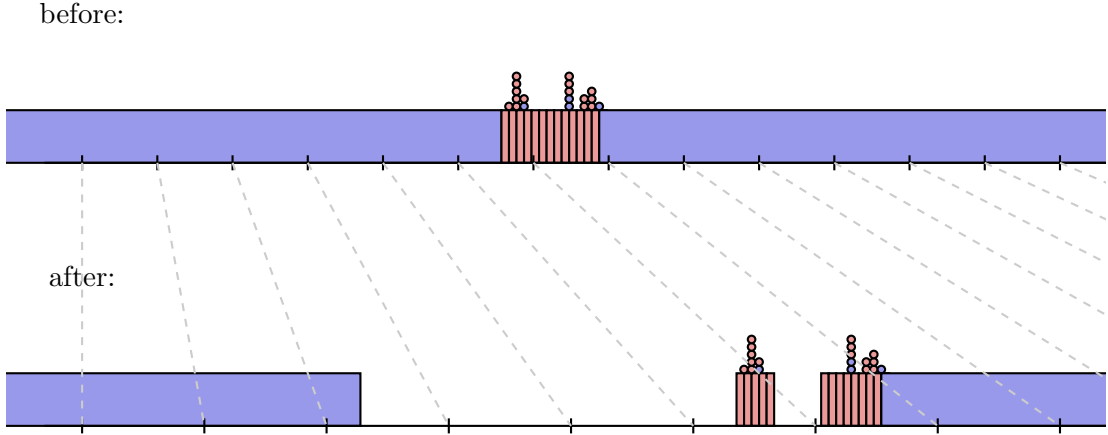
\begin{figure}
    \centering
    \begin{tikzpicture}
	\begin{scope}
		\clip(-1,-0.5) rectangle (13.6,2.5);
		\node at (0,2) {before:};
		
		\longjob{-2.425}
		\shortjobs{5.575,5.675,...,6.88}
		\foreach \x/\num/\col in {5.675/0/red1,5.775/0/red1,5.775/1/red1,5.775/2/red1,5.775/3/red1,5.775/4/red1,5.875/0/blue1,5.875/1/red1,6.475/0/blue1,6.475/1/blue1,6.475/2/red1,6.475/3/red1,6.475/4/red1,6.675/0/red1,6.675/1/red1,6.775/0/red1,6.775/1/red1,6.775/2/red1,6.875/0/blue1}{
			\zerolength{\x}{\num}{\col};
		}
		\longjob{6.875}
		
		\draw[->,thick] (-0.5,0) -- (14,0);
		\foreach \x in {0,1,2,...,13}{
			\draw[thick] (\x,-0.1) -- (\x,0.1);
		}
	\end{scope}
	\begin{scope}[yshift = -3.5cm]
		\clip(-1,-0.5) rectangle (13.6,2.5);
		\node at (0,2) {after:};
		
		\longjob{-4.3}
		\shortjobs{8.7,8.8,...,9.15}
		\shortjobs{9.825,9.925,...,10.65}
		\foreach \x/\num/\col in {8.8/0/red1,8.9/0/red1,8.9/1/red1,8.9/2/red1,8.9/3/red1,8.9/4/red1,9.0/0/blue1,9.0/1/red1,10.225/0/blue1,10.225/1/blue1,10.225/2/red1,10.225/3/red1,10.225/4/red1,10.425/0/red1,10.425/1/red1,10.525/0/red1,10.525/1/red1,10.525/2/red1,10.625/0/blue1}{
			\zerolength{\x}{\num}{\col};
		}
		\longjob{10.625}
		
		\draw[->,thick] (-0.5,0) -- (14,0);
		\foreach \x in {0,1.625,3.25,...,13.5}{
			\draw[thick] (\x,-0.1) -- (\x,0.1);
		}
	\end{scope}
        \begin{scope}
            \clip (-1,-4) rectangle (13.6,2.5);
            \foreach \x/\xx in {0/0,1/1.625,2/3.25,3/4.875,4/6.5,5/8.125,6/9.75,7/11.375,8/13,9/14.625,10/16.25,11/17.875,12/19.5,13/21.125}{
    		  \draw[dashed,thick,draw=black!20] (\x,0) -- (\xx,-3.5);
            }
        \end{scope}
    \end{tikzpicture}
    \caption{Schedule $\schedule^3$. Note how the interval size has increased and how 
    each job starts at the same distance from the left endpoints in the corresponding 
    intervals under $\schedule^2$ and $\schedule^3$.}
    \label{fig:ph3}
\end{figure}

Figure~\ref{fig:ph3} depicts $\schedule^3$.

We argue that schedule $\schedule^3$ is feasible, that type-$2$ jobs 
get processed completely within a single interval, and that after the completion time of 
each long type-$1$ job there is a period of idleness on that machine. 
Furthermore, we give an adapted version of Lemma~\ref{lem:2-types-phase2}. 
\begin{restatable}{lemma}{lemtwotypesphasethreefeasibility}\label{lem:2-types-ph3-feasibility-cost}
    Stochastic schedule $\schedule^3$ is feasible.
\end{restatable}
\begin{restatable}{lemma}{lemtwotypeslessintervals}\label{lem:2types-less-intervals}
    Consider any long type-$1$ job $J$ with $S_J^3\notin I_0'$. Let $i$ be the machine on 
    which $J$ is scheduled in $\schedule^3$, and let $k,\ell$ be so that $C^2_J\in I_k$ 
    and $C^3_J\in I_\ell'$. Then $\ell<k-1$, and in $\schedule^3$ machine $i$ is idle 
    during the interval $[C^3_J,f^{2\to 3}(C^2_J))$. This implies that interval $I'_{\ell+1}$ is idle. 
\end{restatable}
\begin{restatable}{lemma}{lemmmatwotypesphasethree}\label{lem:2-types-phase3}
    Consider any machine $i$. The following hold:
    \begin{enumerate}[label=(\roman*)]
        \item \label{item:equal_volume_1-3} For any time $t\geq 0$, the number of 
        type-$1$ jobs that $\schedule^3$ has started on machine $i$ by time $t$ is 
        equal to the number of type-$1$ jobs that $\schedule^1$ has started on 
        machine $i$ by time $f^{3\to 1}(t)$.
        \item \label{item:ahead-3} For any $t\ge f^{1\to 3}(S_i^+)$:  
        if $\schedule^3$ has not started all type-$2$ jobs by $t$, then 
        \[
            V_2(t,i,\schedule^3) > V_2(f^{3\rightarrow 1}(t),i,\schedule^1) + \eps p_1\,.
        \]
        \item \label{item:idle-3} Let $k,\ell$ be so that $C_i^+\in I_k$ and 
        $S_i^+\in I_\ell$. Note that by construction $\ell \ge k + 3$. Then 
        $\schedule^3$ does not start any jobs on machine $i$ in $I_k'$ or $I_{\ell-1}'$.
    \end{enumerate}
\end{restatable}

The following lemma follows directly from Lemma~\ref{lem:2-2-volume-ph2} and the fact that any job is started in the same interval and machine under $\schedule^3$ as it was under $\schedule^2$.

\begin{restatable}{lemma}{lemmatwotypestype2volphase3}\label{lem:2-phase3-vol2}
    The volume of type-2 jobs started on machine $i$ within an interval $I_k'$ 
    with $k'\geq 1$ under $\schedule^3$ is strictly less than $(1+\eps)\eps p_1$, 
    and, if a long type-1 job starts on $i$ within $I_k'$, 
    then it is even strictly less than $\eps p_j$.
\end{restatable}

\subsection{Phase~\texorpdfstring{$4$}{4}: Reordering jobs}
Define $f^{3\to4}=f^{4\to3}$ as the identity. To obtain $\schedule^4$ from $\schedule^3$, 
we reorder jobs so that every type-$1$ job $J$ with $S^3_J\in I'_k$, $k\ge 1$,
is started at the left endpoint of  $I'_k$, i.e., $S^4_J=l_k'$, and type-2 jobs, if any, follow it.

\noindent\fbox{%
    \begin{minipage}[t][][t]{0.98\textwidth}\vspace{0pt}
        \textbf{Stochastic schedule $\schedule^4$:} Stochastic schedule $\schedule^4$ 
        starts jobs in the same way as $\schedule^3$ in $I_0'$. Consider any other interval $I'_k$ with 
        $k\ge 1$. For all machines $i$, let $\mathcal{J}_k^1(i)$ and $\mathcal{J}_k^2(i)$ be the sets of 
        jobs of type $1$ and $2$, respectively, that are started on machine $i$ under $\schedule^3$ with $S^3_J\in I_k'$. 
        Then $\schedule^4$ starts all jobs of $\mathcal{J}_k^1(i)$ on machine $i$ at $l_k'$, 
        the left endpoint of $I_k'$.
        
        Let $n(\mathcal{J}_k^2(i))$ be the number of long type-$2$ jobs contained in $\mathcal{J}_k^2(i)$. There are two cases for handling $\mathcal{J}_k^2(i)$:
        \begin{enumerate}[label=(\roman*), noitemsep]
            \item\label{item:S4:reserve_space} If $\mathcal{J}_k^1(i)$ contains no long job, 
            then, in $\schedule^4$, we create $n(\mathcal{J}_k^2(i))$ type-$2$-spaces on 
            machine $i$ as early as possible in the interval $I_k'$, namely during $[l_k', 
            l_k'+n(\mathcal{J}_k^2(i))p_2)$.
            \item\label{item:S4:delayed_reserve_space}
            Otherwise, let $J^*$ be the (only) long type-1 job in $\mathcal{J}^1_k(i)$. 
            Note that this is the last job that $\schedule^3$ starts in $I'_k$ on machine $i$.
            Let $\ell$ be such that $C_{J^*}^4\in I_\ell'$.
            In $\schedule^4$, we create $n(\mathcal{J}_k^2(i))$ type-$2$-spaces on machine $i$ 
            as early as possible in the interval \emph{after} $I_\ell'$, namely during $[l_{\ell+1}', 
            l_{\ell+1}'+n(\mathcal{J}_k^2(i))p_2)$.
        \end{enumerate}
        Finally, $\schedule^4$ is obtained by filling the type-2 spaces, as per Definition~\ref{def:filling}.
    \end{minipage}
}
\smallskip

\begin{figure}
    \centering
    \begin{tikzpicture}
	\begin{scope}
		\clip(-1,-0.5) rectangle (13.6,2.5);
		\node at (0,2) {before:};
		
		\longjob{-4.3}
		\shortjobs{8.7,8.8,...,9.15}
		\shortjobs{9.825,9.925,...,10.65}
		\foreach \x/\num/\col in {8.8/0/red1,8.9/0/red1,8.9/1/red1,8.9/2/red1,8.9/3/red1,8.9/4/red1,9.0/0/blue1,9.0/1/red1,10.225/0/blue1,10.225/1/blue1,10.225/2/red1,10.225/3/red1,10.225/4/red1,10.425/0/red1,10.425/1/red1,10.525/0/red1,10.525/1/red1,10.525/2/red1,10.625/0/blue1}{
			\zerolength{\x}{\num}{\col};
		}
		\longjob{10.625}
		
		\draw[->,thick] (-0.5,0) -- (14,0);
		\foreach \x in {0,1.625,3.25,...,13.5}{
			\draw[thick] (\x,-0.1) -- (\x,0.1);
		}
	\end{scope}
	\begin{scope}[yshift = -3.5cm]
		\clip(-1,-0.5) rectangle (13.6,2.5);
		\node at (0,2) {after:};
		
		\longjob{-4.875}
		\shortjobs{3.25,3.35,...,3.8}
		\foreach \x/\num in {3.25/0,3.25/1,3.25/2,3.45/0,3.45/1,3.45/2,3.45/3,3.45/4,3.45/5,3.45/6,3.55/0,3.55/1,3.55/2,3.55/3,3.55/4,3.55/5}{
			\zerolength{\x}{\num}{red1};
		}
		\shortjobs{8.125,8.225,...,8.55}
		\foreach \x/\num in {8.125/1,8.125/2,8.225/0,8.325/0,8.325/1,8.325/2,8.425/0,8.425/1}{
			\zerolength{\x}{\num}{red1};
		}
		\foreach \x/\num/\col in {9.75/0/blue1,9.75/1/blue1,9.75/2/blue1}{
			\zerolength{\x}{\num}{\col};
		}
		\longjob{9.75}
		
		\draw[->,thick] (-0.5,0) -- (14,0);
		\foreach \x in {0,1.625,3.25,...,13.5}{
			\draw[thick] (\x,-0.1) -- (\x,0.1);
		}
		\zerolength{8.125}{0}{blue1};
	\end{scope}
    \end{tikzpicture}

    \caption{Schedule $\schedule^4$. The first block of type-2 jobs corresponds to spaces 
    created due to type-2 jobs scheduled earlier under $\schedule^3$. The second block of 
    type-2 jobs is moved as early as possible in the interval. (Note that job identities 
    may change, because that happens simultaneously on all machines.) Also note that the 
    relative order of type-1 jobs remains the same, and that such jobs only start on 
    left-endpoints of intervals.}
    \label{fig:ph4}
\end{figure}

The reservation of type-2 spaces \emph{after} interval $I_\ell'$ is done to fulfill (ii) 
of Definition~\ref{def:well-formedness-2}, which again is important for the dynamic program 
in Section~\ref{sec:dp}. See Figure~\ref{fig:ph4} for a depiction of schedule~$\schedule^4$.

\begin{restatable}{lemma}{lemmatwotypesphasefournocrossing}\label{lem:2-types-phase4-nocrossing}
    In $\schedule^4$, the execution of any type-2 job is completely in an interval $I_k'$.
\end{restatable}

We come to the crucial part of the analysis of Phase 4. For any machine $i$, we define a 
number of time points: Let $x_{i,k}$ be the start time of the $k$-th long type-$1$ job, 
say $J$, that is started in $\schedule^4$ on machine $i$ at a time $t\ge f^{1\to 4}(S_i^+)$. 
Further, let $\ell$ be such that $C_J^4\in I_\ell'$, and define $y_{i,k}$ to be $l_{\ell+2}'$. 
(If no such job exists, $x_{i,k}=y_{i,k}=\infty$.) 
Note that Lemma~\ref{lem:2types-less-intervals} implies $x_{i,k+1}\ge y_{i,k}$. 
Also note that Lemma~\ref{lem:2types-less-intervals}, Lemma~\ref{lem:2-phase3-vol2}, and the 
construction of $\schedule^4$ imply the following.

\begin{fact}\label{fact:2-yik}
    For any machine $i$ and $k\geq 1$, it holds that 
    \[
        V_2(y_{i,k},i,\schedule^4)-V_2(x_{i,k},i,\schedule^4)<\eps p_1\,.
    \]
\end{fact}
We have the following crucial lemma on processing type-1 jobs and on processing type-2 volume. 
\begin{lemma}\label{lem:2-phase4-volume}
    Consider any machine $i$. The following hold.
    \begin{enumerate}[label=(\roman*)]
        \item\label{item:2-phase4-volume:a} For any time $t\geq 0$, the number of type-$1$ jobs 
        that $\schedule^4$ has started on machine $i$ by time $t$ is at least the number of type-$1$ 
        jobs that $\schedule^1$ has started on machine $i$ by time $f^{4\to 1}(t)$.
        \item\label{item:2-phase4-volume:b} For any $t\ge f^{1\to 4}(S_i^+)$ with 
        $t\notin ( x_{i,k},y_{i,k})$ for all $k\geq 1$ and $t=l_\ell'$ (for some $\ell\ge 1$), 
        if $\schedule^4$ has not started all type-$2$ jobs by time $t$, then 
        \[
            V_2(t,i,\schedule^4) > V_2(f^{4\to 1}(t),i,\schedule^1) + \eps p_1\,.
        \]
    \end{enumerate}
\end{lemma}
\begin{proof}
    We focus on a fixed machine~$i$.
    Part (i) directly follows using that no type-1 job is started later in $\schedule^4$ than 
    in $\schedule^3$ and Lemma~\ref{lem:2-types-phase3} \ref{item:equal_volume_1-3}.

    Note that by Lemma~\ref{lem:2-types-phase3} \ref{item:ahead-3} it suffices to show that for 
    any $t$ with $t\notin ( x_{i,k},y_{i,k})$, for all $k\geq 1$ and $t=l_\ell'$ 
    (for some $\ell\ge 0$), if $\schedule^4$ has not started all type-$2$ jobs by time $t$, then 
    \[
        V_2(t,i,\schedule^4) = V_2(f^{4\to 3}(t),i,\schedule^3)\,.
    \]
    We show this statement by induction on $k$. Recall that $f^{4\to3}(t)=t$. 
    Schedules $\schedule^3$ and $\schedule^4$ process volume of the same type at every $t\in[0,x_{i,1}]$ 
    on machine $i$ and therefore we have $V_2(l_\ell',i,\schedule^4) = V_2(l_\ell',i,\schedule^3)$ 
    for any left endpoint $l_\ell'\in [0,x_{i,1})$, as well as $V_2(x_{i,1},i,\schedule^4) = V_2(x_{i,1},i,\schedule^3)$,
    which completes the base case. Now, assume that the statement holds up to some $x_{i,k}$. 
    We first claim that it then also holds for $y_{i,k}$.  
    This holds because by (ii) in the construction of $\schedule^4$, (if not all type-$2$ volume 
    has been processed by $y_{i,k}$,) the type-$2$ volume processed in interval $[x_{i,k},y_{i,k}]$ 
    is the same in both $\schedule^3$ and $\schedule^4$. 
    To conclude the inductive step, note that the statement extends to all $t\in[y_{i,k},x_{i,k+1}]$ 
    with $t=l_\ell'$ (for some $\ell\ge 1$) because by (i) in the construction of $\schedule^4$, $\schedule^3$ 
    and $\schedule^4$ process the same amount of type-$2$ volume in each interval 
    $[l_\ell',l_{\ell+1}']\subseteq [y_{i,k},x_{i,k+1}]$. 
\end{proof}  

This lemma allows us to show the following central lemma about Phase 4.
	
\begin{lemma}\label{lem:2-types-phase4-feasibility}
    Stochastic schedule $\schedule^4$ is feasible, there exists a non-anticipatory policy $\Pi^4$ 
    such that $S(\Pi^4)=\schedule^4$, and $S_J^4\le (1+5\eps)S_J^1$ for each job $J$. 
\end{lemma}
\begin{proof}
    We first argue the feasibility of $\schedule^4$. By Lemma~\ref{lem:2-types-ph3-feasibility-cost}, 
    it suffices to realize that the ``reordering'' of Phase~4 does 
    no harm to the feasibility. In that respect, the critical step defining $\schedule^4$ is starting the 
    type-1 jobs at the {left endpoints} $l'_{k}$ of intervals $I'_{k}$ in which $\schedule^3$ starts them. 
    This means these jobs are potentially scheduled \emph{earlier} in $\schedule^4$ than in $\schedule^3$. 
    However, Lemma~\ref{lem:2types-less-intervals} guarantees that this cannot create any overlap with 
    earlier processing of type-1 jobs. Moreover, we claim that the subsequent filling of type-2 jobs is 
    feasible, too. Here, it is important to realize that \emph{after} completing a type-1 job in 
    $\schedule^4$ in interval $I'_{\ell}$, Lemma~\ref{lem:2types-less-intervals} shows that the interval 
    $I'_{\ell+1}$ is idle (under $\schedule^3$). That idle interval is sufficient to accommodate the required volume 
    $n(\mathcal{J}_k^2(i))p_2$ of type-2 space (in $\schedule^4$). The other reservations of type-$2$ 
    spaces in (i) are fully contained in an otherwise free interval (Lemma~\ref{lem:2-phase3-vol2}). 
    Thus, this shows that the filling of type-$2$ jobs is feasible, too.
    
    We next show that a non-anticipatory policy $\Pi^4$ exists such that $\schedule(\Pi^4)=\schedule^4$. 
    First, recall that we have only established an analogous statement for Phase 1 (in 
    Lemma~\ref{lem:2-phase1}) and not for Phases 2 and 3. Also recall that $\schedule^1$, 
    $\schedule^2$, $\schedule^3$, and $\schedule^4$ all start jobs identically during 
    interval $I_0$, so $\policy^4$ can simply replicate $\policy^1$ during that interval. 
    Further, during $I_0'\setminus I_0$, $\schedule^4$ does not start any jobs.
    
    Now, consider any interval $I_k'$ with $k\geq 1$. We first show that at the start of the interval, 
    $\policy^4$ has enough information to start all jobs that $\schedule^4$ starts at time $l_k'$.
    To that end, we distinguish the machines based on two different cases, machines $i$ for which 
    $I_k'\subseteq [f^{1\to 4}(C^+_i),f^{1\to 4}(S_i^+))$, and those for which $l_{k+1}'> 
    f^{1\to 4}(S^+_i)$. 
    Note that by Lemma~\ref{lem:2-types-phase2}~\ref{item:2-types-phase2:idle} if 
    $l_k'< f^{1\to 4}(C^+_i)$, no jobs start in $I_k'$, since $k\ne0$. 
    
    On machine $i$ such that $I_k'\subseteq [f^{1\to 4}(C^+_i),f^{1\to 4}(S_i^+))$, 
    interval $I_k'$ only contains start times of type-$2$ spaces, which were defined in Phase $2$, 
    in $\schedule^4$. These start times depend only on the completion time $C_i^+$ of the 
    first job completed after $p_1$ in $\schedule^1$. This in turn only depends on the 
    processing times of the jobs started within $I_0'$, so $\policy^4$ actually has all 
    information necessary to determine all start times of type-$2$ spaces on machine 
    $i$ in $I_k'$. In particular it has sufficient information to start all jobs at time $l_k'$.
    
    For machine $i$ such that $l_{k+1}'> f^{1\to 4}(S^+_i)$, first note that, by 
    Lemma~\ref{lem:2-types-phase2}~\ref{item:2-types-phase2:idle}, no type-$2$ spaces are reserved 
    by $\policy^4$ within the interval $I_k'\cap [f^{1\to 4}(C^+_i),f^{1\to 4}(S_i^+))$.
    I.e., in $I_k'$, policy $\policy^4$ simulates what $\policy^1$ does in $I_k$. 
    Specifically, policy $\policy^4$ starts type-$1$ jobs from $I_k$ at time $l_k'$ in the same order 
    as in $\policy^1$ over all such machines. To be able to do this, $\policy^4$ needs to ``know''
    (i.e., have started) the realization of all type-$2$ jobs that $\policy^1$ starts until $l_{k+1}$, 
    and the realization of all type-$1$ jobs that $\policy^1$ starts until $l_{k}$.
    We make two claims:
    \begin{enumerate}[label=(\roman*), noitemsep]
        \item\label{item:type-1_S4_before_S1} All type-$1$ jobs that start before time $l_k$ 
        in $\schedule^1$ start before time $l_k'$ in $\schedule^4$.
        \item\label{item:type-2_S4_before_S1} All type-$2$ jobs that start before time $l_{k+1}$ 
        in $\schedule^1$ start before time $l_k'$ in $\schedule^4$.
    \end{enumerate}
    Indeed, (i) is implied by Lemma~\ref{lem:2-phase4-volume} (i). To see (ii), note that there 
    are precisely three types of machines to consider:
    \begin{itemize}
        \item Machines $i'$ with $l_{k+1}'\leq f^{1\to 4}(S_{i'}^+)$. We observe that, 
        since $\schedule^4$ and $\schedule^1$ are identical during $I_0$ and $\schedule^1$ does not 
        start any jobs in $[C_{i'}^+,S_{i'}^+)$ on machine $i'$, all type-2 jobs started on machine $i'$ 
        before time $l_{k+1}$ in $\schedule^1$ have been started on machine $i'$ before time $l_k'$ in $\schedule^4$.
        \item Machines $i'$ with $f^{1\to 4}(S_{i'}^+)<l_{k+1}'$ and $l_k'\notin(x_{i',z},y_{i',z})$ 
        for all $z\geq 1$. Lemma~\ref{lem:2-phase4-volume} (ii) implies that 
        $V_2(l_k',i',\schedule^4) > V_2(l_k+\eps p_1,i',\schedule^1) = V_2(l_{k+1},i',\schedule^1)$. 
        In words, by Lemma~\ref{lem:2-phase4-volume} (ii), 
        $\schedule^4$ is more than $\eps p_1$ ahead on type-$2$ volume w.r.t.\ $\schedule^1$, 
        and hence, is still ahead w.r.t.\ $\schedule^1$, at time $\eps p_1$ later.
        \item Machines $i'$ with $f^{1\to 4}(S_{i'}^+)<l_{k+1}'$ and $l_k'\in(x_{i',z},y_{i',z})$ 
        for some $z\geq 1$. Then $l_k'+\eps p_1 < y_{i',z}$ since $y_{i',z}$ is the endpoint of an 
        interval $I_\ell'$. Hence, Lemma~\ref{lem:2-phase4-volume} (ii) combined with Fact~\ref{fact:2-yik} 
        implies that $V_2(l_k',i',\schedule^4) > V_2(l_k+\eps p_1,i',\schedule^1)=V_2(l_{k+1},i',\schedule^1)$.
    \end{itemize}
    Thus, policy $\policy^4$ has sufficient information to start jobs at time $l_k'$. It proceeds as 
    follows. At time $l_k'$, policy $\policy^4$ starts the type-$1$ jobs in the same order that $\policy^1$
    starts them in $I_k$. After that, it has started all necessary type-$1$ jobs, and has 
    all information to determine the number of type-$2$ spaces to reserve on each machine in $I_k'$.
    These are subsequently filled using type-$2$ filling as described in Definition~\ref{def:filling}.
    
    Finally, regarding the start times, we use the fact that $\schedule^1$ and $\schedule^4$ are 
    identical within $I_0$. Moreover, from \ref{item:type-1_S4_before_S1} and 
    \ref{item:type-2_S4_before_S1}, it is clear that no job is delayed 
    by more than a factor $(1+5\eps)$, which is only caused by the stretching in Phase 3, 
    as $l_k'=(1+5\eps)l_k$.
\end{proof}

By construction we also have the following property.

\begin{lemma}\label{lem:2-types-ph5-starttimes}
    Consider a long type-1 job $J$ on machine $i$ with $C_J^4\in I_k'$. Then $\schedule^4$ starts no 
    jobs on $i$ during $I_k'$, and no type-$1$ job during $I_{k+1}'$.
\end{lemma}

We are now ready to define the time points $Q_1,Q_2$, which are the time points at which a stratified 
policy can start jobs of type 1 and~2, respectively. They are later needed to define the dynamic 
program to solve for an optimal stratified policy. 
For $Q_2$, we allow time points that are a multiple of $\eps p_2$ to the right of the left endpoint 
in an interval $I_k'$ such that the corresponding type-2 job would still be fully contained in that interval. 
For $Q_1$, we include all of $Q_2^0$ as well, but only the left endpoint of $I_k'$ for $k\geq 1$.
\begin{definition}
    Let $Q_2:= \bigcup_{k\ge 0}Q_2^k$ where, for all $k\geq 0$,
    \begin{align*}
    	Q_2^k &:= \{ l_k' + i\eps \cdot p_2 \mid (i\in \mathbb{N}_0) \wedge (l_k'+i\eps \cdot p_2 \le l_{k+1}' - p_2)\}.
    \end{align*}
    Moreover,
    \[
        Q_1:=Q_2^0\cup\{l_k' \mid k\geq 1\}\,.
    \]
\end{definition}
Within the definition of stratified policies and in Section~\ref{sec:dp}, 
$Q_j$ is the set of allowed start times for jobs of type $j$. To ensure that
the dynamic-programming algorithm of Section~\ref{sec:dp} works as intended, note that 
the following holds.

\begin{restatable}{lemma}{lemmatwotypesQjs}\label{lem:2types-Qjs}
    Let $t\in Q_2$. If there is a $t'\in Q_{1}$ such that $t<t'<t+p_2$, then $t\in Q_{1}$.
\end{restatable}

Note that, by construction of $\schedule^4$ and $Q_1$, any type-$1$ job starts 
at a timepoint in $Q_1$. In fact, it is not difficult to see that also
jobs of type $2$ start at timepoints in $Q_2$. 
\begin{restatable}{lemma}{lemmatwotypesgridaligned}\label{lem:2types_grid-aligned}
    In $\schedule^4$, any job of type $j\in\{1,2\}$ starts at a timepoint in $Q_j$.
\end{restatable}
\begin{restatable}{lemma}{lemmatwotypesstratified}\label{lem:2-stratified}
    $\schedule^4$ is stratified with respect to $(Q_1,Q_2)$ and $(p_1,p_2)$.
\end{restatable}

We remark that using a grid of width $p_2$ would be sufficient, but using a grid of 
width $\eps p_2$ allows us to generalize our approach more smoothly.

\subsection{Putting all phases together}

The main theorem of this section now follows by combining Lemma~\ref{lem:2-phase1}, 
and Lemma~\ref{lem:2-types-phase4-feasibility} (for feasibility, existence of a 
non-anticipatory policy, and start times) as well as Lemma~\ref{lem:2-stratified} 
(for the policy being stratified).

\begin{theorem}\label{thm:job-by-job-guarantee-2-types}
    Let $\eps \le 3/10$. Consider an input where jobs come from two distinct types with 
    size parameters $p_1>p_2$ and $p_2\le \eps^2p_1$. Then for any non-anticipatory and 
    non-idling policy $\Pi$, there exists a corresponding non-anticipatory (possibly 
    idling) stratified policy $\Pi'$ such that for each realization on which the policies 
    produce respective schedules $\schedule = \schedule(\policy)$ and 
    $\schedule'=\schedule(\policy')$, for each job $J$ holds that
    \begin{align*}
        S_J' \le (1+3(1+\eps)\eps)(1+5\eps)S_J = (1+O(\eps))S_J\,.
    \end{align*}
\end{theorem}

\section{Structural Theorem for \texorpdfstring{$n$}{n} Job Types}\label{sec:multi-types}

In this section, we extend the ideas of Section~\ref{sec:warmup_two_types} to the general case with $n$ job types with size parameters $p_1 > p_2 >\dots >p_n$. We still assume that 
\begin{equation}\label{eq:separation_of_s_j}
    p_j \le \eps^2 p_{j-1}\ \text{for}\  j=2,\dots, n\,,
\end{equation}
which we call \emph{$\eps^2$-separated types}.
We show how to get rid of this assumption in Section~\ref{sec:groups}.
Recall that jobs of the same type are indexed by increasing order of probabilities. 
Furthermore, we assume for the remainder of this section that job sizes divide each other. 
\begin{restatable}{lemma}{lemmadivisibilityforntypes}\label{lem:ntypes-divisible}
    When types are $\eps^2$-separated, at an increase in expected cost of a factor at 
    most $(1+\eps)$, one may assume that $p_j$ divides $p_{j-1}$ for all $j\in\{2,\dots,n\}$.
\end{restatable}

We next define \emph{stratified policies} for $n$ types as a generalization of stratified 
policies for two types. Then, we show how to construct a policy $\policy'$ from $\policy$, 
again in four consecutive phases, such that $\policy'$ is stratified and has total 
expected cost not much higher than policy $\policy$. Before giving the technical details, 
let us sketch the general idea and nomenclature.

With more than two job types, more than two time discretizations are necessary. 
Namely, in steps of size $\eps p_j$ for each type~$j$. 
The roles of the job types which used to be long sizes $p_1$, which we simulated exactly 
as in the original policy, and short sizes $p_2$, for which we made sure to get ahead in 
processing their volume, now will change over time.
We call these jobs \emph{ahead} jobs and \emph{on-par} jobs, respectively. 
Jobs of type $j\in \{1,\dots,n\}$ are on-par jobs roughly until time $p_{j-1}$, 
and they are ahead jobs after that. Intuitively, our policy 
ensures that we have always processed enough extra volume of all ahead types in 
comparison to the optimal policy such that we can replicate what 
the optimal policy does for the on-par jobs, before processing additional ahead jobs. 
To accommodate this, we generate the necessary idle time in Phase 1 by inserting idle intervals 
for the different types. 
We process specific volumes of ahead types in (subintervals of) these intervals in Phase 2.
This ensures that we can build a ``stratified'' version of the optimal policy that can 
simulate the processing of the on-par jobs according to what the optimal policy does.
As in the two-types setting, in order to avoid technical complications such as, e.g., jobs 
``crossing'' the respective time discretizations, we stretch the intervals of the time 
discretizations in Phase 3. This provides space to ensure that the processing of any 
ahead job does not extend beyond an allowed start time of an on-par job. 
In turn, this allows to start the on-par jobs at ``their'' time points, and specifically 
also before the ahead jobs in the respective intervals. This (local) reordering is done in Phase 4. 
We show that the resulting schedule also aligns the start times of the ahead jobs to their allowed 
start times as required by the definition of a stratified policy (see below).
\begin{definition}[Stratified policy for $\eps^2$-separated job types]\label{def:well-formedness}
    Given point sets $Q_\numtypes \supseteq \dots \supseteq Q_1\ni 0$
    and \emph{threshold points} $p_1^\circ\in Q_1, \dots, p_n^\circ\in Q_n$,
    we say policy $\Pi$ is \emph{stratified with respect to $(Q_1,\dots,Q_\numtypes)$ and $(p^\circ_1,\dots,p^\circ_\numtypes)$} if it satisfies the following properties.
    \begin{enumerate}[label=(\roman*), noitemsep]
        \item\label{item:ntypes-stratified:a} If $\Pi$ starts a job $J$ of type $j$ at time $t$, then $t\in Q_j$. 
        \item\label{item:ntypes-stratified:b} If $\Pi$ starts a job $J$ of type $j$, and $J$ is long, 
        then $\Pi$ starts no job on the same machine between $c_J$, the completion time of $J$, 
        and $t' := \min\{t\in Q_j\mid t\ge \max\{p_j^\circ,c_J \}\}$, the next time point from $Q_j$ 
        after~$p_j^\circ$.
        \item\label{item:ntypes-stratified:c} $\Pi$ starts jobs of the same type in order of index.
    \end{enumerate}
\end{definition}

We define the time points $Q_1,\dots,Q_n$ in the last of the four phases.
In each of the four phases, we again first define a stochastic schedule and only later argue 
that there exists a non-anticipatory policy that computes this stochastic schedule. 
The first stochastic schedule is $\schedule=\schedule(\Pi)$, $\Pi$ being an optimal scheduling policy.

\subsection{Phase~\texorpdfstring{$1$}{1}: Generating idle time}

In this phase, we insert an idle interval on every machine around time $p_j$ for each job type that is 
an ahead job from time $p_j$ onward, i.e., all types $j'>j$. These idle times are reserved for processing 
jobs of those types in order to get ahead on them.

Consider job $J$ of type~$j$ and let $S_J$ be its stochastic start time in $\schedule$.

\noindent\fbox{%
    \begin{minipage}[t][][t]{0.98\textwidth}\vspace{0pt}
        \textbf{Stochastic schedule $\schedule^1$:} Start job $J$ at time 
        \[
            S_J^1 = \begin{cases}
                        S_J & \text{if $S_J<p_{n-1}$}\\
                        S_J + (1+\eps)\eps \sum_{j'=j}^{n-1} (n-j'+3) p_{j'} & 
                        \text{if $j\in\{2,\dots,n-1\}$, s.t.\ $p_j\le S_J< p_{j-1}$}\\
                        S_J + (1+\eps)\eps \sum_{j'=1}^{n-1} (n-j'+3) p_{j'} & 
                        \text{if $p_1\le S_J$}
                    \end{cases}
        \]
        on the same machine as in $\schedule$.
    \end{minipage}
}
\medskip

In the following, let $\Cstar{i}{j}$ be the first completion time at or after $p_j$ on machine $i$ in $\schedule$.
Schedule $\schedule^1$ creates idle-time intervals on every machine $i$ and 
for each $j\in \{1,\dots,n-1\}$ of the form $[\Cplus{i}{j},\Splus{i}{j})$, where 
\[
    \Cplus{i}{j} = \Cstar{i}{j} + (1+\eps)\eps \sum_{j'=j+1}^{n-1} (n-j'+3) p_{j'}\,,
\]
and
\[
    \Splus{i}{j}:= \begin{cases}
                    \Cplus{i}{j'} +  (1+\eps)\eps(n-j'+3) p_{j'} & \text{if $j'\in\{2,\dots,n-1\}$, 
                    s.t.\ $p_{j'}\le \Cstar{i}{j}< p_{j'-1}$}\\
                    \Cplus{i}{1} +  (1+\eps)\eps(n+2) p_{1} & \text{if $p_1\le \Cstar{i}{j}$}
                \end{cases}\,.  
\]

If on some machine $i$ no job ends after time $p_j$, creating idle intervals past $p_j$ on that 
machine is unnecessary, and one may think of $\Cstar{i}{j}=\infty$. To understand what this means, 
observe that $\Cplus{i}{j}$ is just the completion time of the first job that ends past $p_j$, 
plus the idle times that have been inserted for all smaller job types before.
As to $\Splus{i}{j}$ note that, if a single job $J$ starts its processing on machine $i$ before $p_j$ and 
completes after $p_{j'}>p_j$, then $\Cstar{i}{j}=\dots=\Cstar{i}{j'}$.
In that case, all the corresponding idle intervals are nested, i.e., $\Cplus{i}{j} < \dots<
\Cplus{i}{j'}$ and  $\Splus{i}{j} = \dots= \Splus{i}{j'}$.

Next, we argue about the properties of the corresponding policy. 

\begin{restatable}{lemma}{lemmantypesphonemain}\label{lem:ntypes-ph1-main}
    Stochastic schedule $\schedule^1$ is feasible, and there exists a policy $\Pi^1$ such 
    that $\schedule(\Pi^1)=\schedule^1$, and for every job $J$, $S_J^1\le (1+ (2n+4)(1+\eps)\eps)S_J$.
\end{restatable}

\subsection{Phase~\texorpdfstring{$2$}{2}: (Partially) filling the idle time}\label{subsec:ntypes_phase2}

In the two-type case, we only introduced idle time after $p_1$, the largest 
possible processing time. However, when we have more than two types, we need to introduce idle time 
before $p_j$ for $j\in \{1,\dots,n-2\}$. In turn, it is possible that some job $J$ with $S_J<p_j$ 
has $S_J^1>p_j$. This results in a situation where in some, relatively small, 
interval after $p_j$, we may want to simulate what $\policy$ does just before $p_j$. 
If this includes starting jobs of type $j$, the interval just after $p_j$ needs time points in 
$Q_j$ (as in the definition of a stratified policy (Def.~\ref{def:well-formedness})) that are less than $\eps p_j$ apart. 
We resolve this issue by introducing for all $j\in\{1,\dots,\numtypes\}$, 
$p_j^*= p_j + (1+\eps)\eps \sum_{j'=j+1}^{n-1} (n-j'+3) p_{j'}$,
which is $p_j$ shifted by the maximum amount of introduced idle time before it and 
equal to $\Cplus{i}{j}$ if $\Cstar{i}{j}= p_j$. This ensures 
that any job that starts before $p_j$ in $\policy$ starts before $p_j^*$ in $\policy^1$.

In this phase, we transform $\schedule^1$ into $\schedule^2$ by partially filling the newly created 
idle intervals $[\Cplus{i}{j},\Splus{i}{j})$ with a specific volume of jobs of each ahead type. 
This results in a policy that, after some specific time point, has an information surplus 
(compared to $\policy^1$) regarding the realizations of such jobs, which in turn allows us 
to ``rearrange'' some jobs later. 

For each machine $i$, let $\mathcal{J}_i= \set{j\in \{1,\dots,n-1\} \given \nexists j'<j,\, 
\Cstar{i}{j'}=\Cstar{i}{j}}$, the set of job types that correspond to inclusion-wise minimal 
intervals $[\Cplus{i}{j},\Splus{i}{j})$. 

\begin{definition}[Ahead and on-par types]
    \label{def:ahead_onpar}
    At time $t$, we say that type $j\in \{2,\dots,n\}$ is an \emph{ahead type} 
    if $t\ge p_{j-1}^*$. All types that are not ahead types at $t$ are \emph{on-par types} at $t$. 
    Additionally, we say that a job $J$ of type $j$ is an \emph{ahead job} if 
    $j$ is an ahead type at time $S^1_J$, the start time of $J$, and, similarly,
    that it is an \emph{on-par job} if $j$ is on-par at time $S^1_J$.
\end{definition}

\noindent\fbox{%
    \begin{minipage}[t][][t]{0.98\textwidth}\vspace{0pt}
        \textbf{Stochastic schedule $\schedule^2$:} For each machine $i$, copy $\schedule^1$ up to time $\Cplus{i}{n-1}$. 
        For each $j\in \mathcal{J}_i$, we reserve spaces in the interval of length 
        $\sum_{k=j+1}^n \lfloor \frac{\eps p_j}{p_k} +1\rfloor p_k$ that starts at  
        $\Cplus{i}{j} +2(1+\eps)\eps p_j$ on machine~$i$.
        Specifically, for all $k\in\{j+1,\dots,n\}$ (in that order and without gaps), 
        we place $\lfloor \frac{\eps p_j}{p_k} +1\rfloor$ consecutive 
        type-$k$~spaces, i.e., all ahead types.
        
        Now, consider a time $t\ge \Cplus{i}{n-1}$, s.t.~$t\notin \bigcup_{j=1}^{n-1}[\Cplus{i}{j},\Splus{i}{j})$, 
        where $\schedule^1$ starts a job~$J$ on machine $i$.  
        There are three cases:
        \begin{compactitem}
            \item If $J$'s type is on-par at time $t$, then in $\schedule^2$, 
            job $J$ is also started on machine $i$ at time $t$.
            \item If $J$'s type is ahead at time $t$ and long, then in $\schedule^2$, 
            we reserve a type-$k$-space $[t,t+p_k)$ on machine $i$.
            \item If $J$'s type is ahead at time $t$ and short, then we ignore it (for now).
        \end{compactitem}
        
        The stochastic schedule $\schedule^2$ is then obtained by filling, for all $j\in\types$, the type-$j$ spaces 
        with the unscheduled type-$j$ jobs and one additional dummy job of type $j$ that is always long, i.e., it is 
        last among type-$j$ jobs to be filled, and its completion time does not contribute to the objective function.
    \end{minipage}
}
\smallskip

Note that no jobs start in the interval $[\Cplus{i}{j}, \Splus{i}{j})$ in $\schedule^1$ for any $i$ or $j$, so the 
definition of $\schedule^2$ considers all jobs started by $\schedule^1$. 
Moreover, if several idle intervals are nested, $\schedule^2$ only fills the smallest one.

We introduce the type-$j$ dummy job to ensure that one space is reserved in this and subsequent schedules where 
zero-length type-$j$ jobs that are not followed by a long type-$j$ job can be started.

In this phase and the next one, we do not specifically argue that there is a (non-anticipatory) 
policy that computes $\schedule^2$, or that the cost of $\schedule^2$ is reasonably bounded, 
even though these statements hold. We later argue that analogous statements hold after Phase $4$. 
To this end, we prove a lemma, for which we let $V_j(t,i,\schedule)$ be the (stochastic) total 
type-$j$ volume processed up to time $t$ by schedule~$\schedule$ on machine $i$.

\begin{restatable}{lemma}{lemmanphasetwofeasible}\label{lem:n-phase2-feasible}
    Stochastic schedule $\schedule^2$ is feasible. 
\end{restatable}

\begin{restatable}{lemma}{lemmantypesphasetwo}\label{lem:n-types-phase2}
    The following hold:
    \begin{enumerate}[label=(\roman*)]
        \item \label{item:n-types-equal_volume_1-2} For any time $t\ge 0$, the number of on-par 
        jobs of type $j$ that $\schedule^2$ has started on machine $i$ at time $t$ is equal to the 
        number of on-par jobs of type $j$ that $\schedule^1$ has started on machine $i$ by time $t$.
        \item \label{item:n-types-ahead-2} For any $t\ge \Splus{i}{k}$ and type $j$ that is {ahead} at $p^*_k$: 
        if $\schedule^2$ has not started all type-$j$ jobs by $t$, then $V_j(t,i,\schedule^2) > V_j(t,i,\schedule^1) + \eps p_k$. 
        \item \label{item:n-types-idle-2} In $\schedule^2$, machine $i$ is idle throughout 
        $[\Cplus{i}{j}, \Cplus{i}{j} +2(1+\eps)\eps p_j)$ and $[\Splus{i}{j} - (1+\eps)\eps p_j, 
        \Splus{i}{j})$, for all $j\in\{1,\dots,n-1\}$.
        \item \label{item:n-types-no-starts-2} In $\schedule^2$, machine $i$ does not start any 
        jobs in $[p_j^*,\Cplus{i}{j})$, for all $j\in\{1,\dots,n-1\}$.
    \end{enumerate}
\end{restatable}

\begin{restatable}{lemma}{lemmaaheadvolumephtwo}
    \label{lem:ahead-volume-ph2}
    The total volume of ahead jobs started on a machine $i$ under $\schedule^2$ in an interval $I_k$ 
    of type $j$ is strictly less than $(1+\eps)\eps p_j$ and, if $\schedule^2$ starts a long on-par 
    job on $i$ within $I_k$, then it is strictly less than $\eps p_j$.
\end{restatable}

\subsection{Phase~\texorpdfstring{$3$}{3}: Extending Intervals}\label{subsec:n-types-Phase3}

To define the next phase, we first partition $[0,\infty)$ into intervals. 
The intervals differ in length depending on the types that are ahead and on par. 
Therefore, we distinguish them based on $p^*_j$.

We specify the intervals by their left endpoint. 
The first is $l_0:=0$. For all $j=n-1,\dots,2$, the next left endpoints are at $p^*_{j}$, 
plus multiples of $\eps p_{j}$, as long as this is strictly smaller than $p^*_{j-1}- \eps p_j$. If $t$
denotes the \emph{last} such point strictly smaller than $p^*_{j-1}- \eps p_j$,  
we add another left endpoint at $t+(p^*_{j-1}-t)/2$. 
By doing this, we ensure a lower bound on \emph{all} interval lengths. 

Finally, we start at $p^*_1$ and space time points $\eps p_1$ apart from then on, indefinitely. 
We denote these points $l_0,l_1,\dots$, indexed in increasing value.
Further, we let $I_{k}=[l_k,l_{k+1})$, $k\ge 0$, be the corresponding time 
intervals in between.
We say that the type of an interval $I_k$ is $j$ if $I_k\subseteq[p_j^*,p_{j-1}^*]$
for $j\in\{2,\dots,n\}$, and its type is $1$ if $I_k\subseteq[p_1^*,\infty)$.
Note that all type-$j$ intervals, $j\in\{2,\dots,n\}$  
have length between $\eps p_j/2$ and $\eps p_j$. 

Observe that in $\schedule^2$ it is possible that a single ahead job is processed 
in two consecutive intervals. As preparation for the next phase,
we insert idle time, such that this is no longer the case. 
In addition, and in anticipation of satisfying the second property of a stratified policy, 
we ensure that, after the completion time  on machine $i$ of any long job of a type that 
is on-par at its start time, there is sufficient time where no job starts on machine $i$. 
For these reasons, we generously extend intervals by a $(1+5\eps)$ factor.
Formally, for $I_k=[l_k,l_{k+1})$ with $k\ge 0$, let $I_k':=[l_k',l_{k+1}')$ with 
$l_k':=(1+5\eps)l_k$.

To keep track of how time points in different schedules correspond to each other, 
for any two phases 
$i$ and $j$, we define the function $f^{i\to j}:\mathbb{R}_{\geq 0}\to\mathbb{R}_{\geq 0}$ 
that maps any time in $\schedule^i$ to a corresponding time point in $\schedule^j$. 
In particular, for $i,j\in\{1,2,3\}$ and some time $t$, the function $f^{i\to j}$ keeps 
the distance from the left endpoint of the ambient interval fixed (where we use 
$I_0,I_1,\dots$ for $\schedule^1,\schedule^2$ and $I_0',I_1',\dots$ for $\schedule^3$). 
It is sufficient to give a formal definition for consecutive phases; the other functions 
can be obtained by composition:
\begin{align*}
    f^{1\to 2}&=f^{2\to 1} \text{ is the identity,}\\
    f^{2\to 3}(t) &= l'_k+(t-l_k)\text{ where $k=\max\{k' \mid l_{k'}\leq t\}$,}\\
    f^{3\to 2}(t) &= \inf\{t' \mid f^{2\to 3}(t')\geq t\} = \min\{l_k+(t-l'_k),l_{k+1}\}
    \text{ where $k=\max\{k' \mid l'_{k'}\leq t\}$.}
\end{align*}
Clearly, all of these functions are non-decreasing.
We now generalize Definition~\ref{def:ahead_onpar}, such that it works in the next 
phases.
\begin{definition}[(Generalization of) Ahead and on-par types]
    In $\schedule^i$, at time $t$, we say that type $j\in \{1,\dots,n\}$ is 
    an \emph{ahead type} 
    if $f^{i\to 1}(t)\ge p_{j-1}^*$, otherwise it is an \emph{on-par type}. 
    Additionally, we say that a job $J$ of type $j$ is an \emph{ahead job} if 
    $j$ is an ahead type at time $S^i_J$, the start time of $J$, and, similarly,
    that it is an \emph{on-par job} if $j$ is on-par at time $S^i_J$.
\end{definition}

\noindent\fbox{%
    \begin{minipage}[t][][t]{0.98\textwidth}\vspace{0pt}
        \textbf{Stochastic schedule $\schedule^3$:} Every job started in $\schedule^2$ 
        at time $t$ is started in $\schedule^3$ at time $f^{2\to3}(t)$ on the same machine.
    \end{minipage}
}
\smallskip

We argue that the resulting schedule $\schedule^3$ is still feasible, that ahead jobs 
are processed completely within a single interval, and that after the completion time 
of each long on-par job there is a specific period of idleness on that machine. 
Furthermore, we show that an adapted version of Lemma~\ref{lem:n-types-phase2} still holds.

\begin{restatable}{lemma}{lemmantypesphthreefeasibilitycost}\label{lem:n-types-ph3-feasibility-cost}
    Stochastic schedule $\schedule^3$ is feasible.
\end{restatable}

\begin{restatable}{lemma}{lemmantypeslessintervals}\label{lem:ntypes-less-intervals}
    Consider any long  job $J$ of an on-par type $j$. Let $i$ be the machine on which $J$ 
    is scheduled in $\schedule^3$, and let $C^3_J\in I_\ell'$. If $C_J^3 < f^{2\to 3}(p_j^*)$, 
    under $\schedule^3$, machine $i$ does not start any jobs in $[C_J^3,\max\{l_{\ell+1}', 
    f^{2\to 3}(p_j^*)\} )$ and neither in the interval after that (this is the interval 
    starting at $\max\{l_{\ell+1}',f^{2\to 3}(p_j^*)\}$).
\end{restatable}

\begin{restatable}{lemma}{lemmantypesphasethree}\label{lem:n-types-phase3}
    Consider any machine $i$ in $\schedule^3$. The following hold:
    \begin{enumerate}[label=(\roman*)]
        \item \label{item:n-types-equal_volume_1-3} For any time $t\ge 0$, and any on-par type $j$, 
        $\schedule^3$ has started exactly the same number of type $j$ jobs at time $t$ on machine $i$ 
        as $\schedule^1$ has started at time $f^{3\to 1}(t)$ on machine $i$.
        \item \label{item:n-types-ahead-3} For any $t\ge f^{2\to 3}(S_i^+(k))$ and type $j$ that is ahead at $p^*_k$:    
        if $\schedule^3$ has not started all type-$j$ jobs by $t$, then 
        \[
            V_j(t,i,\schedule^3) > V_j(f^{3\rightarrow 1}(t),i,\schedule^1) + \eps p_k\,.
        \]
        \item \label{item:n-types-idle-3} Let $S_i^+(k)\in I_\ell$. Schedule $\schedule^3$ does not start 
        any jobs on machine $i$ in $I_{\ell-1}'$.
    \end{enumerate}
\end{restatable}

The following lemma directly follows by Lemma~\ref{lem:ahead-volume-ph2} and the fact that 
any job is started in the same interval and machine under $\schedule^3$ as it was under $\schedule^2$.
\begin{lemma}\label{lem:ahead-volume-ph3}
    The volume of ahead jobs started on machine $i$ within an interval $I_k'$ of type $j$ 
    under $\schedule^3$ is strictly less than $(1+\eps)\eps p_j$, and if a long on-par job 
    starts on $i$ within $I_k'$, then it is even strictly less than $\eps p_j$.
\end{lemma}

\subsection{Phase~\texorpdfstring{$4$}{4}: Reordering jobs}
We define $f^{3\to4}=f^{4\to3}$ as the identity. To obtain $\schedule^4$ from $\schedule^3$, 
we reorder jobs so that every on-par job $J$ with $S^3_J\in I'_k$, $k\ge 1$,
starts at the left endpoint of the corresponding interval, i.e., $S^4_J=l_k'$. 
Intuitively the reason behind this is that it allows us to align
start times to times $Q_1,\dots, Q_n$ (yet to be defined) to satisfy the definition of 
a stratified schedule.
When all long ahead jobs have the same size, as was the case in the previous section, 
then it is easy to make the correct reservations for spaces for the ahead jobs. 
When long ahead jobs can have different sizes, we need to be careful on how we schedule those 
so that they still align to the points in $Q_1,\dots Q_n$ (to be defined shortly). 
We resolve this by greedily reserving the corresponding spaces from largest to smallest 
across all machines. This may result in volume of some type being processed on different machines 
under $\schedule^3$ and $\schedule^4$ within an interval, which additionally exacerbates the proofs.

\noindent\fbox{%
    \begin{minipage}[t][][t]{0.98\textwidth}\vspace{0pt}
        \textbf{Stochastic schedule $\schedule^4$:} Stochastic schedule $\schedule^4$ 
        is identical to $\schedule^3$ in $I_0'$. Consider any other interval $I'_k$ with 
        $k\ge 1$. For any machine $i$, let $\mathcal{J}_k^o(i)$ and $\mathcal{J}_k^a(i)$ be the sets of 
        jobs of on-par and ahead types, respectively, that $\schedule^3$ starts in $I_k'$ on $i$. 
        Then, $\schedule^4$ starts all jobs of $\mathcal{J}_k^o(i)$ on machine $i$ at $l_k'$, 
        the left endpoint of $I_k'$.
        
        Let $n_j(\mathcal{J}_k^a(i))$ be the number of long type-$j$ jobs contained in $\mathcal{J}_k^a(i)$. 
        There are two cases for handling $\mathcal{J}_k^a(i)$:
        \begin{enumerate}[label=(\roman*), noitemsep]
            \item If $\mathcal{J}_k^o(i)$ contains no long on-par job, then for each ahead type $j$, we 
                reserve $n_j(\mathcal{J}_k^a(i))$ many type-$j$ spaces in $I_k'$. 
            \item Otherwise, let $J^*$ be the (only) long on-par job in $\mathcal{J}^o_k(i)$. 
                (Note that this is the last job that $\schedule^3$ starts in $I'_k$ on machine $i$.)
                Let $\ell$ be the index of the interval such that $C_{J^*}^4\in I_\ell'$. 
                For each ahead type $j$, we reserve $n_j(\mathcal{J}_k^a(i))$ type-$j$ spaces in $I_{\ell+1}'$. 
        \end{enumerate}
        Now, for each interval $I_k'$, consider the subset of machines $M^*_k$ on which 
        no long on-par jobs are processed after the above shift of their start times, 
        and consider the total number of reserved spaces for each type of ahead jobs in $I_k'$ 
        over all machines. 
        We greedily schedule these reserved spaces in order of size, large to small, 
        on the machines in $M^*_k$. This is known as LPT scheduling.
        That is, at the earliest time where nothing is reserved on the machines in $M^*_k$, 
        we reserve the next largest 
        space. We repeat this until all spaces have been reserved on the machines in $M^*_k$.
        
        Finally, $\schedule^4$ is obtained by filling the type-$j$ spaces for each ahead type $j$, 
        as per Definition~\ref{def:filling}.
    \end{minipage}
}
\smallskip

\begin{lemma}\label{lem:n-types-phase4-nocrossing}
    In $\schedule^4$, the execution of any ahead job is completely in an interval $I_k'$.
\end{lemma}
\begin{proof}
    Let $j$ be the type of interval $I_k'$, i.e., $|I_k|=\eps p_j$.  
    We first argue that the total volume of spaces that need to be reserved in $I_k'$ does not exceed 
    $(1+\eps)\eps p_j |M_k^*|$.
    Consider any machine $i\in M_k^*$. There are two cases in which we might reserve spaces on $i$: when 
    under $\schedule^3$ no long on-par jobs are processed during $I_k'$ on $i$, or when under $\schedule^4$ 
    a long on-par job completes during $I_{k-1}'$ on $i$.
    For each of these cases, we argue that, for a machine $i$ satisfying the condition, the 
    volume of spaces that need to be reserved in $I_k'$ (on the machines $M_k^*$) induced by $i$ is bounded 
    by $(1+\eps)\eps p_j$.
    \begin{itemize}
    \item Case 1: If no long on-par job is processed on machine $i$ during interval $I_k'$ under $\schedule^3$,
        then the total volume of spaces reserved in $I_k'$ induced by $i$, is at most the total volume of ahead jobs processed 
        during $I_k'$ on $i$ under $\schedule^3$, which by Lemma~\ref{lem:ahead-volume-ph3} is at most $(1+\eps)\eps p_j$.
    \item Case 2: If a long on-par job $J$ which is processed on machine $i$ has $C_J^4\in I_{k-1}'$ and $S_J^4\in I_\ell'$ 
        under $\schedule^4$, then $I_\ell'$ has type $j'\le j$. Therefore, the total volume of spaces reserved 
        in $I_k'$ induced by $i$, is at most the total volume of ahead jobs processed 
        during $I_\ell'$ on $i$ under $\schedule^3$, which by Lemma~\ref{lem:ahead-volume-ph3} is at most $\eps p_j$.
    \end{itemize}
    Summing over all machines in $M^*_k$ proves that the total volume does not exceed $(1+\eps)\eps p_j |M_k^*|$.
    
    Now, we prove that the total volume of the spaces that is reserved in order of size, falls completely within $I_k'$. That is, 
    the latest right endpoint of a reserved space, call it \emph{makespan}, is not larger than $l'_{k+1}$.
    To this end, note that under $\schedule^3$ the total volume of these reserved spaces fits in $|M^*_k|$ 
    intervals of length $(1+\eps)\eps p_j$. For the greedy LPT scheduling of the type-$j$ spaces in $\schedule^4$, 
    it is easy to establish an upper bound on the makespan of this schedule: It is bounded by the sum 
    of the shortest possible makespan, plus the length of the latest space, the one that defines the makespan. 
    Therefore, since all spaces have length at most $\eps^2 p_j$, the makespan is bounded from above 
    by $l'_k+(1+2\eps)\eps p_j\le l'_k+(1+5\eps)\eps p_j \le l'_{k+1}$.
    
    Finally,  note that by our filling procedure any ahead job is completely processed within a space 
    of its type, and thus the proof is complete.
\end{proof}

In fact, if no long on-par jobs are processed by machine $i$ in $I_k'$, the following follows from 
the proof of Lemma~\ref{lem:n-types-phase4-nocrossing}.
\begin{corollary}\label{cor:ntypes_aheadjobs_emptyspace}
    If machine $i$ does not process any long on-par jobs in a type-$j$ interval $I_k'$ 
    under~$\schedule^4$, then $i$ is idle during the interval $(l_{k+1}'-3\eps^2 p_j,l_{k+1}')$. 
\end{corollary}

We further analyze Phase 4. For any machine $i$, we define a number of time points: 
Let $x_{i,\kk}$ be the start time of the $\kk$-th long on-par job, $J$, that $\schedule^4$ starts on machine $i$. 
Note that since at time $0$ all jobs are on-par, $x_{i,1}=0$.
Furthermore, recall that $\ell$ is the index such that $C_J^4\in I_\ell'$. 
Define $y_{i,\kk}$ to be $l_{\ell+2}'$. If less than $h$ many long on-par jobs start on machine $i$ in $\schedule^4$,
define $x_{i,\kk}=y_{i,\kk}=\infty$. 
Let $j$ be the type of the interval that $J$ starts in, i.e., the interval of which $x_{i,\kk}$ is the left endpoint.
By Lemma~\ref{lem:ntypes-less-intervals} and~\ref{lem:ahead-volume-ph3}, 
and the construction of $\schedule^4$, since the interval that ${J}$ 
starts in under $\schedule^4$ contained a total volume of ahead jobs under $\schedule^3$ of less than $\eps p_j$ 
and thus also under $\schedule^1$, the following holds.
 
\begin{fact}\label{fact:n-yik}
    For any machine $i$, ahead (at $x_{i,\kk}$) type  $j$, and $\kk\geq 1$, 
    it holds that 
    \[
        V_{j}(f^{4\to1}(y_{i,\kk}),i,\schedule^1)-V_{j}(f^{4\to1}(x_{i,\kk}),i,\schedule^1)<\eps p_k\,,
    \] 
    where $k$ is the type of the interval of which $x_{i,\kk}$ is the left endpoint.
\end{fact}

For all time points $t$, all types $j$, and any machine $i$, define $V^*_{j}(t,i,\schedule^4)$ as the 
volume that $\schedule^4$ processed by $t$ of type-$j$ jobs that $\schedule^3$ processed on machine~$i$.

We have the following crucial lemma. 

\begin{lemma}\label{lem:n-phase4-volume}
    The following hold:
    \begin{enumerate}[label=(\roman*), nosep]
        \item For any machine $i$ and any time $t\geq 0$, the number of on-par type-$j$ jobs that start 
            on machine~$i$ by time $t$ under $\schedule^4$ is at least the number of type-$j$ jobs that start on 
            machine $i$ by time $f^{4\to 1}(t)$ under $\schedule^1$.
        \item For any $t\ge f^{1\to 4}(S_i^+(k))$, such that $t=l_\ell'$ for some $\ell\ge1$, and 
            $t\notin (x_{i,\kk},y_{i,\kk})$ for any $\kk\ge 1$, 
            and type $j$ that is ahead at~$p^*_k$, if $\schedule^4$ has not started all type-$j$ jobs by $t$, then 
            \[
                V^*_j(t,i,\schedule^4) > V_j(f^{4\rightarrow 1}(t),i,\schedule^1) + \eps p_k\,.
            \]
    \end{enumerate}
\end{lemma}
\begin{proof}

    Fix any machine $i$; we focus on that machine.
    Part (i) directly follows using that no on-par job is started later in $\schedule^4$ than 
    in $\schedule^3$ and Lemma~\ref{lem:n-types-phase3} (i).

    Define $y_{i,0}=0$. To prove (ii), we show by induction that, for all $\kk\ge 1$, for any 
    $t\in [y_{i,\kk-1},x_{i,\kk}]$, such that $t=l_\ell'$ for some $\ell\ge1$, and any type $j$, 
    we have 
    \[
        V^*_j(t,i,\schedule^4) = V_j(t,i,\schedule^3)\,.
    \]
    Then, the lemma follows from Lemma~\ref{lem:n-types-phase3}.

    For $\kk=1$, note that $t=y_{i,0}=x_{i,1}=0$ and the statement follows trivially. For the induction step, 
    notice that, for any $\kk$, within the interval $[x_{i,\kk},y_{i,\kk})$ we have that, for any type $j$, 
    \[
        V^*_j(y_{i,\kk},i,\schedule^4)-V^*_j(x_{i,\kk},i,\schedule^4) = V_j(y_{i,\kk},i,\schedule^3) - 
        V_j(x_{i,\kk},i,\schedule^3)\,,
    \]
    by construction. Therefore, for $t=y_{i,\kk}$, 
    \[
        V^*_j(t,i,\schedule^4) = V_j(t,i,\schedule^3)\,.
    \]
    Furthermore, for any $[l_\ell',l_{\ell+1}')\subseteq [y_{i,\kk},x_{i,\kk+1}]$ and any type $j$ we have 
    \[
        V^*_j(l_{\ell+1}',i,\schedule^4) - V^*_j(l_{\ell}',i,\schedule^4)= 
            V_j(l_{\ell+1}',i,\schedule^3) - V_j(l_\ell',i,\schedule^3)
    \]
    by construction, and, since $y_{i,\kk}$ and $x_{i,\kk+1}$ are endpoints 
    of intervals, the statement follows.
\end{proof}

Lemma~\ref{lem:n-phase4-volume} allows us to show the following central lemma.
	
\begin{lemma}\label{lem:n-types-phase4-feasibility}
    Stochastic schedule $\schedule^4$ is feasible, there exists a non-anticipatory policy $\Pi^4$ 
    such that $S(\Pi^4)=\schedule^4$, and for each job $J$ it holds that $S_J^4\le (1+5\eps)S_J^1$.
\end{lemma}
\begin{proof}
    We first show feasibility of $\schedule^4$. Note that, by Lemma~\ref{lem:n-types-phase3} and the 
    fact that $\schedule^4$ is identical to $\schedule^3$ within $I_0'$, we only have to make two observations.
    
    The first observation is that the execution intervals of on-par jobs started on 
    some machine~$i$ at or after $f^{2\to 4}(\Splus{i}{n})$
     are pairwise non-overlapping: Consider 
    two on-par jobs $J$ and $J'$ that start on the same machine $i$ under $\schedule^3$ 
    (and therefore also under $\schedule^4$), such that $S_J^4\in I_k'$ and $S_{J'}^4\in 
    I_\ell'$. If $k=\ell$, then one of $J,J'$ must be short and is started no later than the 
    other one in both $\schedule^3$ and $\schedule^4$. If $k\neq \ell$, w.l.o.g.\ let 
    $k<\ell$. By Lemma~\ref{lem:ntypes-less-intervals}, we know that $C_J^3\in I_z$ with 
    $z<\ell$ and thus also $C_J^4<S_{J'}^4$, so the execution intervals are indeed non-overlapping.
    
    The second observation is that all spaces created in the 
    description of $\schedule^4$ do not interfere with another space or execution 
    interval of an on-par job, this is by construction and Lemma~\ref{lem:n-types-phase4-nocrossing}. 
    Indeed, consider any long on-par job scheduled on machine $i$ 
    in $\schedule^4$ with $S_J^4\geq f^{2\to 4}(\Splus{i}{j})$. Then, 
    Lemma~\ref{lem:ntypes-less-intervals} implies that no job is started in $\schedule^3$ 
    during $I_\ell'\ni C^3_j$ and $I_{\ell+1}'$. Hence, the spaces created within $I_{\ell+1}'$ in 
    $\schedule^4$ can indeed all be created without conflicts.
    
    So overall, the execution intervals of no pair of jobs intersects each other. Furthermore, 
    by construction, $\schedule^4$ starts all jobs. Therefore $\schedule^4$ is feasible.
    
    We next show that there exists a non-anticipatory policy $\Pi^4$ such that 
    $S(\Pi^4)=\schedule^4$. First recall that we have only established an analogous statement 
    for Phase 1 (in Lemma~\ref{lem:ntypes-ph1-main}) and not for Phases 2 and 3. Also recall that 
    $\schedule^1$, $\schedule^2$, $\schedule^3$, and $\schedule^4$ all start jobs identically 
    during $I_0$, so $\Pi^4$ can simply simulate $\Pi^1$ during that interval. Further, 
    during $I_0'\setminus I_0$, $\schedule^4$ does not start any jobs.
    
    We prove the argument by induction on $j\in \{n, \dots, 1\}$. 
    By slight abuse of notation, assume that there exists a policy $\policy^4$ such that $S(\policy^4)=\schedule^4$ 
    on the interval $[0, f^{2\to 4}(p_j^*))$. Note that the description of $\policy^4$
    for $I_0'$ above, implies that such a policy exists up to $f^{2\to 4}(p_{n}^*)$.
    
    First we note that $\schedule^4$ does not start any jobs in $[f^{2\to 4}(p_{j}^*),f^{2\to 4}(\Cplus{i}{j}))$.
    We now only consider machines $i$ such that $f^{2\to 4}(\Cplus{i}{j})<f^{2\to 4}(p_{j-1}^*)$, 
    since on machines for which this is not the case, no jobs are started under $\schedule^4$ 
    within $[f^{2\to 4}(p_{j}^*),f^{2\to 4}(p_{j-1}^*))$.
    
    Next, the left endpoint of any space in $\schedule^4$ within the interval 
    $[f^{2\to 4}(\Cplus{i}{j}), f^{2\to 4}(\Splus{i}{j})$ on any machine $i$ (those that stem from Phase 2) 
    only depends on $\Cplus{i}{j}$, which in turn only depends on the processing times of jobs started earlier. 
    Hence, for any of these spaces, the non-anticipatory policy $\Pi^4$ can fill them just as described in $\schedule^4$.
    
    It remains to define $\Pi^4$ such that it starts the remaining jobs, i.e., those jobs that are started in $\schedule^4$ on some machine $i$ within $[f^{2\to 4}(\Splus{i}{j}),f^{2\to 4}(p_{j-1}^*))$. Consider any interval $I_k'$ of type $j$ that intersects $[f^{2\to 4}(\Splus{i}{j}),f^{2\to 4}(p_{j-1}^*))$. 
    We make two claims:
    \begin{enumerate}[label=(\roman*), nosep]
        \item All on-par jobs that start before time $l_k$ in $\schedule^1$ start before time $l_k'$ in $\schedule^4$.
        \item All ahead jobs that start before time $l_{k+1}$ in $\schedule^1$ start before time $l_k'$ in $\schedule^4$.
    \end{enumerate}
    Indeed, (i) is implied by Lemma~\ref{lem:n-phase4-volume} (i). To see (ii), 
    first note that it is sufficient to prove that, for each ahead type $j'$, we have
    \[
        V_{j'}^*(l_k',i,\schedule^4)> V_{j'}(l_{k+1},i,\schedule^1)\,.
    \]
    We consider two cases, either $l_{k}'\notin(x_{i,z},y_{i,z})$ for all $z\geq 1$, or 
    there exists a $z\ge1$ such that $l_{k}'\in(x_{i,z},y_{i,z})$.
    \begin{itemize}
        \item If $l_{k}'\notin(x_{i,z},y_{i,z})$ for all $z\geq 1$, then
        by Lemma~\ref{lem:n-phase4-volume} (ii) and the fact that $|I_k|= \eps p_j$, 
        we have that $V_{j'}^*(l_k',i,\schedule^4)> V_{j'}(l_{k+1},i,\schedule^1)$.
        \item If $l_k'\in(x_{i,z},y_{i,z})$ for some $z\geq 1$, then $l_{k+1}' < y_{i,z}$ since $y_{i,z}$ 
        is the endpoint of some interval $I_\ell'$. Hence, Lemma~\ref{lem:n-phase4-volume} (ii), 
        applied to $x_{i,z}$,
        combined with Fact~\ref{fact:n-yik} implies that $V_{j'}^*(l_k',i,\schedule^4)
        >V_{j'}(l_{k+1},i,\schedule^1)$.
    \end{itemize}
    
    Now, at time $l_k$, (i) and (ii) allow $\Pi^4$ to simulate $\Pi^1$ in the background in order to start 
    jobs during $I_k$: Any on-par jobs that the simulation of $\Pi^1$ starts on some machine $\Pi^4$ starts 
    in the same global order on the same machine at time $l_k'$ (in accordance with $\schedule^4$), by which 
    it learns the realization, allowing it to continue the simulation.
    As a result of the simulation, $\Pi^4$ discovers $n_j(\mathcal{J}_k^a(i))$ for all machines $i$. 
    It can therefore fill the corresponding spaces, which are described in the definition of $\schedule^4$. 
    Note that some of these spaces are only assigned to a later interval than $I_k'$.
    It remains to argue for interval $[f^{2\to 4}(p_1^*),\infty)$, but the inductive step argument can also 
    be used identically for this interval.
    
    Finally, regarding the start times, we use the fact that $\schedule^1$ and $\schedule^4$ are identical 
    within $I_0$ (resp. $I_0'$) as well as (i) and (ii) (for on-par and ahead jobs, respectively) to get that, 
    for any job $J$, indeed $S_J^4\leq (1+5\eps)S_J^1$.
\end{proof}

By construction we have the following property.

\begin{lemma}\label{lem:n-types-ph5-starttimes}
    Consider a long on-par job $J$ on machine $i$ with $C_J^4\in I_k'$. 
    Then $\policy^4$ starts no jobs on $i$ during $I_k'$.
\end{lemma}

We are now ready to define the sets of time points $Q_1,\dots, Q_n$, which are the time points at which a stratified 
policy can start jobs of the respective types. They are needed to define the dynamic program to solve for an 
optimal stratified policy. 
\begin{definition}
    Let $p_j^\circ = f^{2\to 4}(p_j^*)$ and note that, 
    since for some $k\ge 1$, 
    $p_j^*=l_k$, for that same~$k$, $p_j^\circ=l_k'$, the left endpoint of an interval $I_k'$.
    
    Now, let
    \[
        Q_1' = \left\{l_k' \ge p_1^\circ\,\middle|\, k\in\N\right\}\,,
    \]
    for all $j\in\{2,\dots,n-1\}$, 
    \[
        Q_j' = \left\{l_k' \in [p_j^\circ,p_{j-1}^\circ )\,\middle|\, k\in\N \right\}
                \cup\left\{l_k' + i\eps p_{j} \,\middle|\, (k,i\in\N) \wedge (l_k'\ge p_{j-1}^\circ) 
                \wedge ( l_k' + i\eps p_{j}<l_{k+1}'-p_j)\right\} \,,
    \]
    and
    \begin{align*}
        Q_n' &= \{0\}\cup\left\{l_k' \in [p_n^\circ,p_{n-1}^\circ )\,\middle|\, k\in\N \right\}
                \\&\qquad\cup\left\{l_k' + i\eps p_{n} \,\middle|\, (k,i\in\N) \wedge (l_k'\ge p_{n-1}^\circ) 
                \wedge ( l_k' + i\eps p_{n}<l_{k+1}'-p_j)\right\} \,.
    \end{align*}
    We are now ready to define the sets $Q_j$ recursively. Let
    \begin{align*}
        Q_j = \begin{cases}
            Q_j'&\text{for $j=n$,}\\
            Q_j'\cup ( Q_{j+1} \cap [0,p_j^\circ)) &\text{for $j\in \{1,\dots, \numtypes-1\}$}\,.
        \end{cases}
    \end{align*}
\end{definition}
Within the definition of stratified policies and in Section~\ref{sec:dp}, 
$Q_j$ is the set of allowed start times for jobs of type $j$. To ensure that
the dynamic programming algorithm of Section~\ref{sec:dp} works as intended, note that 
the following holds.
\begin{restatable}{lemma}{lemmantypesQjs}\label{lem:ntypes-Qjs}
    Let $t\in Q_j$. If there is a $t'\in Q_{j-1}$ such that $t<t'<t+p_j$, then $t\in Q_{j-1}$.
\end{restatable}
Note that, by construction of $\schedule^4$, any on-par type-$j$ job starts 
at a timepoint in $Q_j$. In fact, it is not difficult to see that the ahead 
jobs of type $j$ also start at timepoints in $Q_j$. 
\begin{restatable}{lemma}{lemmantypegridaligned}\label{lem:ntypes_grid-aligned}
    In $\schedule^4$, any job of type $j$ starts at a timepoint in $Q_j$.
\end{restatable}

\begin{restatable}{lemma}{lemmantypesstratified}\label{lem:n-types-stratified}
    $\schedule^4$ is stratified with respect to $(Q_1,\dots,Q_\numtypes)$ and $(p^\circ_1,\dots,p^\circ_\numtypes)$.
\end{restatable}

\subsection{Putting all phases together}
Combining the transformations described in the four phases (Lemmas~\ref{lem:ntypes-ph1-main} 
and Lemma~\ref{lem:n-types-phase4-feasibility}), we obtain the following structural result.
\begin{theorem}\label{thm:job-by-job-bound-n-types}
    Let $\eps \le 1/13$. Consider an input where jobs come from $n$ distinct types with size 
    parameters $p_1>\dots>p_n$ and $p_{j+1}\le \eps^2p_j$ for all $j\in \{1,\dots,n-1\}$. 
    For any non-anticipatory and non-idling policy $\Pi$ there exists a corresponding 
    non-anticipatory stratified policy $\Pi'$ such that for each realization on which the 
    policies produce respective schedules $\schedule,\schedule'$, for each job $J$, we have 
    \begin{align*}
        S_J' \le (1+(2n+4)(1+\eps)\eps)(1+5\eps) S_J = (1+O(\eps))S_J\,.
    \end{align*}
\end{theorem}

\section{Grouping Types with Comparable Sizes}\label{sec:groups}
Now, we argue how to get rid of  assumption~\eqref{eq:separation_of_s_j} which separates the size 
parameters of the different job types. Given an arbitrary instance with constantly many size 
parameters $p_1>\dots > p_n$, we define ``size groups'' $G_1,\dots,G_\numbergroups$ as follows. The 
largest size group $G_1$ contains  types with large sizes, namely the largest type $p_1$, 
and recursively all types of size $p_\ell$, for $\ell\in\{2,3,\dots\}$, as long as $p_{\ell}>
{\eps^2}p_{\ell-1}$. Suppose this process ends with putting type $p_{k}$ into $G_1$, so that 
$p_{k+1}\le{\eps^2}p_{k}$. Then we put type $p_{k+1}$ into the next largest size group 
$G_{2}$, and again fill $G_2$ recursively in the same fashion. 
This results in $\numbergroups\le n$ many size 
groups $G_1,\dots,G_\numbergroups$, so that the size parameters of any two jobs in consecutive groups 
are separated by a factor at least $\eps^{-2}$.  Note that $n=|G_1|+\dots+|G_\numbergroups|$.
Also note that the size parameters of all types within a group are within a constant factor of 
each other, since we still assume the total number of different sizes is bounded by a constant. 
In particular, $\numbergroups=1$ is possible, in which case 
the size parameters of all $\numberjobs$ jobs are within a factor~$\eps^{-2(\numtypes-1)}$.
We will subsequently refer to the fact that the size parameters within groups are within a 
constant factor of each other, and that they are sufficiently separated across groups, as follows.
\begin{align}
    \label{eq:p_j_within_groups}\text{For any two types $j,j'\in G_h$, $j<j'$:} &&p_j \le&\ \eps^{-2(|G_h|-1)}p_{j'}\\
    \label{eq:p_j_across_groups}\text{For any two types $j\in G_h, j'\in G_{h'}$, $h<h'$:} &&p_j \ge&\ \eps^{-2}p_{j'}
\end{align}

In the next section, we describe the dynamic program based on job types that are grouped 
as described above. The idea is that all jobs with types in group $G_h$, 
even though of possibly different sizes, are only allowed to be started at all 
time points which are feasible for the \emph{smallest} jobs within~$G_h$. 
In fact, it is important to realize that each job type  
$j\in\{1,\dots,n\}$ is still treated separately by the dynamic program. 
Denote by $p_{G_h}$ the smallest of the size parameters of all types in 
group $G_h$, for $h\in\{1,\dots,\numbergroups\}$.
In particular, by \eqref{eq:p_j_across_groups}, we have that $p_{G_1} \ge \eps^{-2} p_{G_2}\ge \dots \ge 
\eps^{-2(\numbergroups-1)} p_{G_{\numbergroups}}= \eps^{-2(\numbergroups-1)} p_n$. 
I.e., the separation of sizes 
that we assumed in \eqref{eq:separation_of_s_j}, holds for the size 
representatives.
In addition to the notation $G_h$ for groups, we use $G(j)$ to denote the 
group that jobs of type $j$ belong to.

Note that, unlike in the $\eps^2$-separated case (Lemma~\ref{lem:ntypes-divisible}), 
we cannot assume anymore that all sizes divide the larger sizes without significant loss. 
Instead, we make the following similar assumptions 
that are w.l.o.g.\ at a loss of a $(1+\eps)$ factor in the objective function.
\begin{restatable}{lemma}{lemmadivisibilityforgroups}\label{lem:groups-divisible}
    At an increase in expected cost of a factor at most $(1+\eps)$, one may assume that 
    \begin{enumerate}[label=(\roman*),nosep]
        \item $p_{G_h}$ divides $p_{G_{h-1}}$, for all $h\in\{2,\dots,\numbergroups\}$, and
        \item $\eps p_{G_h}$ divides $p_j$, for all $j\in G_h$.
    \end{enumerate}
\end{restatable}
In order to (re)define stratified policies, we redefine the sets of time points that 
are feasible for all jobs within a group of jobs $G_h$. 
We construct $Q_{G_h}$ for type groups $G_h$ similar to how we constructed $Q_j$ 
for types $j$, with the following notes.
\begin{itemize}
    \item All types within a type group share the same set of time points.
    \item The threshold point $p_{G_h}^*$ and the interval lengths between points $l_k'$ 
        and $l_{k+1}'$, as well as the points in $Q_{G_h}$ between $l_k'$ and $l_{k+1}'$ are 
        defined based on the size $p_{G_h}$. In particular, the points in $Q_{G_h}$ between $l_k'$ 
        and $l_{k+1}'$ are spaced $\eps p_{G_h}$ apart, where the last interval before $l_{k+1}'$ 
        is longer: The length-$\eps p_{G_h}$ intervals are only defined as long as their right 
        endpoint is smaller than $l_{k+1}'$ minus the largest processing time in $G_h$; 
        see also Phase 4 in Appendix~\ref{sec:app_groups}.
\end{itemize}
This implies that, in the context of the structural theorem, all types in one type group
transition from on-par to ahead at the same time point. At that time point, we schedule 
an interval for each such type to ensure that the policy is sufficiently 
ahead for all those types. We stop scheduling types from a type group if there exists
a type from that group that would cross a time point from the next type group with larger 
sizes if its processing time is long.

\begin{definition}[Stratified policy for grouped types]\label{def:group_well-formedness}
    Given size groups $G_1,\dots,G_\numbergroups$, 
    point sets $Q_{G_\numbergroups} \supseteq \dots \supseteq Q_{G_1}\ni 0$,
    and \emph{threshold points} $p_{G_1}^\circ\in Q_{G_1}, \dots, p_{G_\numbergroups}^\circ\in Q_{G_\numbergroups}$,
    we say policy $\Pi$ is \emph{stratified with respect to $(Q_{G_1},\dots,Q_{G_\numbergroups})$ 
    and $(p^\circ_{G_1},\dots,p^\circ_{G_\numbergroups})$} if it satisfies the following properties.
    \begin{enumerate}[label=(\roman*), noitemsep]
        \item If $\Pi$ starts a job $J$ of type $j$ at time $t$, then $t\in Q_{G(j)}$. 
        \item If $\Pi$ starts a job $J$ of type $j$ and $J$ is long, 
            then $\Pi$ starts no job on the same machine between $C_J$, the completion time of $J$, 
            and $t' :=\arg\min_{t\in Q_{G(j)}} \{t\ge \max\{p_{G(j)}^\circ,C_J \}\}$, the next time point from $Q_{G(j)}$ 
            after~$p_{G(j)}^\circ$.
        \item $\Pi$ starts jobs of the same type in order of index.
    \end{enumerate}
\end{definition}

The grouping of job sizes requires to modify some technical claims that we made earlier 
in the paper, namely whenever these claims were based on the $\eps^2$ separation of size 
parameters, because this separation is now only true for the size parameters of different 
groups, but not within a group. The correspondingly modified claims and proofs are contained 
in Appendix~\ref{sec:app_groups}. We also need accordingly adapted  time points $Q$ and 
threshold points $p^\circ$, which we give here for later reference.
\begin{definition}\label{def:Q_groups}
    We let the thresholds be 
    $p_{G_h}^\circ = f^{2\to 4}(p_{G_h}^*)$ for all $h=1,\dots\numbergroups$, and
    \[
        Q_{G_1}' = \left\{l_k' \ge p_{G_1}^\circ\,\middle|\, k\in\N\right\}\,.
    \]
    For all $h\in\{2,\dots,\numbergroups-1\}$, let
    \begin{align*}
        Q_{G_h}' &= \left\{l_k' \in [p_{G_h}^\circ,p_{G_{h-1}}^\circ )\,\middle|\, k\in\N \right\}
                \\&\qquad\cup\left\{l_k' + i\eps p_{G_h} \,\middle|\, (k,i\in\N) \wedge (l_k'\ge p_{G_{h-1}}^\circ) 
                \wedge ( l_k' + i\eps p_{G_h}<l_{k+1}'-\pmax{G_h})\right\} \,,
    \end{align*}
    and 
    \begin{align*}
        Q_{G_\numbergroups}' &= \{0\}\cup\left\{l_k' \in [p_{G_\numbergroups}^\circ,
        p_{G_{\numbergroups-1}}^\circ )\,\middle|\, k\in\N \right\}\\
            &\qquad\cup\left\{l_k' + i\eps p_{G_\numbergroups} \,\middle|\, (k,i\in\N) 
            \wedge (l_k'\ge p_{G_{\numbergroups-1}}^\circ) 
            \wedge ( l_k' + i\eps p_{G_\numbergroups}<l_{k+1}'-\pmax{G_\numbergroups})\right\} \,.
    \end{align*}
    We again define the sets $Q_{G_h}$ recursively. Let
    \begin{align*}
        Q_{G_h} = \begin{cases}
            Q_{G_h}'&\text{for $h=\numbergroups$,}\\
            Q_{G_h}'\cup ( Q_{G_{h+1}} \cap [0,p_{G_h}^\circ)) &\text{for $h\in \{\numbergroups-1,\dots,1\}$}\,.
        \end{cases}
    \end{align*}
\end{definition}

The grouping of types results in the equivalent of Theorem~\ref{thm:job-by-job-bound-n-types}, 
but now without the condition that the sizes be $\eps^2$ separated.
\begin{theorem}\label{thm:job-by-job-bound-groups}
    Let $\eps \le 1/13$. Consider an input where jobs come from $n$ distinct types with size 
    parameters $p_1>\dots>p_n$. For any non-anticipatory and non-idling policy $\Pi$ there 
    exists a corresponding non-anticipatory stratified policy $\Pi'$ such that for each 
    realization on which the policies produce respective schedules $\schedule,\schedule'$, 
    for each job $J$, we have 
    \begin{align*}
        S_J' \le (1+(2n+4)(1+\eps)\eps)(1+5\eps) S_J = (1+O(\eps))S_J\,.
    \end{align*}
\end{theorem}

\section{Dynamic Programming Algorithm}\label{sec:dp}
We use most of the same notation as in Section~\ref{sec:notation} for the description 
of the dynamic programming algorithm (DP). The main difference between what we describe 
here and the approach described in Section~\ref{sec:notation} is that the decisions 
of the DP are organized along what we refer to as ``relevant'' points in time, which 
is a subset of $Q$. The decision to be made at such point in time~$t$, given the 
state of the schedule at time $t$, is described using the number of jobs of each type 
that have not yet been started and a profile $\vprofile$ of all machine loads. 
In contrast to the simple DP of Section~\ref{sec:notation}, we now let $\vprofile$ be 
lexicographically ordered, i.e., we do not specify which machine has a particular load. 
Moreover, while we still call them ``machine loads'', this is technically no longer 
the correct term, as the schedule also includes idle times. 
Specifically, $m_1$ equals the smallest completion time of any machine in the current 
(partial) schedule, $m_2$ equals the second smallest completion, etc. 
Here the completion time of a machine is the time at which it processed all jobs 
scheduled on that machine at a decision time $t$ including any scheduled idle time.
In particular, if idle time is scheduled after the completion of the last job, this is included as well. 
Of note is that the completion time of each machine is non-anticipatorily known at 
time $t$, since at that time all jobs scheduled on the machine have been started, 
and potential idle time after the completion of the last job is known through the 
definition of stratified policies (Definition~\ref{lem:n-types-stratified}.
Regardless, for simplicity we refer to vector $\vprofile$ as \emph{machine load profile}.

Given a vector of machine loads $\vprofile$, denote by 
\[
    t^*:=t^*(\vprofile):=\min_\ell \profile_\ell=m_1
\] 
the earliest 
point in time that a machine is available to start processing a job. 
Unlike in Section~\ref{sec:notation}, we now denote by~$\vjobsleft$ the vector of leftover 
jobs per type, i.e., $\jobsleft_j$ denotes the number of jobs of type~$j$ that not yet have  
started processing. Now, $\ecost{\vprofile,\vjobsleft}$
denotes the minimum expected sum of completion times to schedule unscheduled jobs $\vjobsleft$, 
given the current machine loads~$\vprofile$. 

Recall that jobs $J,J'$ of the same type $j$ can have probabilities $q_J\neq q_{J'}$. 
In view of this and Lemma~\ref{lem:identical_s_J}, denote by $\Prob{j}(\jobsleft_j)$ 
the smallest probability $q_J$ among the $\jobsleft_j$ unscheduled jobs of type~$j$.

To decrease the size of the dynamic program DP to polynomial size, we redefine the 
basic DP from Section~\ref{sec:notation} to only construct stratified policies 
following Definition~\ref{def:group_well-formedness}. 
This restricts which job types can start at any particular time $t^*$ (if any). 
We argue that this reduces the number of machine load profiles that need to be 
considered to polynomially many.

The DP only starts type-$j$ jobs at times $t\in Q_{G(j)}$.
We denote by 
\[
    \timedtypes{t} := \set*{j \in \{1,\dots,n\}\given t\in Q_{G(j)}} 
\]
the set of job types for which time $t$ is an allowed start time.

The output of the DP is $\ecost{\vzero,\vnumberofjobspertype}$. 
Computing that value requires recursively computing values 
$\ecost{\vprofile,\vjobsleft}$ that are necessary for this. 
Finding the necessary values can be done, e.g., by a simple 
breath-first or depth-first search of the decision tree.
We refer to these the parameters that correspond to one of these 
necessary DP cells as \emph{relevant} parameter values. We implicitly 
let $\ecost{\vprofile,\vjobsleft}=\infty$ for all non-relevant values of 
$\vprofile,\vjobsleft$, or, in fact, just ignore them.
The base case is $\ecost{\vprofile,\vzero} = 0$, for any $\vprofile$. 
The recursion for the dynamic program is as follows. 
\begin{equation}\label{eq:DPrecursion}
    \ecost{\vprofile,\vjobsleft} = 
    \begin{cases}
        \ecost{\vprofile^0,\vjobsleft} & \text{if $\nu_j = 0$ for all ${j\in \timedtypes{t^*}}$,}\\
        \begin{aligned}
            \min_{j\in \timedtypes{t^*}}\Bigg[\Prob{j}(\jobsleft_j)
        \biggl(\ecost{\vprofile^j,\vjobsleft-\ve_j}+(t^*+p_j)\biggr) \\
        {} + (1-\Prob{j}(\jobsleft_j))\biggl(\ecost{\vprofile,\vjobsleft-\ve_j}+t^*\biggr)\Bigg]
        \end{aligned}
         & \text{otherwise.}
    \end{cases}
\end{equation}

Here, $\vprofile^0$ is the updated machine-load profile when no job is 
started at time $t^*$. This happens when none of the available job types 
can be started at time $t^*$ without violating the conditions of a stratified policy, 
i.e., when all jobs of types in $\timedtypes{t^*}$ have already been scheduled. 
In that case, the machine load profile is updated to $\vprofile^0$ by 
inserting additional idle time on the least loaded machine(s), 
as explained below in Definition~\ref{def:profileupdate}. 
In case a job~$J$ of type $j$ is started, the number $\jobsleft_j$ of remaining 
jobs of type $j$ is reduced by $1$, so $\vjobsleft$ is replaced by 
$\vjobsleft-\ve_j$, where $\ve_j$ denotes the unit-length vector with~$j$-th 
coordinate equal to $1$. If the started job of type $j$ has non-zero processing time, 
$\vprofile^j$ is the updated machine load profile (see 
Definition~\ref{def:profileupdate}) and $(t^* +p_j)$ the completion time. This 
happens with probability $\Prob{j}(\jobsleft_j)$. If the started job has processing 
time zero, the completion time is $t^*$ and the machine load profile is left 
unchanged. This happens with probability $1-\Prob{j}(\jobsleft_j)$. 

First, for an \emph{optimal} stratified policy with respect to 
$(Q_{G_1},\dots,Q_{G_\numbergroups})$ and $(p^\circ_{G_1},\dots,p^\circ_{G_\numbergroups})$,
we prove that it does not start a job at time $t^*$ if and only if 
$\nu_j = 0$ for all ${j\in \timedtypes{t^*}}$ or all machines are busy. Then, we
discuss how $\vprofile^j$ and $\vprofile^0$ are derived from $\vprofile$. What remains then, 
is to bound the total computation time.

The following lemma is the equivalent of Lemma~\ref{lem:non-idling} for optimal stratified policies. 
\begin{restatable}{lemma}{lemmadpnonidling}\label{lem:dp-nonidling}
    At any time $t\in Q$, any optimal stratified policy $\Pi'$ with respect to $(Q_{G_1},\dots,Q_{G_\numbergroups})$ and 
    $(p^\circ_{G_1},\dots,p^\circ_{G_\numbergroups})$,  starts  a job of a type in 
    $\timedtypes{t}$ until all machines are busy or no jobs of any type in $\timedtypes{t}$ are left.
\end{restatable}

\begin{definition}[$\vprofile^j$ and $\vprofile^0$] \label{def:profileupdate}
    Machine-load profile $\vprofile^j$ is the updated machine load profile when 
    the started job of type $j$ has non-zero size, and 
    $\vprofile^j$ is constructed by replacing $\vprofile_1$ with $\vprofile_1+p_j$. 
    Then, to comply with Definition~\ref{def:well-formedness}, rounding 
    up to the next smallest time point $t\ge \vprofile_1+p_j$, $t\in Q_{G(j)}$,
    and finally reordering the resulting machine load vector lexicographically. 
    
    Machine load profile $\vprofile^0$ is the updated machine load profile when no job is started. 
    Let $j^* = \max \set*{j \given \nu_j> 0 }$ be the index of the smallest size of the remaining job types,
    and denote by $t(j^*) = \min \set*{t \in Q_{G(j^*)}\given t \ge t^*(\vprofile)}$  
    the next time point when a job can be started. Note $t(j^*)> t^*(\vprofile)$.
    We construct $\vprofile^0$ by setting $\vprofile^0_i = t(j^*)$ for all machines 
    $i$ such that $\vprofile_i< t(j^*)$ and 
    $\vprofile^0_i = \vprofile_i$ otherwise (no reordering is necessary).
\end{definition}

\begin{definition}[Relevant parameter values]
    We say that a combination of values $(\vprofile,\vjobsleft)$ is \emph{relevant} if 
    that combination of  values is necessary to recursively compute $\ecost{\vzero,\vnumberofjobspertype}$.
    Additionally, we say that a time point $t$ is relevant if $t=t^*(\vprofile)$ for 
    some relevant $(\vprofile,\vjobsleft)$.
\end{definition}

It may seem that the definition of relevant time point depends on the relevant machine 
loads. However, both are in fact governed by the recursion that computes $\ecost{\vzero,\vnumberofjobspertype}$. 
To bound the overall computation time, 
we divide the proof into three lemmas. 
First the number of vectors of leftover jobs is easily bounded. Then we bound 
the number of relevant time points $t^*$. Finally, we bound the number of relevant 
machine load profiles $\vprofile$ with $t^*=t^*(\vprofile)$. Combining these three lemmas, 
gives us the desired bound on the total number of 
relevant values $(\vprofile,\vjobsleft)$.  

\begin{restatable}{lemma}{lemmarelevantjobsleft}\label{lem:relevant_jobsleft}
    There are $\bigO{\numberjobs^\numtypes}$  vectors of leftover jobs per type, $\vjobsleft$.
\end{restatable}

Note that it is generally not true that the number of \emph{all} time points in $Q$ can 
be polynomially bounded. Therefore, we bound the number of \emph{relevant} time points in $Q$. 
\begin{restatable}{lemma}{lemmarelevanttimepoints}\label{lem:relevant_time_points}
    There are at most $\numberjobs^\numtypes\eps^{-2\numtypes}$ relevant time points in $Q$.
\end{restatable}

We next argue about the total number of relevant machine load profiles $\vprofile$, 
and count them by their corresponding (relevant) time points $t^*=t^*(\vprofile)$.
For convenience in what follows, let us define a parameter for the group sizes.
\begin{definition}\label{def:z}
    Let $\maxgroupsize := \max_{\ell\in\{1,\dots,\numbergroups\}} |G_\ell|$ be the maximum size of any of the groups. 
\end{definition}
Note that $z\le n$ is a constant, and we have $z=1$ for the case that the job 
sizes are separated as in \eqref{eq:separation_of_s_j}. 

\begin{restatable}{lemma}{lemmarelevantprofiles}\label{lem:relevant_profiles}    
    For each relevant time point $t^*$ in $Q$ there are $\BigO{m^{n\eps^{-2\maxgroupsize}}}$ 
    relevant machine load profiles $\vprofile$ so that $t^*=t^*(\vprofile)$. 
\end{restatable}

Bringing Lemma~\ref{lem:relevant_jobsleft}, Lemma~\ref{lem:relevant_time_points} 
and Lemma~\ref{lem:relevant_profiles} together, 
we conclude that the number of different parameter values of the dynamic program 
that need to be computed to calculate $\ecost{\vzero,\vnumberofjobspertype}$ 
is bounded by a polynomial in $\numberjobs$ and $m$, if $n\in\bigO{1}$, 
since $\maxgroupsize\le n$. 
\begin{lemma}\label{lem:countDPCells}
    The number of parameter values $(\vprofile,\vjobsleft)$ in the dynamic program that 
    need to be computed to recursively calculate $\ecost{\vzero,\vnumberofjobspertype}$
    is bounded by $\BigO{\numberjobs^{2\numtypes}\eps^{-2\numtypes}m^{n\eps^{-2\maxgroupsize}}}$.
\end{lemma}
We conclude with the main theorem of this section.
\begin{restatable}{theorem}{theoremoptimalstratified}\label{thm:optimal_stratified}
    An optimal stratified scheduling policy according to Definition~\ref{def:group_well-formedness}, 
    and with the time points $Q_{G_1}$, \dots,  $Q_{G_\numbergroups}$ and 
    thresholds $p_{G_1}^\circ$, \dots, $p_{G_\numbergroups}^\circ$ as in 
    Definition~\ref{def:Q_groups}, can be computed in polynomial time, 
    given that the number $n$ of different size parameters $p_j$ is constant.
\end{restatable}

We conclude with a proof of the main theorem of this paper.
\mainthmboundedn*
\begin{proof}
    By Theorem~\ref{thm:job-by-job-bound-groups}, for any $\eps^*>0$, there 
    exists a stratified scheduling policy with performance guarantee $(1+\eps^*)$. 
    Theorem~\ref{thm:optimal_stratified} shows that an optimal stratified 
    scheduling policy can be computed in polynomial time, when the number $n$ 
    of different size parameters $p_j$ is constant.
\end{proof}

\section{Unbounding the Number of Size Parameters}\label{sec:unboundedtypes}
Here we argue how the PTAS for a constant number of job types can be used to 
obtain a \bigO{\log N}-approximation algorithm that runs in quasi-polynomial 
time for the problem without a bound on the number of job types. 

So consider an arbitrary instance for stochastic scheduling with Bernoulli jobs without restrictions on the number of different job types. As a first first step, following \cite{GMZ2023}, one can reduce this problem to a problem with only \bigO{\log N} many types with corresponding sizes that are all powers of some constant $c\ge 2$,
by simply rounding size parameters up to the next largest $c^k$, for integer $k\ge 1$, while losing only a constant factor in the approximation guarantee. 
We refer to Lemma~\ref{lem:log-n-many-sizes} in the appendix.  
Second, we argue how to obtain a \bigO{\log N}-approximation algorithm for the case with $n=\bigO{\log N}$ many types which are all powers of $c$.
To that end, first observe that choosing 
$\eps=1/13$ in \eqref{eq:separation_of_s_j} is a feasible choice for all earlier proofs and technical lemmas. That fixed, we can choose the appropriate constant $c=13^2$, so that any two distinct size parameters $p_j\neq p_k$ are separated by a factor $\ge 1/\eps^2$, because all job sizes are powers of $c$. Consequently we have that assumption \eqref{eq:separation_of_s_j} holds because all sizes are powers of $c$. Hence there is no need for grouping types of comparable sizes as in Section~\ref{sec:groups}, respectively, the \bigO{\log N} many type groups $G_i$ contain only one type each, so $|G_i|=1$, and $z=1$ in Definition~\ref{def:z}.
Given that $\eps=1/13$ is a constant, all this implies that Lemma~\ref{lem:countDPCells}
limits the number of parameter values for the dynamic program to \bigO{(Nm/\eps)^{\bigO{\log N}}}, hence Theorem~\ref{thm:optimal_stratified} implies that we can compute an optimal stratified scheduling policy in quasi-polynomial time, given that there are \bigO{\log N} many types with size parameters that are powers of $c$.

As to the performance guarantee for the stratified scheduling policy per Theorem~\ref{thm:job-by-job-bound-n-types}, 
recall that it is $1+\bigO{n\eps}=\bigO{\log N}$, which is the guarantee that we now have for any instance with \bigO{\log N} many types with size parameters that are powers of $c$.
Finally using Lemma~\ref{lem:log-n-many-sizes} with $c=13^2$, we obtain the existence of a policy that can be computed in quasi-polynomial time for an instance without any restriction on the number of size parameters, and with performance guarantee  $13^4+13^2\cdot\bigO{\log N}=\bigO{\log N}$. This proves our second main theorem, which we repeat here.

\mainthmunboundedn*

\section{Conclusion}
Even though the conceptual idea underlying our results is simple, the technical details to work them out for the most general case are not. Indeed, the major complications show up when moving from two job types to a larger number of job types. This happens because of the interplay of three issues, namely (i) the allowed time points for job types being nested in a nontrivial way, (ii) the necessity to schedule ``enough volume'' in the corresponding intervals per job type, and (iii) the necessity to keep the modified optimal policy non-anticipatory. Although simplifications are certainly achievable, we currently do not see how to move to an unbounded number of job types without losing a $\bigO{\log N}$ factor. That still leaves open the (im)possibility of constant factor approximations for the problem. Yet, our work has significantly reduced the complexity gap for stochastic parallel machine scheduling.

\paragraph{Acknowledgements.} The authors would like to thank Schloss Dagstuhl and the organizers of the 2025 workshop on ``Scheduling \& Fairness'', as well as the subsequent Dagstuhl workshop on ``Approximation Algorithms for Stochastic Optimization'', which helped finishing some of the final bits for this work.

\paragraph{Disclaimer.} No polylog($N$) factors have been oppressed during this work.

\bibliographystyle{abbrv}
\bibliography{bernoulli}

\pagebreak
\section*{Appendix}
\appendix

\section{From \texorpdfstring{\bigO{\log N}}{O(log N)} Many Types to Arbitrarily Many Types}\label{app:log-n-many-sizes}
Here we argue, analogous to \cite{GMZ2023}, that we can reduce the general problem without any restrictions on the size parameters, to the problem where we have only $\bigO{\log N}$ many size parameters, all powers of some constant $c\ge 2$, at the expense of losing only a constant factor in the performance guarantee. 
\begin{lemma}[Cp.\ Lemma 3.3 in \cite{GMZ2023}]\label{lem:log-n-many-sizes}
    Suppose we can compute, in (quasi-)polynomial time, a policy~$\Pi$ for scheduling $N$ Bernoulli jobs with $n=\bigO{\log N}$ different size parameters, all powers of some $c\ge 2$, and with performance guarantee $\alpha\ge 1$. Then we can compute, in (quasi-)polynomial time, a policy for scheduling $N$ Bernoulli jobs without any restrictions on the size parameters and with performance guarantee $(c^2+c\alpha)$.
\end{lemma}
The proof follows exactly the same lines as that of~\cite[Appendix C]{GMZ2023}; it is given here for the sake of completeness, and with only small adaptations.
\begin{proof}
    Given $N$ Bernoulli jobs\footnote{Recall that, with a slight abuse of notation, $N$ denotes the number of jobs as well as the set of all jobs.} without restrictions on the number of size parameters,  round all $p_J$ to the next largest power of $c$, say $p_J'$. Any policy for the original instance with sizes $p_J$ can be translated into a feasible policy for an instance with sizes $cp_J$, with expected cost at most $c$ times higher.  It follows that the optimal policy $\Pi'$ for the instance with rounded sizes $p_J'\le cp_J$ has expected cost at most $c$ times the expected cost of the optimal policy for the original instance. Now $\Pi'$ yields a feasible policy also for the original instance (by leaving machines idle), with expected cost at most $c$ times the expected cost of the optimal policy.  So losing at most a factor $c$, we may assume from now on that all $p_J$ are powers of $c\ge 2$.
   
    Next, for simplicity, for the given job set $N$ assume w.l.o.g.\ that all $p_J$ are scaled uniformly so that the expected cost of an optimal policy $\Pi^*(N)$ for the instance with rounded size parameters equals $\eopt:=\sum_{J\in N} \E{C^{\Pi^*(N)}_J}=1$.  Note that this also implies that $\E{X_J}\le 1$ for all jobs $J$.

    Now partition the set of jobs $N$ into three classes  $S\cup M \cup L$ so that $S$ are the small jobs with $p_J<\frac{1}{N^2}$, $M$ are the medium jobs with $\frac{1}{N^2}\le p_J < N^8$, and $L$ are the large jobs with $p_J\ge N^8$. Note that, due to the rounding to powers of $c\ge 2$, $M$ can only contain jobs with $\bigO{\log N}$ different size parameters $p_J$. The idea, somewhat standard in scheduling, is that only jobs $M$ really matter.

    We claim that the following policy $\Pi'$ has expected cost at most $2+\alpha$: First, greedily schedule the jobs in $L$ in arbitrary order. If any of these jobs turns out long, greedily schedule all remaining jobs, again in arbitrary order.  Otherwise, all $L$-jobs have $X_J=0$, then continue scheduling jobs in $S$ greedily, in arbitrary order. Note that, due to their small size, they all finish by time ${1}/{N}$. Then schedule jobs $M$ from time ${1}/{N}$ on, using given policy $\Pi$ for job set $M$.

    To see that the claim on the expected performance is true, distinguish the two cases that (i) all $L$-jobs have $X_j=0$ and (ii) that at least some $L$-job turns out long.
    \begin{enumerate}[label=(\roman*)]
        \item Conditioned on all $L$-jobs having size $X_J=0$, the expected cost of $\Pi'$ is at most 
            \[
                \sum_{J\in N}\E{C_J^{\Pi'}} \le N\cdot\frac{1}{N} + N\cdot\frac{1}{N} + \sum_{J\in M}\E{C_J^{\Pi(M)}} \le 2+ \alpha \sum_{J\in M}\E{C_J^{\Pi^*(M)}} \le 2+\alpha= (2+\alpha)\,\eopt\,.
            \] 
            Here, the first inequality follows by definition of policy $\Pi'$. The second inequality is the performance guarantee for the given policy $\Pi$. The third inequality is true because $\Pi^*(M)$ is defined to be an optimal policy to schedule jobs $M$, but since $M\subseteq N$, we have 
            $\sum_{J\in M} \E{C^{\Pi^*(M)}_J}$ $\le$ $\sum_{J\in M} \E{C^{\Pi^*(N)}_J}$ $\le$ $\sum_{J\in N} \E{C^{\Pi^*(N)}_J}= 1$. Here, note that indeed, $\Pi^*(N)$ is a well defined policy for job set $M$: when scheduling $N$, all jobs $S\cup L$ could have processing time $X_J=0$. 
        \item Conditioned on some $L$-job turning out long, observe that policy $\Pi'$ is 
            just some greedy list scheduling algorithm; this is the same as in \cite{GMZ2023}: In a nutshell, either only one job turns out long and all other jobs at least a factor $N^2$ shorter, in which case none of the other jobs finishes later than the long job, leading to a (conditional) bound $2\eopt$. If this is not the case, one conditions on the lengths of the longest and  second longest jobs. But taking into account that by scaling we have $\E{X_J}\le 1$, and by independence of the $X_J$, the probability of having two long jobs is very small by Markov's inequality, which outweighs the cost of the resulting schedule. Exactly as in the proof of~\cite[Appendix C]{GMZ2023}, one can therefore show that also for this case, one gets the (conditional) bound 
            \[
                \sum_{J\in N}\E{C_J^{\Pi'}}  \le c\,\eopt\,.
            \]
            We do not repeat the details of these arguments here; we only note that with a little care one readily sees that the arguments given in \cite[Appendix C]{GMZ2023} lead to the bound $c\,\eopt$.
    \end{enumerate}
    In summary, for both cases the policy $\Pi'$ has expected cost no more  than $(c+\alpha)\eopt$ (as $c\ge 2)$. Since $\eopt$ was defined to be the cost of the optimal solution after the rounding, and the rounding incurs an additional factor $c$, we arrive at the claimed performance guarantee $(c^2+c\alpha)$.
\end{proof}

\section{Omitted Proofs}\label{app:omittedproofs}

\nonidling*
\begin{proof}
    The proof is by contradiction.
    Consider a state at time $t$ where optimal policy $\Pi^*$ decides to leave machine $i$ idle, say for a period $[t,t'')$, while there are jobs yet unscheduled at $t$. (Also $t''=\infty$ is possible.)
    Now consider the first job, say job $J$, that policy $\Pi^*$ starts on any machine \emph{after} time $t$, say at time $t'>t$ on machine $i'$. (Also $i'=i$ is possible.)
    Note that $\Pi^*$'s decision to schedule $J$ at time $t'$ on machine $i'$ depends only on information that was already available at time $t$, because no other job was started in period $[t,t')$, and the processing times of all busy jobs at time $t$ are known, since they all follow a Bernoulli distribution. 
    We therefore can redefine $\Pi^*$ for the state at time $t$ by scheduling job $J$ on machine $i$, instead of leaving machine $i$ idle. 
    Job $J$ now finishes $t'-t>0$ time units earlier than before. 
    When $X_J=0$, we can simulate $\Pi^*$ after finishing $j$, by introducing the same amount of idle time on machine $i$ as the original $\Pi^*$ did.  
    When $X_J=p_J>0$, we can switch roles of machines $i$ and $i'$, meaning that machine $i'$ takes the jobs that were originally planned on machine $i$ after the idle time, and machine $i$ schedules the jobs that were originally scheduled on machine $i'$ after completion of job $j$.
    Since $t< t'\le t''$, this is indeed possible by leaving machine $i'$ idle in $[t',t'')$, and by introducing another idle time of length $t'-t$ after finishing job $J$ on machine $i$. 
    That yields that all jobs except $J$ are finishing at the same time as in $\Pi^*$.
    This modification improves policy $\Pi^*$'s expected objective value by an amount $t'-t>0$, contradicting optimality. 
\end{proof}

\lemtwophaseone*
\begin{proof}
    We claim that, for any pair of jobs $J,J'$, if $S_J\le S_{J'}$ then $S_{J'}^1-S_J^1 \ge S_{J'}-S_J$. Hence stochastic schedule $\schedule^1$ is feasible because $\schedule$ is. Moreover, it follows that $\policy^1$ with $\schedule(\policy^1)=\schedule^1$ can be defined by simulating $\policy$ because the claim in particular implies that the order of start times remains unaffected. 
    
    To show the claim, distinguish two cases:
    \begin{compactitem}
    \item $S_J< p_1$: Then $S_J^1=S_J$ and $S_{J'}^1\ge S_{J'}$ and the statement trivially follows.
    \item $S_J,S_{J'} \ge p_1$: Then $S_{J'}^1-S_J^1 = S_{J'}-S_J + 3(1+\eps)\eps p_1 - 3(1+\eps)\eps p_1 = S_{J'}-S_J$.
    \end{compactitem}
    It remains to argue about the increase in start time. Consider an arbitrary job $J$. If $S_J< p_1$, then $S_J^1= S_J < (1+\eps_1)S_J$ for any $\eps_1\ge 0$. Otherwise,
    \begin{align*}
        S_J^1 = S_J+3(1+\eps)\eps p_1 \le (1+3(1+\eps)\eps)S_J,
    \end{align*}
    using that $S_J\geq p_j$, and the statement holds for $\eps_1 = 3(1+\eps) \eps \in \bigO{\eps}$.
\end{proof}

\lemmatwophasefeasible*
\begin{proof}
    First note that no two jobs are processed during the same time on the same machine: This follows directly from Lemma~\ref{lem:2-phase1}, the fact that the non-intersecting type-$2$ spaces are inserted in the formerly idle interval $[C_i^+,S_i^+)$, and the fact that, at any time $t\ge S_i^+$ at which $\schedule^2$ starts some type-$2$ job, $\schedule^1$ starts a \emph{long} type-$2$ job. Finally note that all the jobs are scheduled because the total volume of type-$2$ spaces created is larger than the total volume of type-2 jobs, for each realization of $X_J$'s.
\end{proof}

\lemtwotypesphasetwo*
\begin{proof}
    Fix any machine $i$, and focus on that machine. Part~\ref{item:2-types-phase2:equal_volume_1} directly follows because $\schedule^1$ and $\schedule^2$ schedule type-$1$ jobs identically. Part~\ref{item:2-types-phase2:ahead} holds because by time $S_i^+$ the schedule $\schedule^2$, if it has not started all type-$2$ jobs and compared to $\schedule^1$, has processed an additional amount of type-$2$ volume of exactly 
    \begin{align*}
        p_2\Big(\lfloor \frac{\eps p_1}{p_2}\rfloor+1\Big) > \eps p_1
    \end{align*}
    on machine $i$, and, for any $t\geq S_i^+$ at which  $\schedule^1$ schedules a type-2 job on machine $i$, $\schedule^2$ also schedules a type-2 job on machine $i$, unless it has started all type-2 jobs already.
    Finally, for Part~\ref{item:2-types-phase2:idle} notice that $(\lfloor \frac{\eps p_1}{p_2}\rfloor + 1) p_2= \lfloor \frac{\eps p_1}{p_2}\rfloor p_2 + p_2 \le \eps p_1 +\eps^2 p_1 = (1+\eps)\eps p_1$. So by construction of $\schedule^2$, no job is started on machine $i$ at or after $C_i^+ + (1+\eps)\eps p_1 + (\lfloor\frac{\eps p_1}{p_2}\rfloor+1) p_2 \le C_i^+ + 2(1+\eps) \eps p_1$ and before~$S_i^+$.
\end{proof}

\lemmatwotwovolumephtwo*
\begin{proof}
    The interval $I_k$ has length at most $|I_k| \le \eps p_1$. Note that at most one job that starts within $I_k$ can extend beyond the right endpoint of $I_k$. 
    If this is an type-2 job, its processing time is at most $\eps^2 p_1<\eps p_1$, so the total volume of type-2 jobs that start on 
    machine $i$ in $I_k$ is strictly less than $(1+\eps)\eps p_1$. If it is a type-1 job  (note that any long 
    type-1 job that starts in $I_k$ is such a job), the total volume of type-2 jobs 
    that start on $i$ in $I_k$ is strictly less than $\eps p_1$.
\end{proof}

\lemtwotypesphasethreefeasibility*
\begin{proof}
    Consider jobs $J,J'$ with $S_J^2\le S_{J'}^2$. We argue that $S_{J'}^3-S_J^3\ge S_{J'}^2-S_J^2$, which implies feasibility using feasibility of $\schedule^2$. Indeed, let $S_J^2\in I_k$ and $S_{J'}^2\in I_\ell$. If $k=\ell$, then the statement follows directly since both jobs got delayed by the same amount.  
    If, on the other hand $k<\ell$, then $S_J^2=l_k+x$ for some $x\ge 0$ and $S_{J'}^2=l_\ell+y$ for some $y\ge 0$. Hence, $S_{J'}^3-S_J^3 = (1+5\eps)(l_\ell-l_k) + x - y \ge l_\ell  - l_k +x-y = S_{J'}^2-S_J^2$.
\end{proof}

\lemtwotypeslessintervals*
\begin{proof}
    By construction, if $S^2_J\in I_b$ for some index $b$, then $S^3_J\in I_b'$, and furthermore, if $S^2_J=l_b+x$ for $x<\eps p_1$, then $S^3_J=l_b'+x$. We have $C^2_J = l_b+x+p_1$ and $C^3_J = l'_b+x+p_1$. Therefore, 
    \begin{align*}
    k -b = \left\lceil \frac{x+p_1}{\eps p_1} \right\rceil,
    \text{ and }
    \ell -b = \left\lceil \frac{x+p_1}{(1+5\eps)\eps p_1} \right\rceil.
    \end{align*}
    To see that $\ell < k-1$, it therefore suffices to show that 
    \begin{align*}
     \frac{x+p_1}{\eps p_1} \ge \frac{x+p_1}{(1+5\eps)\eps p_1}+2,
    \end{align*}
    which can be verified to hold for $\eps \le {3}/{10}$. The idleness property follows simply by construction and feasibility of $\schedule^2$ (Lemma~\ref{lem:2-phase2-feasible}).
\end{proof}

\lemmmatwotypesphasethree*
\begin{proof}
    For parts~\ref{item:equal_volume_1-3} and \ref{item:ahead-3}, it suffices, by Lemma~\ref{lem:2-types-phase2} and the fact that $f^{i\to j}$ is the identity for $i,j\in\{1,2\}$, to show that at any time $t$
    \begin{itemize}
    \item the set of jobs started on machine $i$ until time $t$ under $\schedule^3$ is exactly the same as that started on $i$ until time $f^{3\rightarrow 2}(t)$ under $\schedule^2$, and
    \item for any job $J$ that has been started on $i$ with $S_J^3\le t$, it holds:
    \begin{align*}
        t-S_J^3 \le f^{3\rightarrow 2}(t) - S_J^2.
    \end{align*}
    \end{itemize}
    The first point directly follows because $f^{3\rightarrow 2}$ is monotone and, for any job $J$,
    \begin{align*}
        S_J^2 = f^{3\rightarrow 2}(S_J^3).
    \end{align*}
    The second point follows directly because $f^{3\rightarrow 2}$ is $1$-Lipschitz.
    
    For part~\ref{item:idle-3}, again by Lemma~\ref{lem:2-types-phase2}, we have that $\schedule^2$ is idle on machine~$i$ in intervals $[C_i^+,C_i^+ +(1+\eps)\eps p_1)$ as well as  $[C_i^+ +2(1+\eps)\eps p_1, S_i^{+})$. 
    Thus, under $\schedule^2$ no job is started in either  $I_k$ or $I_{\ell-1}$ and the statement carries over also to $\schedule^3$ and $I_k'$, $I_{\ell - 1}'$ by construction of $\schedule^3$.
\end{proof}

\lemmatwotypesphasefournocrossing*
\begin{proof}
    By Lemma~\ref{lem:2-phase3-vol2}, the total volume of type-2 jobs processed by $\schedule^4$ in $I_k'$ on machine $i$ is strictly less than $(1+\eps)\eps p_1$. The claim then follows because this volume is processed by $\schedule^4$ without leaving idle time and $|I_k'|\geq (1+\eps)\eps p_1$.
\end{proof}

\lemmatwotypesQjs*
\begin{proof}
    First note that, until time $p_{1}$, the two sets $Q_2$ and $Q_{1}$ coincide. 
    So if $t\le p_{1}$, the lemma holds. If $t> p_{1}$, then 
    $t\in \left\{l_k' + i\eps p_{2} \,\middle|\, (k,i\in\N) \wedge (l_k'\ge p_{1}) 
            \wedge ( l_k' + i\eps p_{2}<l_{k+1}'-p_2)\right\}$
    and therefore $t_1<t<t_2-p_2$ for two consecutive timepoints $t_1,t_2\in Q_{1}$, 
    which implies that no $t'\in Q_{1}$ exists such that $t<t'<t+p_2$.
\end{proof}

\lemmatwotypesgridaligned*
\begin{proof}
    For type-$1$ jobs, the lemma follows from the fact that $Q_1$ 
    contains all $l_k'$ and $\schedule^4$ starts type-$1$ jobs only at such points. 
    For a job $J$ that starts in $I_0'$, note that its start time $S_J^4$ is determined by the number of type-$2$ 
    jobs that precede it on the same machine that are long. Since we assume that $\frac1\eps$ is integer and $S_J^4$ is just a multiple of $p_2$, $S_J^4$ aligns to $Q_1$ and $Q_2$.
    For any type-$2$ job $J$ that starts in $I_k'$ for some $k\ge1$, its starting time 
    is again a multiple of $p_2$ plus the left endpoint $l_k'$ of $I_k'$. Since $\frac1\eps$ is integer, this aligns to $Q_2$.
\end{proof}

\lemmatwotypesstratified*
\begin{proof}
    We need to show that the three properties in Definition~\ref{def:well-formedness-2} are satisfied. 
    
    Property~\ref{item:2types-stratified:a}, follows from Lemma~\ref{lem:2types_grid-aligned}.
    
    Property~\ref{item:2types-stratified:b}, follows from Lemma~\ref{lem:2-types-ph5-starttimes} 
    for type-$1$ jobs. For any type-$2$ job $J$, note that by construction it must complete at a $Q_2$-timepoint and the property trivially holds. 
    
    Finally, property~\ref{item:ntypes-stratified:c} holds directly by construction: all 
    schedules considered during the transformation satisfy this property by design.
\end{proof}

\lemmadivisibilityforntypes*
\begin{proof}
    Assume the divisibility assumption holds  for $p_n',\dots, p_j'$, then it suffices to round $p_{j-1}$ up to $p_{j-1}'$ 
    by at most $p_j'\le p_j+p_{j+1}'\le \cdots \le p_j +\dots+p_n \le 
    \sum_{k=j}^{n}\eps^{2(k-j+1)} p_{j-1}= \eps^2 p_{j-1}\sum_{k=0}^{n-j}\eps^{2k}
    = \eps^2 p_{j-1}{(1-\eps^{2(n-j+1)})}/{(1-\eps^2)}
    \le \eps^2 p_{j-1} /{(1-\eps^ 2)}
    \le \eps p_{j-1}$ for $\eps \le 0.618$. Then observe that any policy for the original instance with sizes $p_j$ can be translated into a feasible 
    policy for an instance with sizes $(1+\eps)p_j$, with expected cost at most $(1+\eps)$ times higher.
    It follows that the optimal policy for the instance with rounded sizes $p_j'\le (1+\eps)p_j$ 
    has expected cost at most $(1+\eps)$ times the expected cost of the optimal policy for the original 
    instance. This in turn yields a feasible policy for the original instance (by leaving 
    machines idle), with expected cost at most $(1+\eps)$ times the expected cost of the optimal policy.
\end{proof}

\lemmantypesphonemain*
\begin{proof}
    We claim that, for any pair of jobs $J,J'$, if $S_J\le S_{J'}$ then $S_{J'}^1-S_J^1 \ge S_{J'}-S_J$. 
    Hence stochastic schedule $\schedule^1$ is feasible because $\schedule$ is. 
    Moreover, it follows that $\policy^1$ with $\schedule(\policy^1)=\schedule^1$ can be defined by 
    simulating $\policy$ because the claim in particular implies that the order of start times remains unaffected. 

    To show the claim, note from the definition that $S^1_J-S_J$ is non-decreasing in $S_J$.
    Therefore, for any pair of jobs~$J,J'$, if $S_{J'}\le S_J$, we have that $S_J^1-S_{J'}^1 \ge S_J-S_{J'}$.
    
    It remains to argue about the increase in start time. Consider an arbitrary job $J$. 
    Let $j \in\{1,\dots,n\}$ be such that 
    we have 
    \begin{align*}
        S_J^1&= S_J + (1+\eps)\eps \sum_{j'=j}^{n-1} (n-j'+3) p_{j'}\,.
    \intertext{Then, because the types are $\eps^2$-separated,}
        S_J^1 
        &\le S_J + (1+\eps)\eps \sum_{j'=j}^{n-1} (n-j'+3) \eps^{2(j'-j)}p_{j} \\
        &\le S_J + (n+2)(1+\eps)\eps \sum_{j'=0}^{n-1-j}  \eps^{2j'}p_{j} \\
        &\le S_J + (n+2)(1+\eps)\eps \frac{1-\eps^{2(n-j)}}{1-\eps^{2}}p_{j}\\ 
        &\le S_J + 2(n+2)(1+\eps)\eps p_{j} 
        \ \le\  (1+2(n+2)(1+\eps)\eps)S_J\,. 
    \end{align*}
    The last inequality holds because  $p_j\le S_J$, and  the second-to-last inequality holds for $\eps\le 1/\sqrt{2}$. 
\end{proof}

\lemmanphasetwofeasible*
\begin{proof}
    First note that no two jobs are processed during the same time on the same machine: This follows 
    from Lemma~\ref{lem:ntypes-ph1-main}, the fact that spaces are non-intersecting and are 
    either reserved in the formerly idle intervals or where $\schedule^1$ starts a long job of the 
    corresponding type. 
    Finally, note that all jobs are scheduled, because the total volume type-$j$ spaces created is larger 
    than the total type-$j$ volume.
\end{proof}

\lemmantypesphasetwo*
\begin{proof}
    Fix any machine $i$, and focus on that machine. 
    Parts~\ref{item:n-types-equal_volume_1-2} directly follows because if in $\schedule^1$ a long type-$j$ job is started at $t$, then by construction $\schedule^2$ either starts exactly the same job (if $j$ is on-par at $t$) or reserve a type-$j$ space starting at $t$ (if $j$ is ahead at $t$).  
    
    To prove Part~\ref{item:n-types-ahead-2}, we consider $t$ such that $\schedule^2$ has not 
    started all type-$j$ jobs by time $t$.  Note that, during any time interval before~$t$,  
    $\schedule^2$ processes at least as much type-$j$ volume as $\schedule^1$.
    Additionally, in an interval $[\Cplus{i}{k},\Splus{i}{k})$ before $t$ during which $j$ is an 
    ahead type, $\schedule^2$ processes an additional amount of type-$j$ volume of at least 
    \begin{align*}
        \left\lfloor \frac{\eps p_k}{p_j}+1\right\rfloor p_j > \eps p_k
    \end{align*}
    on machine $i$. 
    
    For Part~\ref{item:n-types-idle-2} notice that $\lfloor \frac{\eps p_j}{p_k}+1 \rfloor p_k\le (1+\eps)\eps p_j$ for all $k\in\{j+1,\dots,n\}$. 
    Hence, by construction of $\schedule^2$, no job starts on machine $i$ at or after $\Cplus{i}{j} + 2(1+\eps)\eps p_j + (n-j) \lfloor \frac{\eps p_j}{p_k} +1 \rfloor p_k \le \Cplus{i}{j} + (n-j+2)(1+\eps)\eps p_j \le \Splus{i}{j} - (1+\eps)\eps p_j$ but before $\Splus{i}{j}$.
    
    Finally, for Part~\ref{item:n-types-no-starts-2}, consider the job $J^\star$ that has the first completion time 
    at or after $p_j$ on machine $i$ in $\schedule$. I.e., job $J^\star$ determines $\Cstar{i}{j}$. 
    By definition, $\schedule$ starts $J^\star$ before $p_j$. Therefore, $\schedule^2$ starts $J^\star$ 
    before $p_j^*$. Furthermore, $\schedule$ starts the next job on $i$ not before $\Cstar{i}{j}$, and 
    therefore, $\schedule^1$ does not start the next job before $\Splus{i}{j}$. Since $\schedule^2$ only fills 
    the idle time between $\Cplus{i}{j} +2(1+\eps)\eps p_j$ and $\Splus{i}{j} - (1+\eps)\eps p_j$, the statement follows.
\end{proof}

\lemmaaheadvolumephtwo*
\begin{proof}
    The interval $I_k$ has length at most $|I_k| \le \eps p_j$. Note that at most one job that starts within $I_k$ can extend beyond the right endpoint of $I_k$. 
    If this is an ahead job, its processing time is at most $\eps^2 p_j<\eps p_j$, so the total volume of ahead jobs that start on 
    machine $i$ in $I_k$ is strictly less than $(1+\eps)\eps p_j$. If it is an on-par job  (note that any long 
    on-par job that starts in $I_k$ is such a job), the total volume of ahead jobs 
    that start on $i$ in $I_k$ is strictly less than $\eps p_j$.
\end{proof}

\lemmantypesphthreefeasibilitycost*
\begin{proof}
    Consider jobs $J,J'$ with $S_J^2\le S_{J'}^2$. We argue that $S_{J'}^3-S_J^3\ge S_{J'}^2-S_J^2$, which implies feasibility using feasibility of $\schedule^2$. Indeed, let $S_J^2\in I_k$ and $S_{J'}^2\in I_\ell$. If $k=\ell$, then the statement follows directly since both jobs got delayed by the same amount.  
    If, on the other hand $k<\ell$, then $S_J^2=l_k+x$ for some $x\ge 0$ and $S_{J'}^2=l_\ell+y$ for some $y\ge 0$. Hence, $S_{J'}^3-S_J^3 = (1+5\eps)(l_\ell-l_k) + x - y \ge l_\ell  - l_k +x-y = S_{J'}^2-S_J^2$.
\end{proof}

\lemmantypeslessintervals*
\begin{proof}
    If $C_J^3<f^{2\to 3}(p_j^*)$, since $J$ is a long type-$j$ job, $C_J\ge p_j$. 
    Therefore, $\schedule$ cannot start any job on $i$ at or after $C_J$ but before $p_j$.
    Thus, by construction, $\schedule^3$ does not start any jobs on $i$ in $[C_J^3,f^{2\to 3}(p_j^*))$. 
    Moreover, since by Lemma~\ref{lem:n-types-phase2}~\ref{item:n-types-idle-2} and~\ref{item:n-types-no-starts-2}, $\schedule^3$ also does not start any jobs on machine $i$ in the interval of 
    length $\eps p_j$ that starts at time $p_j^*$, the statement holds.
    
    If $C_J^2\ge p_{j-1}^*$, since $J$ is on-par, we know that $S_J^3<f^{2\to 3}(p_{j-1}^*)$, and,
    since $p_j\le \eps^2 p_{j-1}\le \eps p_{j-1}$, we know that $C_J^3$ falls in the first interval after 
    $f^{2\to 3}(p_{j-1}^*)$. By the same argument as above, $\policy^3$ does not start any jobs 
    in the first two intervals after $f^{2\to 3}(p_{j-1}^*)$, so the statement holds.
    
    If $f^{2\to 3}(p_j^*)\le C_J^3< f^{2\to 3}(p_{j-1}^*)$, let $S^2_J\in I_b$ for some index $b$, 
    which implies that $S^3_J\in I_b'$. 
    Let $I_k$ be the interval that contains $C^2_J-2\eps p_j$. Note that, 
    since $C_J^2< p_{j-1}^*$, $I_k$ is at least two intervals earlier 
    than the interval that contains $C^2_j$. Note that $l_k \ge C_J^2-3\eps p_j$ and in turn that $l_k - S_J^2 \ge (1-3\eps)p_j$. It therefore suffices to show that
    \[
        (1-3\eps)p_j(1+5\eps) \ge p_j,
    \]
    which holds for $\eps\le{2}/{15}$.
\end{proof}

\lemmantypesphasethree*
\begin{proof}
    For parts~\ref{item:n-types-equal_volume_1-3} and \ref{item:n-types-ahead-3}, it suffices, by Lemma~\ref{lem:n-types-phase2} and the fact that $f^{i\to j}$ is the identity for $i,j\in\{1,2\}$, to show at any time $t$
    \begin{itemize}
    \item that the set of jobs started on machine $i$ until time $t$ under $\schedule^3$ is exactly the same as that started on $i$ until time $f^{3\rightarrow 2}(t)$ under $\schedule^2$, and
    \item that, for any job $J$ that has been started on $i$ with $S_J^3\le t$, it holds:
    \begin{align*}
        t-S_J^3 \le f^{3\rightarrow 2}(t) - S_J^2.
    \end{align*}
    \end{itemize}
    The first point directly follows because $f^{3\rightarrow 2}$ is monotone and, for any job $J$,
    \begin{align*}
        S_J^2 = f^{3\rightarrow 2}(S_J^3).
    \end{align*}
    The second point directly follows because $f^{3\rightarrow 2}$ is $1$-Lipschitz.
    
    For part~\ref{item:n-types-idle-3}, again by Lemma~\ref{lem:n-types-phase2}, we have that $\schedule^2$ is idle on machine~$i$ in the interval $[S_i^+(k) -(1+\eps)\eps p_k, S_i^{+}(k))$. 
    Thus, under $\schedule^2$ no job is started in $I_{\ell-1}$ and the statement carries over also to $\schedule^3$ and $I_{\ell - 1}'$ by construction of $\schedule^3$.
\end{proof}

\lemmantypesQjs*
\begin{proof}
    First note that, until time $p_{j-1}^\circ$, the two sets $Q_j$ and $Q_{j-1}$ coincide. 
    So if $t\le p_{j-1}^\circ$, the lemma holds. If $t> p_{j-1}^\circ$, then 
    $t\in \left\{l_k' + i\eps p_{j} \,\middle|\, (k,i\in\N) \wedge (l_k'\ge p_{j-1}^\circ) 
            \wedge ( l_k' + i\eps p_{j}<l_{k+1}'-p_j)\right\}$
    and therefore $t_1<t<t_2-p_j$ for two consecutive timepoints $t_1,t_2\in Q_{j-1}$, 
    which implies that no $t'\in Q_{j-1}$ exists such that $t<t'<t+p_j$.
\end{proof}

\lemmantypegridaligned*
\begin{proof}
    For on-par jobs that start in $I_k'$, for $k\ge1$, the lemma follows from the fact that $Q_j\cap [0,p_{j-1}^\circ]$ 
    contains all $l_k'\le p_j^\circ$ and $\schedule^4$ starts on-par jobs only at such points. 
    For a job $J$ that starts in $I_0'$, note that its start time $S_J^4$ is determined by the 
    jobs that precede it on the same machine. Since by Lemma~\ref{lem:ntypes-divisible} 
    and the assumption that $1/\eps$ is integer, $p_j$ is a multiple 
    of $\eps p_n$, for all $j\in \{1,\dots,\numberjobs\}$, this aligns to the points in $Q_n\cap [0,p_n^\circ]$, 
    which is a subset of $Q_j$, for all $j\in\{1,\dots,\numberjobs\}$.
    For any ahead job $J$ of type $j$,  schedule $\schedule^4$ consecutively starts some number $\numberjobs'$ of  
    jobs $J_1,\dots,J_{ \numberjobs'}$, starting at some time point $l_k'$, $k\ge1$, followed by $J$. Since $\schedule^4$ orders the jobs 
    large to small, we have $p_{J_1}\ge \dots \ge p_{J_{ \numberjobs'}}\ge p_j$. 
    Now, since by Lemma~\ref{lem:ntypes-divisible}, $p_j$ divides $p_{j'}$ for all 
    $j'\le j$, the difference between the start time of $J$ and $l_k'$
    is divisible by $p_j$.
    Thus, since we assume $1/\eps$ is an integer, the difference is also 
    divisible by $\eps p_j$ and therefore the start time of $J$ is an element of $Q_j$.
\end{proof}

\lemmantypesstratified*
\begin{proof}
    We need to show that the three properties in  Definition~\ref{def:well-formedness} are satisfied. 
    
    Property~\ref{item:ntypes-stratified:a}, follows from Lemma~\ref{lem:ntypes_grid-aligned}.
    
    Property~\ref{item:ntypes-stratified:b}, follows from Lemma~\ref{lem:n-types-ph5-starttimes} 
    for on-par jobs. For a type-$j$ ahead job $J$, note that it follows from Corollary~\ref{cor:ntypes_aheadjobs_emptyspace} 
    that any type-$j'$ interval $I'_k$ is idle from $l_{k+1}'-3\eps^2p_{j'}$. 
    By Lemma~\ref{lem:ntypes-Qjs}, 
    either $J$ completes at time $C_J^4\in Q_j$ or $C_J^4>l_{k+1}'-p_j\ge l_{k+1}'-
    \eps^2p_{j'}>l_{k+1}'-3\eps^2p_{j'}$, a contradiction.
    
    Finally, property~\ref{item:ntypes-stratified:c} holds directly by construction: all 
    schedules considered during the transformation satisfy this property by design.
\end{proof}

\lemmadivisibilityforgroups*
\begin{proof}
    (i) follows from Lemma~\ref{lem:ntypes-divisible}, since $p_{G_1},\dots,p_{G_\numbergroups}$ 
    are $\eps^2$-separated. (ii) follows, from the fact that $p_j\ge p_{G_h}$ and that it suffices 
    to round up $p_j$ by at most $\eps p_{G_h}$ to $p_j'\le p_j+\eps p_{G_h}\le p_j+\eps p_j\le (1+\eps)p_j$.
\end{proof}

\lemmadpnonidling*
\begin{proof}
    Suppose, at time $t\in Q$, a machine is available and there is an unprocessed job of type $j \in\timedtypes{t}$, but no job is started.
    Let $i$ be an available machine at time $t$ and consider $J$, the next job that is started on any machine, 
    let $j'$ be its type, $i'$ the machine it is processed on, and $t'$ the time at which it starts processing. Note that policy $\policy'$  did not 
    gain any information between time $t$ and $t'$, since no jobs started processing. 
    Moreover, machine $i$ is idle between $t$ and $t'$.
    If $j'\in \timedtypes{t}$, we can start processing job $J$ (of type $j'$) on machine $i$ at time $t$, and by switching the roles of machines $i$ and $i'$, and inserting additional idle time where necessary as in the proof of Lemma~\ref{lem:non-idling}, we have constructed another stratified policy with respect to $(Q_{G_1},\dots,Q_{G_\numbergroups})$ and $(p^\circ_{G_1},\dots,p^\circ_{G_\numbergroups})$ with smaller expected cost, contradicting optimality of $\policy'$.

    If $j'\notin \timedtypes{t}$, then by Lemma~\ref{lem:ntypes-Qjs} (Lemma~\ref{lem:groups-Qjs})
    we have that $t+p_j\le t'$. That is, starting a job of 
    type~$j$ at time $t$ does not interfere with the start of job $J$ at time $t'$, also not if  $i=i'$. 
    Thus we can schedule a job of type $j$ on machine $i$ at time $t$, everything else kept the same, 
    which again would decrease the expected cost of $\policy'$.
\end{proof}

\lemmarelevantjobsleft*
\begin{proof}
    For each type $j$ there are $\numberjobs_j$ possible values for $\jobsleft_j$. 
    Therefore, there are $\prod_j \numberjobs_j \in \bigO{\numberjobs^\numtypes}$ possible values for $\vjobsleft$.
\end{proof}

For use in the proofs of the next two lemmas, for $h\in\{1,\dots,\numbergroups-1\}$, let $Q'_{G_h}=\{t\in Q_{G_h}\mid t\ge p_{G_h}^\circ\}$, 
and let $Q'_{G_\numbergroups}={Q}_{G_\numbergroups}$. Note 
that $\bigcup_{h=1}^\numbergroups Q'_{G_h} = Q$ and for any $t,t'\in Q'_{G_h}$, $t\neq t'$, 
we have $| t-t'|\ge \frac12\eps p_{G_h}$. Furthermore, let $\pmax{G_h}= \max_{j\in G_h}p_j$ 
for all $h\in\{1,\dots,\numbergroups\}$.

\lemmarelevanttimepoints*
\begin{proof}     
    First, we show that, in any realization of job processing times,
    an optimal policy does not start any job after the $\numberjobs{\eps^{-2|G_1|-3}}$-rd
    time point in $Q'_{G_1}$. 
    To that end, note that all jobs can be started at time points in $Q'_{G_1}$,
    since $Q_{G_1}\subseteq\dots\subseteq Q_{G_\numberjobs}$. 
    The number of time points from $Q'_{G_1}$ ``crossed'' by any one job is bounded from above by 
    \begin{equation}
        \left\lceil\frac{\pmax{G_1}}{\frac\eps2 p_{G_1}}\right\rceil \le 
        \left\lceil\frac{2\pmax{G_1}}{\eps\, \eps^{2(|G_1|-1)}\pmax{G_1}}\right\rceil = 
        \left\lceil2\eps^{-2|G_1|+1}\right\rceil \le 2\eps^{-2|G_1|+1}+1 
        \label{eq:timepointbound}
    \end{equation}
    Here, the first inequality follows from \eqref{eq:p_j_across_groups}. 
    After any job's completion time an additional idle time may occur in accordance with 
    the second property of stratified policies.
    This idle time is until at most the next time point in $Q'_{G_1}$.
    This adds another crossed time point for a total of  $2\eps^{-2|G_1|+1}+2\le \eps^{-2|G_1|}$, where the latter inequality holds for $\eps\le 1/3$.
    Therefore, even if we schedule \emph{all} jobs on a single machine, we need to consider 
    at most $\numberjobs{\eps^{-2|G_1|}}$ time points in $Q'_{G_1}$.
    Thus, the DP does not need to start any job after the $\numberjobs{\eps^{-2|G_1|}}$-th time point in $Q'_{G_1}$, and there are at most $\numberjobs{\eps^{-2|G_1|}}$ relevant $Q'_{G_1}$ time points. 
    We proceed to argue by induction on the type groups.
    
    Our induction hypothesis (IH) is: There are at most $\prod_{i=1}^{\ell} \numberjobs{\eps^{-2|G_i|}}$ 
    relevant $Q'_{G_\ell}$ time points.
    Assume that IH holds. In between two consecutive time points in $Q'_{G_\ell}$, 
    only jobs in $\cup_{i=\ell+1}^\numberjobs G_i$ can be started at time points in $Q'_{G_{\ell+1}}$. 
    Jobs of the largest type in $\cup_{i=\ell+1}^{\numbergroups}G_{i}$ cross at most 
    ${\eps^{-2|G_{\ell+1}|}}$ such time points, including idle time after the completion 
    of a job. 
    Therefore, even if we schedule 
    all those jobs on a single machine between two consecutive time points in 
    $Q'_{G_\ell}$, we  need 
    at most $\numberjobs{\eps^{-2|G_{\ell+1}|}}$ time points in $Q'_{G_{\ell+1}}$. 
    In other words, there are at most $\numberjobs{\eps^{-2|G_{\ell+1}|}}$ relevant $Q'_{G_{\ell+1}}$ time points between 
    two consecutive time points in $Q'_{G_\ell}$. 
    Similarly, there are at most $\numberjobs{\eps^{-2|G_{\ell+1}|}}$ relevant 
    $Q'_{G_{\ell+1}}$ time points before the first time point in $Q'_{G_\ell}$.
    Therefore, there are at most $\prod_{i=1}^{\ell+1} \numberjobs{\eps^{-2|G_i|}}$ 
    relevant $Q'_{G_{\ell+1}}$ time points in total. 
    Finally, recall that $Q=Q'_{G_\numbergroups}=Q_{G_\numbergroups}$, and note that $\prod_{i=1}^{\numbergroups} \numberjobs{\eps^{-2|G_i|}} \le ({\numberjobs^\numbergroups}{\eps^{-2\numtypes}} )\leq \numberjobs^\numtypes\eps^{-2\numtypes}$.
\end{proof}

\lemmarelevantprofiles*
\begin{proof}
    First we partition the entries of vector $\vprofile=(m_1,\dots, m_m)$. 
    Note that, 
    by construction (Definition~\ref{def:profileupdate}), each value $\profile_i$ is 
    an element of $Q$. We partition the values $\profile_i$ by type group $G_h$ as follows: Define subprofile $\vprofile_h$ of $\vprofile$ by taking only values $m_i$ with
    $\profile_i\in Q'_{G_h}$ and $\profile_i\notin Q'_{G_{h-1}}$.
    We argue by separately counting subprofiles $\vprofile_h$, 
    for all $h\in\{1,\dots,\numbergroups\}$. 
    Note that $\vprofile_1, \dots, \vprofile_\gamma$ uniquely determine $\vprofile$,
    since all 
    values are lexicographically ordered.
    
    Consider a value $\profile_i$ of $\vprofile_h$, which is the result of an updated machine load profile following Definition~\ref{def:profileupdate} at some time $t\le t^*$. Therefore, 
    \[
        \profile_i \le t^*+\pmax{G_h}+t^\text{IDLE}\,,
    \]
    where $t^\text{IDLE}$
    denotes the time between $t^*+\pmax{G_h}$ and the next time point in $Q'_{G_h}$. 
    We can view the subprofile $\vprofile_h$ as a \emph{staircase} with steps of unit height and the end point of any step a time point from $Q'_{G_h}$. 
    We know that $[t^*+\pmax{G_h},t^*+\pmax{G_h}+t^\text{IDLE}]$ contains exactly one time point 
    from $Q'_{G_h}$, while $[t^*,t^*+\pmax{G_h}]$ contains at most 
    \[ 
        \left\lceil \pmax{G_h}/(\tfrac\eps2 p(G_h))\right\rceil
         \le 2\eps^{-2|G_h|+1}+1
    \]
    time points from $Q'_{G_h}$, by \eqref{eq:timepointbound}. Therefore, the total 
    number of time points in $Q'_{G_h}\cap[t^*,t^*+\pmax{G_h}+t^\text{IDLE}]$ 
    is bounded by ${2\eps^{-2|G_h|+1}}+2\le \eps^{-2|G_h|}$ (for $\eps \le 1/3$).
    With $k:=\eps^{-2|G_h|}$, denote by $x_i$ the number of machines from $\vprofile_h$ that have the same load, which is one of the at most $k$ different time points. Then $x_1+\cdots+x_{k}\le m$. Recall that the number of different solutions to the equation $x_1+\cdots+x_{k}=m$ with integers $x_i\ge 0$ equals $\binom{k+m-1}{k-1}$. Therefore
    we can bound the number of different subprofiles $\vprofile_h$ by 
    \[
        \binom{\eps^{-2|G_h|}+m}{\eps^{-2|G_h|}} \le \binom{\eps^{-2\maxgroupsize}+m}{\eps^{-2\maxgroupsize}} \in \BigO{m^{\eps^{-2\maxgroupsize}}}\,.
    \]
    For all groups, $h=1,\dots,\gamma$, together, this implies the lemma.
\end{proof}

\theoremoptimalstratified*
\begin{proof}
    By Lemma~\ref{lem:dp-nonidling}, the dynamic program described by \eqref{eq:DPrecursion} computes an optimal stratified policy with respect to $(Q_1,\dots,Q_\numtypes)$ and $(p^\circ_1,\dots,p^\circ_\numtypes)$ according to Definition~\ref{def:group_well-formedness}. 
    The claim on the computation time follows from Lemma~\ref{lem:countDPCells} when $n\in\bigO{1}$. This is because the recursion of the dynamic program itself in \eqref{eq:DPrecursion} takes polynomial time. To see why, observe that there are \bigO{n} job types to possibly start next, and moreover, the computation of $\vprofile^j$ and $\vprofile^0$ from $\vprofile$ is polynomial, too. Clearly, to avoid generating (exponentially many) duplicates of identical sub-trees in the recursion tree, one needs to maintain a corresponding lookup table for previously computed profiles $(\vprofile,\vjobsleft)$. The size of this lookup table however is polynomially bounded by Lemma~\ref{lem:countDPCells}.
\end{proof}

\section{Adapted Claims for Grouped Size Parameters}\label{sec:app_groups}

\subsection{Phase 1}
The first modification is necessary in Phase~1 when generating enough idle time for subsequent phases, 
and in the proof of Lemma~\ref{lem:ntypes-ph1-main}
which reconfirms the existence of a corresponding scheduling policy $\Pi^1$ that leaves enough idle time. 
With job sizes separated into groups $G_1,\dots, G_\numbergroups$, define

\noindent\fbox{%
    \begin{minipage}[t][][t]{0.98\textwidth}\vspace{0pt}
        \textbf{Stochastic schedule $\schedule^1$:} 
        Start job $J$ at time 
        \[
            S_J^1 = \begin{cases}
                        S_J & \text{if $S_J<p_{G_{\numbergroups-1}}$}\\
                        S_J + (1+\eps)\eps\sum\limits_{k=h}^{\numbergroups-1} 
                        \left(\sum\limits_{i=k+1}^{\numbergroups} |G_i|+3\right) p_{G_{k}} & 
                        \text{if $h\in\{2,\dots,\numbergroups-1\}$, s.t.\ $p_{G_{h}}\le S_J< p_{G_{h-1}}$}\\
                        S_J + (1+\eps)\eps\sum\limits_{k=1}^{\numbergroups-1} 
                        \left(\sum\limits_{i=k+1}^{\numbergroups} |G_i|+3\right) p_{G_{k}} & \text{if $p_{G_1}\le S_J$}
                    \end{cases}
        \]
        on the same machine as in $\schedule$.
    \end{minipage}
}
\medskip

Analogous to the definitions of $\Splus{i}{j}$, $\Cplus{i}{j}$, and $\Cstar{i}{j}$, for all types $j$,
redefine $\Splus{i}{h}$, $\Cplus{i}{h}$, and $\Cstar{i}{h}$, for all groups $h$:
\[
    \Cplus{i}{h} = \Cstar{i}{h} +(1+\eps)\eps\sum\limits_{k=h}^{\numbergroups-1} 
                    \left(\sum\limits_{i=k+1}^{\numbergroups} |G_i|+3\right) p_{G_{k}}\,,
\] 
and
\begin{align*}
    \Splus{i}{h}:= \begin{cases}
                    \Cplus{i}{h'} +  (1+\eps)\eps\left(\sum\limits_{i=k+1}^{\numbergroups} |G_i|+3\right) p_{G_{h'}} & 
                    \begin{tabular}{l}
                        \text{if $h'\in\{2,\dots,\numbergroups-1\}$,} \\
                        \text{s.t.\ $p_{G_{h'}}\le \Cplus{i}{h}< p_{G_{h'-1}}$}
                    \end{tabular}\\
                    \Cplus{i}{1} +  (1+\eps)\eps\left(\sum\limits_{i=2}^{\numbergroups} |G_i|+3\right) p_{G_{1}} & \text{if $p_{G_1}\le \Cplus{i}{h}$}
                \end{cases}\,.  
\end{align*}
Where, $\Cstar{i}{h}$ is the first completion time at or after $p_{G_h}$ on machine $i$ in $\schedule$.
If on some machine $i$ no job ends after time $p_{G_h}$, let $\Cstar{i}{h}=\infty$.
The proofs for the two subsequent lemmas of Phase 1 are analogous to the ones in Section~\ref{sec:multi-types}. As a matter of fact, 
the bound on the start times of the jobs in $\schedule^1$ remains exactly the same as in the case without 
type groups.

\begin{lemma}[Corresponding to Lemma~\ref{lem:ntypes-ph1-main}]
    Stochastic schedule $\schedule^1$ is feasible, and there exists a policy $\Pi^1$ 
    such that $\schedule(\Pi^1)=\schedule^1$, and for every job $J$, $S_J^1\le (1+ (2n+4)(1+\eps)\eps)S_J$.
\end{lemma}

\subsection{Phase 2}
For Phase 2, we first define for all $h\in\{1,\dots,\numbergroups\}$
\[
    p_{G_h}^* = p_{G_h} + (1+\eps)\eps\sum\limits_{k=h}^{\numbergroups-1} 
                    \left(\sum\limits_{i=k+1}^{\numbergroups} |G_i|+3\right) p_{G_{k}} \,,
\]
and
\begin{definition}[Ahead and on-par types]
    In $\schedule^i$, at time $t$, we say that type $j\in \{1,\dots,n\}$ from group~$G_h$ is 
    an \emph{ahead type} if $f^{i\to1}(t)\ge p_{G_h}^*$, otherwise it is an \emph{on-par type}. 
    Additionally, we say that a job $J$ of type $j$ is an \emph{ahead job} if 
    $j$ is an ahead type at time $S^i_J$, the start time of $J$, and, similarly,
    that it is an \emph{on-par job} if $j$ is on-par at time $S^i_J$.
\end{definition}

The functions $f^{i\to j}$ are defined analogously to how they are defined in Section~\ref{sec:multi-types}.

Now, the definition of Stochastic schedule $\schedule^2$ is exactly the same as in Section~\ref{subsec:ntypes_phase2}
with (slightly) adapted number of reserved spaces and interval lengths.

\begin{lemma}[Corresponding to Lemma~\ref{lem:n-phase2-feasible}]
    Stochastic schedule $\schedule^2$ is feasible.
\end{lemma}

\begin{lemma}[Corresponding to Lemma~\ref{lem:n-types-phase2}]
    The following hold:
    \begin{enumerate}[label=(\roman*)]
       \item  For any time $t\ge 0$, the number of on-par jobs of type $j$ that $\schedule^2$ has started on machine $i$ at time $t$ is equal to the number of on-par jobs of type $j$ that $\schedule^1$ has started on machine $i$ by time $t$.
        \item  For any $t\ge \Splus{i}{h}$ and type $j$ that is {ahead} at $p^*_{G_h}$:  
        if $\schedule^2$ has not started all type-$j$ jobs by $t$, then $V_j(t,i,\schedule^2) > V_j(t,i,\schedule^1) + \eps p{G_h}$. 
        \item  In $\schedule^2$, machine $i$ is idle throughout the intervals $[\Cplus{i}{h}, \Cplus{i}{h} + 2(1+\eps)\eps p_{G_h})$ and $[\Splus{i}{h} - (1+\eps)\eps p_{G_h}, \Splus{i}{h})$, for all $h\in\{1,\dots,\numbergroups-1\}$.
        \item In $\schedule^2$, machine $i$ does not start any jobs in $[p_{G_h}^*,\Cplus{i}{h})$, for all $h\in\{1,\dots,\numbergroups-1\}$.
    \end{enumerate}
\end{lemma}

\subsection{Phase 3}
We define time points $l_0,\dots$ and the corresponding intervals $I_0,\dots$ analogously 
to how we do so in Section~\ref{subsec:n-types-Phase3}. 
Obviously, instead of associating intervals to types, we now associate them to groups. 
We say that an interval $I_k$ is a group-$h$ interval if $I_k\subseteq [p_{G_h}^*,p_{G_{h-1}}^*)$. 
Note that a group-$h$ interval has length at least $\frac\eps2 p_{G_h}$ and at most $\eps p_{G_h}$.

\begin{lemma}[Corresponding to Lemma~\ref{lem:ahead-volume-ph3}]
    The total volume of ahead jobs started on a machine $i$ under $\schedule^2$ in a 
    group-$h$ interval $I_k$ is strictly less $(1+\eps)\eps p_{G_h}$ and if $\schedule^2$ 
    starts a long on-par job on $i$ within $I_k$ then it is strictly less than $\eps p_{G_h}$.
\end{lemma}

The remainder of Phase 3 is identical to Phase 3 without groups, i.e., ignoring the groups completely.

\subsection{Phase 4}
Phase 4 with groups is identical to Phase 4 without groups, i.e., ignoring the groups completely, 
up to the definition of the sets of time points $Q_{G_{1}},\dots,Q_{G_\numbergroups}$ and the points $p^\circ_{G_{1}},\dots,p^\circ_{G_\numbergroups}$, as described in Definition~\ref{def:Q_groups}, and the following lemmas. 

\begin{lemma}[Corresponding to Lemma~\ref{lem:ntypes-Qjs}]\label{lem:groups-Qjs}
    Let $t\in Q_{G_h}$, if there is a $t'\in Q_{G_{h-1}}$ such that $t<t'<t+\pmax{G_h}$, then $t\in Q_{G_{h-1}}$.
\end{lemma}
\begin{proof}
    First note that, until time $p_{G_{h-1}}^\circ$, the two sets $Q_{G_{h}}$ and $Q_{G_{h-1}}$ coincide. 
    So if $t\le p_{G_{h-1}}^\circ$, the lemma holds. If $t> p_{G_{h-1}}^\circ$, then 
    \[
        t\in \left\{l_k' + i\eps p_{G_h} \,\middle|\, (k,i\in\N) \wedge (l_k'\ge p_{G_{h-1}}^\circ) 
            \wedge ( l_k' + i\eps p_{G_h}<l_{k+1}'-\pmax{G_h})\right\}
    \]
    and therefore $t_1<t<t_2-\pmax{G_h}$ for two consecutive timepoints $t_1,t_2\in Q_{G_{h-1}}$, 
    which implies that no $t'\in Q_{G_{h-1}}$ exists such that $t<t'<t+\pmax{G_h}$.
\end{proof}

\begin{lemma}[Corresponding to Lemma~\ref{lem:ntypes_grid-aligned}]\label{lem:groups_grid-aligned}
    In $\schedule^4$, any job of type $j$ starts at a timepoint in $Q_{G(j)}$.
\end{lemma}
\begin{proof}
    For on-par jobs that start in $I_k'$, for $k\ge1$,, the lemma follows the fact that $Q_{G(j)}\cap [0,p_j^\circ]$ 
    contains all $l_k'\le p_{G(g)}^\circ$ and $\schedule^4$ starts on-par jobs only at such points.
    For a job $J$ that starts in $I_0'$, note that its start time $S_J^4$ is determined by the 
    jobs that precede it on the same machine that are long. Since by Lemma~\ref{lem:groups-divisible} 
    and the assumption that $\frac1\eps$ is integer, $p_j$ is a multiple 
    of $\eps p_{G_\numbergroups}$, for all $j\in \{1,\dots,\numberjobs\}$, this aligns to the points in $Q_{G_\numbergroups}\cap [0,p_n^\circ]$, 
    which is a subset of $Q_{G_h}$, for all $h\in\{1,\dots,\numbergroups\}$.
    For any ahead job $J$ of type $j\in G_h$, note that from some timepoint $l_k'\in Q_{G_{h'}}$, 
    for some $h'<h$, schedule $\schedule^4$ consecutively starts some number $\numberjobs'$ of long 
    jobs $J_1,\dots,J_{ \numberjobs'}$, followed by $J$. Since $\schedule^4$ orders the jobs 
    large to small, we have $p_{J_1}\ge \dots \ge p_{J_{ \numberjobs'}}\ge p_{G_h}$. 
    Now, since by Lemma~\ref{lem:groups-divisible} $\eps p_{G_h}$ divides $p_{\bar{j}}$ for all 
    $\bar{j}\le j$, the difference between the start time of $J$ and $l_k'$
    is divisible by $\eps p_{G_h}$.
    Thus, by construction, the start time of $J$ an element of $Q_{G_h}$.
\end{proof}

\begin{lemma}[Corresponding to Lemma~\ref{lem:n-types-stratified}]
    $\schedule^4$ is stratified with respect to $(Q_{G_1},\dots, Q_{G_\numbergroups})$ and 
    $(p_{G_1}^\circ,\dots, p_{G_\numbergroups}^\circ)$.
\end{lemma}
\begin{proof}
    The proof is identical to the proof of Lemma~\ref{lem:n-types-stratified}.
\end{proof}

\end{document}